%% file: arxiv.tex
\documentclass[11pt,a4paper]{article}
\usepackage[margin=1in]{geometry}

\usepackage[font=bera]{chihao}

\usepackage{tocloft}
\usepackage{booktabs}
\usepackage{threeparttable}
\usepackage{multirow}
\usepackage{array}
\usepackage{makecell}

\usepackage{catchfilebetweentags}

\newcommand{\DTV}{\!{TV}}
\newcommand{\vol}{\!{vol}}

\title{On the query complexity of sampling from non-log-concave distributions\thanks{Accepted for presentation at the
Conference on Learning Theory (COLT) 2025.}}

\author{Yuchen He\\Shanghai Jiao Tong University \\ \textsf{yuchen\_he@sjtu.edu.cn} \and Chihao Zhang \\Shanghai Jiao Tong University \\ \textsf{chihao@sjtu.edu.cn}}

\begin{document}
\maketitle

\begin{abstract}
    We study the problem of sampling from a $d$-dimensional distribution with density $p(x)\propto e^{-f(x)}$, which does not necessarily satisfy good isoperimetric conditions.

    Specifically, we show that for any $L,M$ satisfying $LM\ge d\ge 5$, $\eps\in \tp{0,\frac{1}{200}}$, and any algorithm with query accesses to the value of $f(x)$ and $\grad f(x)$, there exists an $L$-log-smooth distribution with second moment at most $M$ such that the algorithm requires $\tp{\frac{LM}{d\eps}}^{\Omega(d)}$ queries to compute a sample whose distribution is within $\eps$ in total variation distance to the target distribution. We complement the lower bound with an algorithm requiring $\tp{\frac{LM}{d\eps}}^{\+O(d)}$ queries, thereby characterizing the tight (up to the constant in the exponent) query complexity for sampling from the family of non-log-concave distributions.

    Our results are in sharp contrast with the recent work of Huang et al. (COLT'24), where an algorithm with quasi-polynomial query complexity was proposed for sampling from a non-log-concave distribution when $M=\!{poly}(d)$. Their algorithm works under the stronger condition that all distributions along the trajectory of the Ornstein-Uhlenbeck process, starting from the target distribution, are $\+O(1)$-log-smooth. We investigate this condition and prove that it is strictly stronger than requiring the target distribution to be $\+O(1)$-log-smooth. Additionally, we study this condition in the context of mixtures of Gaussians.

    Finally, we place our results within the broader theme of ``sampling versus optimization'', as studied in Ma et al. (PNAS'19). We show that for a wide range of parameters, sampling is strictly easier than optimization by a super-exponential factor in the dimension $d$.

\end{abstract}

\newpage

\setcounter{tocdepth}{3}
\tableofcontents

%\ctodo{make sure the TOC does not exceed one page.}

\newpage

\section{Introduction}
The problem of sampling from a given distribution has gained much attention in recent years due to its  wide applications in many fields such as machine learning, physics, finance and statistics (e.g. \cite{ADDJ03,K06,D07,LB21}). Given query access to the value and gradients of a potential function $f:\bb R^d\to \bb R$, the task is to generate a sample from the distribution $\mu$ with density $p_\mu\propto e^{-f}$ on $\bb R^d$ within as few queries as possible.

Many algorithms have been developed to address this problem. The Langevin algorithm and its variants are the most commonly used ones (e.g. \cite{RT96,CB18,CCBJ18,Wib19,CLA+21,LPW24}). 
% This family includes unadjusted Langevin algorithms (\cite{RT96,CB18,DM19,VW19,CEL+24,MFH+23}), underdamped Langevin algorithms (\cite{CCBJ18,SL19,ZCL+23,LPW24}), Metropolis-adjusted Langevin algorithms (\cite{RS02,BH13,DCWY19,CLA+21}), proximal Langevin algorithms (\cite{Ber18,Wib19,LL22}) and so on. 
Polynomial-time convergence of Langevin-based algorithms is guaranteed when the target distribution $\mu$ exhibits good properties such as isoperimetric properties (e.g. \cite{VW19,CEL+24,MFH+23,ZCL+23}) or log-concavity (e.g. \cite{SL19,DCWY19,AC24}).
% However, for general non-log concave distributions, these algorithms may not be effective.

However, many interesting sampling problems are not log-concave and do not satisfy good isoperimetric properties. In recent years, there has been growing interest in sampling from general non-log-concave distributions. In this work, as well as in most of the literature on non-log-concave sampling, the distributions are required to satisfy only the following two weak assumptions.

\begin{assumption}\label{assump:moment}
    The second moment of $\mu$ is bounded, i.e. $\E[X\sim \mu]{\|X\|^2}\leq M$ for some $M<\infty$.
\end{assumption}

\begin{assumption}\label{assump:smooth}
    The potential function $f$ is differentiable and $\grad f$ is $L$-Lipschitz, i.e, for any $x,y\in \bb R^d$, $\|\grad f(x)-\grad f(y)\| \leq L \|x-y\|$.
\end{assumption}

Note that \Cref{assump:smooth} is equivalent to requiring $\mu$ to be $L$-log-smooth. There are mainly two lines of work regarding non-log-concave sampling. 

The first line aims at understanding the astonishing performance of the  score-based generative models (SGMs) in practice. It is proved that assuming the score functions can be accurately estimated, an accurate sampler with polynomially many queries exists (e.g. \cite{CCL+23,CCL+23b,CLL23,BBDD24}). %\ctodo{There might be many other refs.}

Another line of work is to identify the inherent sample complexity under the value / gradient oracle model. For special distributions such as the mixture of Gaussians of some particular shapes, polynomially many queries are sufficient to get an accurate sample (\cite{GLR18,GTC24}). On the other hand, a lower bound of $e^{\Omega(d)}$ for general non-log-concave sampling problems is established in \cite{GLR18}.

%In this paper, we require the target distributions to satisfy only the following two very weak assumptions, which are commonly used in most works.
% In most studies, the target distributions are only required to satisfy the following two very weak assumptions.

% \htodo{mention that $LM\geq d$.}
%Yet, due to the challenging nature of this problem, there is still no clear understanding of it. 
To the best of our knowledge, the query complexity for sampling in terms of the parameters $L$ and $M$ is not known yet. The work of~\cite{HRT24} derived a bound of $\exp\tp{\+O\tp{d\log \frac{L(d+M)}{\eps}} + Z}$, where $Z$, the maximum norm of the particles appeared in their algorithm, does not have an explicit bound. The recent work~\cite{HZD+24} designed a quasi-polynomial algorithm under a stronger assumption. Their algorithm further requires that \emph{the distributions during the Ornstein-Uhlenbeck process (OU process) starting from the target distribution $\mu$ to be $L$-log-smooth} and has query complexity roughly $\exp\tp{\+O(L^3)\cdot \!{polylog}(Ld+M)}$. Therefore, if all the distributions along the trajectory of the OU process is $\+O(1)$-log-smooth and $M=\!{polylog}(d)$, their algorithm is quasi-polynomial. Breaking the exponential barrier turns out to be significant in both theory and applications. However, as far as we can see, the relationship between this ``smoothness along the trajectory'' assumption used in their work and \Cref{assump:smooth} remains unclear, which limits the applicability of their results.

In this work, we make progress in understanding the query complexity for sampling under \Cref{assump:moment} and \Cref{assump:smooth}. Additionally, we further investigate the ``smoothness along the trajectory'' property mentioned above and, finally, compare our results with the query complexity for general non-convex optimization.

\subsection{Main results}
% Our main results can be summarized in the following four aspects.
We summarize our main results in the following four aspects.

\paragraph{Query complexity lower bounds for sampling}
% We prove a general lower bound with respect to $L,M,d$ and $\eps$. If an algorithm can sample from any distribution that satisfies \Cref{assump:moment}~and~\ref{assump:smooth} with a total variation distance error of no more than $\eps$, it will inevitably require $e^{\frac{d}{2}\log \Omega\tp{\frac{LM}{d\eps}}}$ queries (see \Cref{thm:}). 

First, we provide a general query complexity lower bound. We prove that to guarantee an error of $\eps$ in total variation distance, $\tp{\frac{LM}{d\eps}}^{\Omega(d)}$ queries is inevitable. See \Cref{thm:main-lb} below for the formal statement. 

\begin{theorem}\label{thm:main-lb}
    Let $\eps\in \tp{0,\frac{1}{200}}$. For any $L,M>0$ such that $LM\geq d$ and for any $d\geq 5$, if a sampling algorithm $\+A$ always terminates within $K$ queries on any target distribution $\mu$ under \Cref{assump:moment}~and~\ref{assump:smooth}, and guarantees that the distribution of $\+A$'s output, denoted as $\tilde \mu$, satisfies $\DTV(\tilde \mu, \mu)\leq \eps$, then $K = \tp{\frac{LM}{d\eps}}^{\Omega(d)}$.
\end{theorem}
%\htodo{I write $e^{\Omega\tp{\frac{d}{2}\log \frac{LM}{d\eps}}}$ rather than $e^{\frac{d}{2}\log \Omega\tp{\frac{LM}{d\eps}}}$ because it seems that we cannot get $e^{\frac{d}{2}\log \Omega\tp{\frac{LM}{d\eps}}}$ and can only get $e^{\frac{d}{2}\log \tilde\Omega\tp{\frac{LM}{d\eps}}}$}
%\ctodo{Can we write $\tp{\frac{LM}{d\eps}}^{\frac{d}{2}\tp{1-o(1)}}$. (I think it is somehow important to keep the constant $\frac{1}{2}$ explicit in the bound, since we believe that it is the correct constant. This constant matters a lot in some cases, e.g. \Lovasz local lemma)}
% We remark that the assumption $LM=\Omega(d)$ is natural because under \Cref{assump:moment}~and~\ref{assump:smooth}, $LM$ is at least $d$ (see \Cref{lem:lb-LM}).
Besides, when $\frac{LM}{d\eps}=\omega(1)$, the constant coefficient in the exponent is $\frac{1}{2}-o\tp{\frac{1}{d}}$, i.e., we can prove a lower bound of the form $\tp{\frac{LM}{d\eps}}^{\frac{d}{2}-o(1)}$. The detailed formula is given in \Cref{thm:main}.  

%To the best of our knowledge, this is the first general lower bound result on this problem. Our lower bound covers the previous lower bound in~\cite{GLR18} and our techniques can also be applied to~\cite{Cha24} to establish a more general bound in their context.

Note that $LM\ge d$ is a natural assumption for distributions (see \Cref{lem:lb-LM}). This indicates that, in general exponentially many queries are needed for sampling, and when $LM=\omega(d)$, super-exponential queries are required. This result significantly improved previous lower bound of $e^{\Omega(d)}$ in~\cite{GLR18}. Our techniques can also be applied to~\cite{Cha24} and may offer some inspiration for non-log-concave sampling problems under other settings like~\cite{HB23} to establish more general bounds in their contexts. 

The dependency of other parameters also provides new insights. Notably, we first incorporate the dependence of $\eps$ in the lower bound, establishing an $(1/\eps)^{\Omega(d)}$ dependence. This rules out the possibility of fast samplers, i.e., the samplers with $\!{polylog} \tp{1/{\eps}}$ dependence, in non-log-concave cases, which contrasts with the case of log-concave sampling. For another example, according to the result of \cite{HMBE24}, dependence of $\+O\tp{\log \frac{1}{\eps}}$ can be achieved when the distribution satisfies the functional \Poincare inequality (FPI). Although their result is in terms of the expected query complexity, we can terminate the algorithm once the number of queries exceeds $\+O\tp{\frac{1}{\eps} \log \frac{1}{\eps}}$ to get a worst case complexity bound. By Markov's inequality, this still ensures an error of at most $\+O(\eps)$ in total variation distance. Compared with \Cref{thm:main-lb}, this implies that in high-dimensional cases, the FPI does not hold in general for distributions under Assumption~\ref{assump:moment}~and~\ref{assump:smooth}.

%Besides, we design a sampling algorithm which requires at most $\tp{\frac{LM}{d\eps}}^{\+O(d)}$ queries of $f$ and $\grad f$. Taken together, we get a tight (up to the constant in the exponent) bound for general non-log-concave sampling problems.

% \htodo{Our methods are similar to holden lee. But in \Cref{sec:lb}, we demonstrated why we cannot introduce the dependency on $\eps$ if we still use their hard instances. }

\paragraph{Query complexity upper bounds for sampling}
We also design an algorithm to sample from $\mu$ with an $\eps$ error in total variation distance for any $\mu$ satisfying \Cref{assump:moment}~and~\ref{assump:smooth}, with the query complexity bounded by $\tp{\frac{LM}{d\eps}}^{\+O(d)}\cdot \!{poly}(\eps^{-1},d,L,M)$. %See the following \Cref{thm:main-ub}.
\begin{theorem}\label{thm:main-ub}
    Assume $d\ge 3$. There exists an algorithm such that, for any distribution $\mu$ with density $p_\mu(x)\propto e^{-f_\mu(x)}$ where $f_\mu\in C^1(\bb R^d)$, $\grad f(0)=0$, and satisfies \Cref{assump:moment}~and~\ref{assump:smooth}, outputs a sample $x$ with distribution $\tilde \mu$ satisfying $\DTV(\mu,\tilde \mu)\le \eps$ within $\tp{\frac{LM}{d\eps}}^{\+O(d)}\cdot \!{poly}\tp{\eps^{-1},d,L,M}$ query accesses to $f_\mu$ and $\grad f_\mu$, for any $\eps\in (0,1)$. 
\end{theorem}
The assumption $\grad_\mu f(0)=0$ is for the sake of simplicity. One can first apply a gradient descent algorithm to find a stationary point of $f$ and shift the origin to that point. The form of this bound matches our lower bound in \Cref{thm:main-lb}, which suggests that the optimal query complexity for sampling from non-log-concave distributions should be $\tp{\frac{LM}{d\eps}}^{\Theta(d)}$. 
%\ctodo{The notation $\E[X\sim p]{.}$.}

\paragraph{The smoothness evolvement during the OU process}

As mentioned before, the work of~\cite{HZD+24} presented a quasi-polynomial algorithm for sampling from distributions with second moment at most $M=\!{poly}(d)$. They further require that all distribution along the trajectory of the OU process initialized at the target distribution is $\+O(1)$-log-smooth. Our lower bound already implies that the condition is strictly stronger than merely requiring the target distribution to be $\+O(1)$-smooth. We further construct an explicit family of distributions, called ``stitched Gaussians'', and calculate its evolvement of the smoothness property along the OU trajectory. This family of distributions are $\+O(1)$-log-smooth while along the OU process, the Hessian of their log density becomes unbounded.

%We also investigate the ``smoothness along the trajectory'' condition mentioned above and show that it is strictly stronger than merely requiring the initial distribution to be smooth. Specifically, we construct explicit distributions which are $\+O(1)$-log-smooth while along the OU process, the Hessian of their log density becomes unbounded. Therefore, it is interesting to study under what conditions, the $\+O(1)$-log-smoothness condition can be preserved during the OU process. We obtain some partial results for the mixtures of Gaussians.

%Besides the work on complexity bounds, to provide an intuitive comparison between our results and the results in \cite{HZD+24}, we compare \Cref{assump:smooth} with the smoothness condition in their work. They assume the distributions throughout the entire OU process to be $L$-log-smooth. Previous works do not pay much attention to distinguishing between these two smoothness conditions. However, we find that even though the initial distribution is $\+O(1)$-log-smooth, it might evolve into an $\omega(1)$-log-smooth distribution during the OU process. 
\begin{theorem}\label{thm:main-smooth}
    For the OU process $\ab\{X_t\}_{t\geq 0}$, let $\mu_t$ be the distribution of $X_t$ and its density is $p_t$. For arbitrary $s=\Omega(d)$, there exists an initial distribution $\mu$ which is $\+O(1)$-log-smooth but $\norm{\grad^2 \log p_t(x)}_{\!{op}} \geq \Omega\tp{e^{-2t}s-1}$ for any $t\geq \frac{\log 10}{2}$ at some point $x$.
\end{theorem}
% The specific distributions we find that satisfy the property in \Cref{thm:main-smooth} is a family of multi-modal distributions called the stitched Gaussian. 
Our results indicate that on those distributions with the property in \Cref{thm:main-smooth}, the upper bound of \cite{HZD+24}, which is $e^{\tilde{\+O}(L^3)}$ with $L$ being the smoothness bound during whole the OU process, can be very large. In contrast, our algorithm only requires smoothness of the initial distribution, allowing it to achieve a better complexity bound than~\cite{HZD+24} in these cases.

Moreover, our results also suggest that it is interesting to study under what conditions, the $\+O(1)$-log-smoothness condition can be preserved during the OU process. We obtain some partial results for mixtures of Gaussian and summarize them in \Cref{tab:result-comp}.

%Furthermore, this surprising result raises an interesting question: when will the property of $\+O(1)$-smoothness be preserved throughout the OU process? We provide an analysis to this question on the mixture of Gaussian distributions and summarize our results in \Cref{tab:result-comp}. 

% We compare our \Cref{assump:smooth} with the smoothness assumption in previous DDPM-based works, where they assume the distributions throughout the entire OU process are $L$-smooth. We prove that, even though the initial distribution is $\+O(1)$-smooth, it might evolve into an $\omega(1)$-smooth distribution during the OU process. This is a surprising result and suggests that assuming the smoothness throughout the entire OU process may not be trivial.

\paragraph{The comparison between sampling and optimization} 

Finally, we provide a comparison between the sampling and optimization tasks. We establish the lower bound in \Cref{thm:main-opt-lb} for optimization problems in non-convex cases based on the results in \cite{MCJ+19}. Let $x^*$ be the minimizer of function $f$ and $x^{(k)}$ be the point queried by the algorithm at step $k$. 
\begin{theorem}\label{thm:main-opt-lb}
    Assume $LM=\Omega(d)$ and $d\geq 8$. If an optimization algorithm $\+A$ queries at most $K$ points and can guarantee that $\min_{k\leq K} \abs{f(x^{(k)}) - f(x^*)}<1$ with constant probability for any $L$-smooth function $f:\bb R^d\to \bb R$ satisfying that the second moment of the distribution $\mu$ with density $\propto e^{-f}$ is $\Theta(M)$, then $K$ is at least $ (\alpha\cdot LM)^{\frac{d}{2}}$ for some universal constant $\alpha>0$.
\end{theorem}
On the other hand, for constant $\eps>0$, our upper bound for the sampling task shows that one requires only $\tp{\frac{LM}{d}}^{\+O(d)}$ queries to draw an approximate sample from $\mu$. This demonstrates that, in the context of non-convex optimization and non-log-concave sampling under \Cref{assump:moment}~and~\ref{assump:smooth}, when $LM=\Theta(d)$, sampling from a distribution with density $\propto e^{-f}$ can indeed be faster than finding the minimizer of $f$ by a super-exponential factor.

% We compared the complexity of the sampling and optimization problems under these two assumptions. The relationship between optimization and sampling has always been a topic of interest.  
% This suggests that, for this specific class of functions, sampling is harder than optimization.

\subsection{Technical overview}

\subsubsection{The lower bound}

The general idea to prove a query complexity lower bound, as used in previous work~\cite{GLR18,Cha24}, is to construct many distributions that are difficult for the algorithm to distinguish. Specifically, for each $v\in \bb R^d$, one can perturb the density of a base distribution $\mu_0$ in $\+B_r(v)$, a ball of radius $r$ centered at $v$, to get a new distribution $\mu_v$.  For two vectors $u$ and $v$ with $\+B_r(u)\cap \+B_r(v)=\emptyset$, after proper perturbation, one can get $\DTV(\mu_{v},\mu_u)\approx \eps$. Suppose one can find $n$ disjoint balls, with centers $\set{v_1,\dots,v_n}$. Then if the algorithm wants to sample from $\mu_{v_i}$ within error $\eps$, it should first recognize $v_i$ among the $n$ candidates. This indicates that when the input instance is $\mu_0$, the algorithm must query almost every $\+B_r(v_i)$, which gives the lower bound $\Omega(n)$.

Then at a high level, proving lower bounds can be viewed as packing as many disjoint $\+B_r(v)$'s in the domain of the base distribution while maintaining the desired properties (smoothness and bounded second moment) of the distribution. Previous works~\cite{GLR18,Cha24} 
%Suppose one can pack $n$ disjoint balls with centers $\set{v_i}_{i\in [n]}$. Then at a high level, the lower bound come from the following fact: any algorithm that can correctly sample from all these distributions must query almost every $B_r(v_i)$ when ithe
chose a Gaussian distribution as the base $\mu_0$ and proved a lower bound of $e^{\Omega(d)}$ using the above arguments. 

In this work, we want to derive a general lower bound with regard to all parameters $d,L,M$ and $\eps$. Our key observation is that using Gaussian as the base distribution is suboptimal. Let us explain the reason.
%To achieve the correct dependence on $\eps$, however, this simple base distribution is not viable. 
%If we simply choose $\mu_0=\+N\tp{0,\frac{M}{d}I_d}$ (the reason for using this Gaussian is that its second moment is exactly $M$), we claim that $r$ should be $\Omega\tp{\sqrt{\frac{d}{L\eps}}}$ and the number of balls we can pack is smaller than $n$.
There exists a trade-off between packing more disjoint balls and maintaining properties of the perturbed distributions. To pack more balls, one needs the radius $r$ to be as small as possible, and the center $v_i$ should be as far as possible from the origin. However,
%Note that the main challenge for the above strategy is to design a proper perturbation rule. The following requirements are necessary for a legal perturbation:
% To see the reasons behind this, we now examine what issues arise when $\mu_0=\+N\tp{0,\frac{M}{d}I_d}$ (the reason for using this Gaussian is that its second moment is exactly $M$). 

% To simulate this new distribution, the algorithm must recognize the perturbed region first. For simplicity, we let the perturbed region to be a ball $\+B_r(v)$ with center vector $v$ and radius $r$. To select such $r$ and $v$'s, we need the following requirements:
\begin{itemize}
    % \item [1.] 
    % By perturbing the density of $\mu_0$ in $\+B_r(v)$, we get a new distribution $\mu_v$. 
    % After perturbation, the mass in $\+B_r(v)$ increases by $\approx \eps$. Then the total variation distance between $\mu_{v}$ and $\mu_{u}$ will be $\approx \eps$ for those $u,v$ such that $\+B_r(u)\cap \+B_r(v)=\emptyset$. Thus, we can reduce problem of sampling with error $\epsilon$ to the problem of recognizing the perturbed regions.
    \item [1.] As the mass of perturbed area is approximately $\eps$ more than the base distribution, each $\+B_r(v_i)$ should be placed inside $\+B_R$ for $R=\+O\tp{\sqrt{\frac{M}{\eps}}}$ to ensure that the second moment of $\mu_{v_i}$ remains $\+O(M)$.
    \item [2.] The radius $r$ must be large enough to guarantee a smooth transition at the boundary, ensuring that $\mu_{v_i}$ remains $\+O(L)$-log-smooth.
\end{itemize}
% Under these conditions, we can find $n=()$ disjoint balls with centers $\ab\{v_i\}_{i\in[n]}$ and $r=\+O\tp{\sqrt{\frac{d}{L}}}$. To recognize the perturbed ball among the $n$ ones, it will take $()$ queries. 
Suppose we pick a Gaussian distribution as the base distribution $\mu_0$ as in the previous work, namely that $\mu_0 = \+N\tp{0,\frac{M}{d}\!{Id}_d}$ (so that the second moment of $\mu_0$ is $M$). Let function $h_0$ and $f_i$ be the log-density of $\+N\tp{0,\frac{M}{d}\!{Id}_d}$ and $\mu_{v_i}$ respectively.  Suppose we want to pack all balls with their centers at $R\cdot\+S_{d-1}$ (the sphere with radius $R$). Then to maintain the second moment of each $\mu_{v_i}$ to be $\+O(M)$, one should pick $R = \+O\tp{\sqrt{\frac{M}{\eps}}}$. For those points $x$ with $\norm{x}\approx R\approx \sqrt{\frac{M}{\eps}}$, we have $h_0(x)\approx \frac{d}{\eps}$. 

% If we pack these balls inside $\+B_{R}$ with $R=\+O\tp{\sqrt{\frac{M}{\eps}}}$, to obtain the desired form lower bound, we want at least $()$ balls and thus we need $r=\+O(\sqrt{\frac{d}{L}})$. 
To guarantee the mass of $\+B_r(v)$ to be $\approx \eps$, we have $\int_{\+B_r(v_i)} \tp{e^{-f_{i}(x) } - e^{-h_0(x)}} \d x \approx \eps$. This indicates that $h_0(x) -f_i(x) \approx \frac{d}{\eps}$ inside $\+B_r(v_i)$. To further guarantee the $\+O(L)$-smoothness of $f_i$, $r$ should be $\Omega\tp{\sqrt{\frac{d}{L\eps}}}$. One can pack approximately $\tp{\frac{R}{r}}^d$ (\Cref{lem:disjointcap}) many disjoint balls with radius $r$ centering at $R\cdot \+S_{d-1}$, which gives a lower bound approximately $\tp{\alpha\cdot \frac{LM}{d}}^{\frac{d}{2}}$ for some universal constant $\alpha>0$. The dependency on $\eps$ has been cancelled, and it is obviously not optimal!
% Then the difference in the potential function between the center of the ball and its outer border needs to be $\Omega\tp{\frac{d}{\eps}}$ and  the potential function in this region is only $\+O\tp{\frac{L}{\eps}}$-smooth rather than $\+O(L)$-smooth. When $\eps$ is small, this falls far short of meeting our third requirement. 

The key trade-off of the construction comes from how to efficiently perturb the mass in each $\+B_r(v)$. Recall that one requires
\[
    \int_{\+B_r(v_i)} \tp{e^{-f_{i}(x) } - e^{-h_0(x)}} \d x \approx \eps
\]
while keeping $r$ small and simultaneously $f_i - h_0$ small. We observe that, ignoring low order terms in $r$, the integral on the LHS is almost equal to $e^{-h_0(x)}e^{h_0(x)-f_i(x)}$. Imagine we are filling a mound with sand within a small circular area at an elevation of $ e^{-h_0(x)} $. We aim for the amount of sand used to satisfy $ e^{-h_0(x)} \cdot e^{h_0(x)-f_i(x)} \approx \eps $, while ensuring the mound is not too steep -- meaning the height difference \( h_0(x) - f_i(x) \) should be minimized. The most efficient approach is to raise the base elevation, i.e., maximize $e^{-h_0(x)}$, thereby reducing the required height increment.
%Moreover, the value of $f_i-h_0$ affects the value of the integral at an exponential rate while affects the value of 

Therefore, at a very high level, our construction for the base distribution $p_{\mu_0}\propto e^{-f_{\mu_0}}$ in the lower bound proof is to first modify the Gaussian by creating a plateau in a ring, and then perturb the mass on the plateau. The construction is illustrated in \Cref{fig:lb}. Of course, the plateau itself may affect the smoothness and the second moment. Nevertheless, we carefully pick the location, shape and the mass of the plateau, and are able to pack approximately $\tp{\alpha\cdot \frac{LM}{d\eps}}^{\frac{d}{2}}$ disjoint balls with radius $r=\tilde{\+O}\tp{\sqrt{\frac{d}{L}}}$, which provides the desired optimal lower bound. 

% We artificially increase the density of $\+N\tp{0,\frac{M}{d}I_d}$ (or equivalently, decreasing the value of $h_0$) on the ring $\+B_R\setminus \+B_{\frac{R}{2}}$ to create a new base distribution $\mu_0$ with density $\propto e^{-f_{\mu_0}}$. Then, we place $\+B_r(v_i)$ with $\|v_i\|=\frac{3R}{4}$ and $r=\tilde{\+O}\tp{\sqrt{\frac{d}{L}}}$ on the ring. By perturbing the values in $\+B_r(v_i)$, we get a new distribution $\mu_{v_i}$. By carefully choosing the density values on the ring, we can ensure that \Cref{assump:moment} still hold for $\mu_0$ and each $\mu_{v_i}$. Most importantly, $f_{\mu_0}(x) - f_i(x) = {\+O}\tp{d\log \frac{LM}{d\eps}} \ll h_0(x) - f_i(x)$ inside $\+B_r(v_i)$. This ensures the $\+O(L)$-smoothness inside $\+B_r(v_i)$ when $r=\+O\tp{\sqrt{\frac{d}{L}}}$. 
% the second moment of the new base distribution is still $\+O(M)$, and the difference in potential function values between the center and outer border of $\+B_r(v)$ is $\+O\tp{d\log \frac{LM}{d\eps}}$. This ensures the $\+O(L)$-smoothness inside and around $\+B_r(v)$. 

%Intuitively, the advantage of this $\mu_0$ is that,  it provides a \emph{high platform} on the ring and the transition from $h_0$ to $f_i$ is partly undertaken by the ring. Although the width of this ring is $\+O(R)$, it is part of the base instance and does not affect the packing of small balls. Therefore, this addresses the issue when using Gaussian as the base instance.

\begin{figure}[h!]
	\centering
    % \ExecuteMetaData[figures.tex]{lbfigure}
    \includegraphics[scale=0.3]{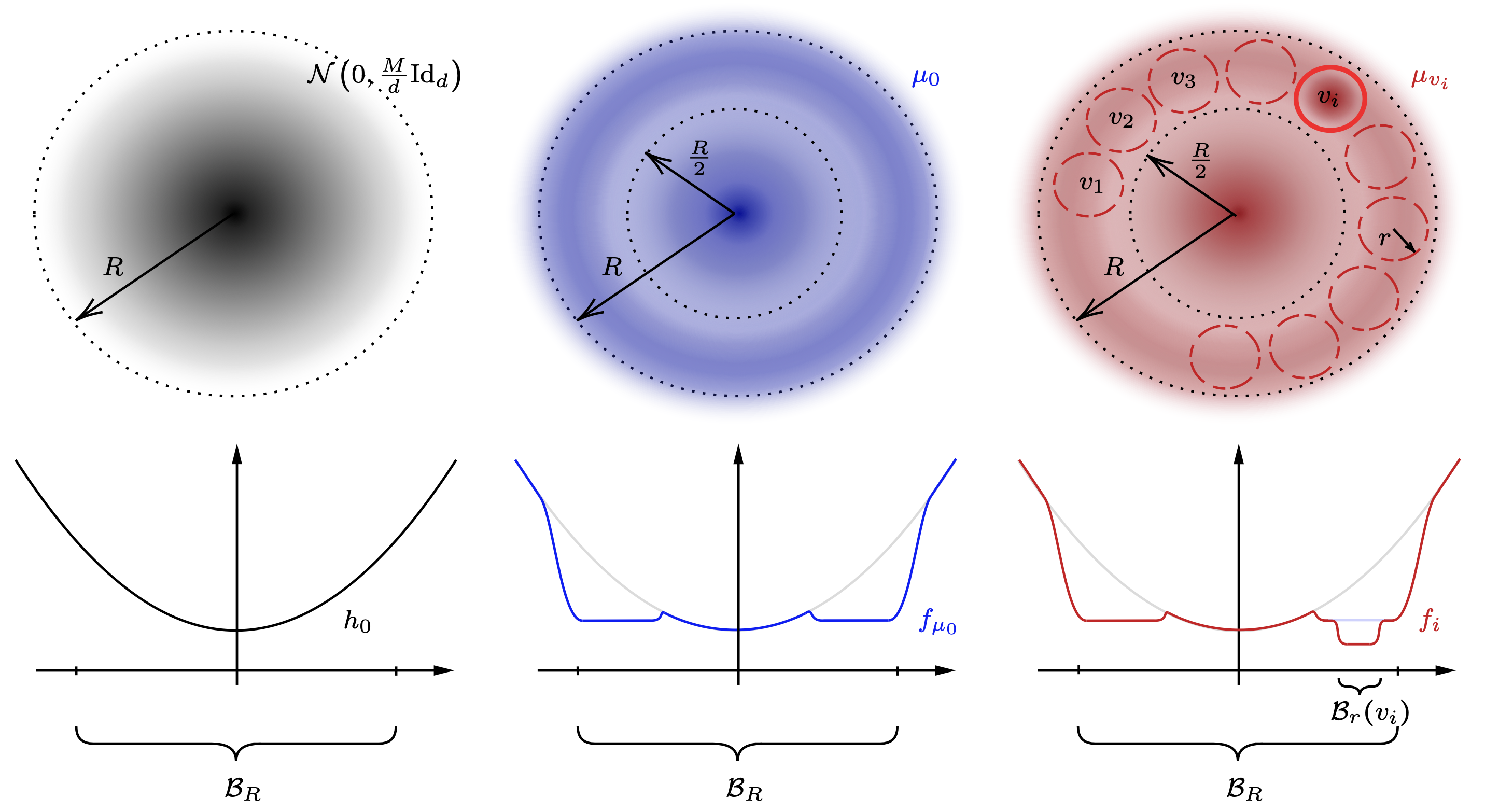}
  \caption{Construction of $f_{\mu_0}$ and $f_{\mu_v}$.\protect\footnotemark}
  \label{fig:lb}
\end{figure}
\footnotetext{In the above three figures of \Cref{fig:lb}, the deeper color represents larger density.}

\subsubsection{The upper bound}

%To sample from general non-log-concave distributions under \Cref{assump:moment}~and~\ref{assump:smooth}, the main challenge is that, the target distribution $\mu$ may have extremely small probability values in certain regions, causing the algorithm to get trapped locally and converge slowly. To mitigate this issue, we first modify the potential function $f_{\mu}$ before running a sampling algorithm. 

It is challenging to establish suitable isoperimetric inequalities directly for distributions solely satisfying \Cref{assump:moment} and~\ref{assump:smooth} due to the existence of point with extremely small density. However, we observe that: 1) the mass of the target distribution $\mu$ with density $\propto e^{-f_{\mu}}$ is concentrated in $\+B_{R}$ for a sufficiently large radius $R$; 2) the total mass of the region with extremely small density values inside $\+B_{R}$ is small.
% For those extremely small $\mu(x)$ inside $\+B_{R}$, we can truncate these extreme values of $f_{\mu}$ without affect $\mu$ much.
Based on these observations, we construct another distribution $\pi$  with density $\propto e^{-f_{\pi}}$ which is close to $\mu$ in total variation distance and is easier to sample from. Basically the distribution $\pi$ is Gaussian outside $\+B_R$ and is the truncation (remove points with extremely small density) of $\mu$ inside $\+B_R$. The construction of $f_{\pi}$ can be described in the following steps:
%\htodo{Here I use $\+B_R$ and I don't mention the $\+B_{2R}$ used in \Cref{sec:ub} for simplicity. So here the descriptions are slightly different with our actual implementation. Does this matter?} 
\begin{itemize}
    \item Step 1: Discretize $\+B_R$ into small cubes. Use the value of $f_{\mu}$ at the center of each cube to estimate the values inside the cube. Use these approximations of each cube to calculate a rough estimation of the minimum value $\wh f^* \approx f^*= \min_{x\in \+B_R} f(x)$ and the normalizing factor $\wh Z_{\mu}\approx \int_{\bb R^d} e^{-f_{\mu}(x)} \dd x$.
    \item Step 2: For $x\in \+B_R$, as shown in \Cref{fig:ub}, if $f_{\mu}(x) - \wh f^*$ exceeds some threshold $h_1-\wh f^*$, let $f_{\pi}(x)$ be the smooth truncation of $f_{\mu}(x)$. Otherwise, let $f_{\pi}(x) = f_{\mu}(x)$. By doing so, we guarantee that $f_{\pi}(x) - \wh f^*$ is always bounded by $h_2-\wh f^*$ for some value $h_2$ and extremely small density values in $\pi$ is circumvented.
    \item Step 3: For $x\not \in \+B_R$, %since the impact of this area is minimal, we can adopt a more straightforward construction. As \Cref{fig:ub} shows, 
    we define $f_{\pi}(x)$ by replacing the original density $\frac{e^{-f_{\mu}(x)}}{\wh Z_{\mu}}$ with the density of a Gaussian $\+N\tp{0, \frac{M}{\eps d}\cdot \!{Id}_d}$.
\end{itemize}

\begin{figure}[h!]
	\centering
    % \ExecuteMetaData[figures.tex]{ubfigure}
    \includegraphics[scale=0.35]{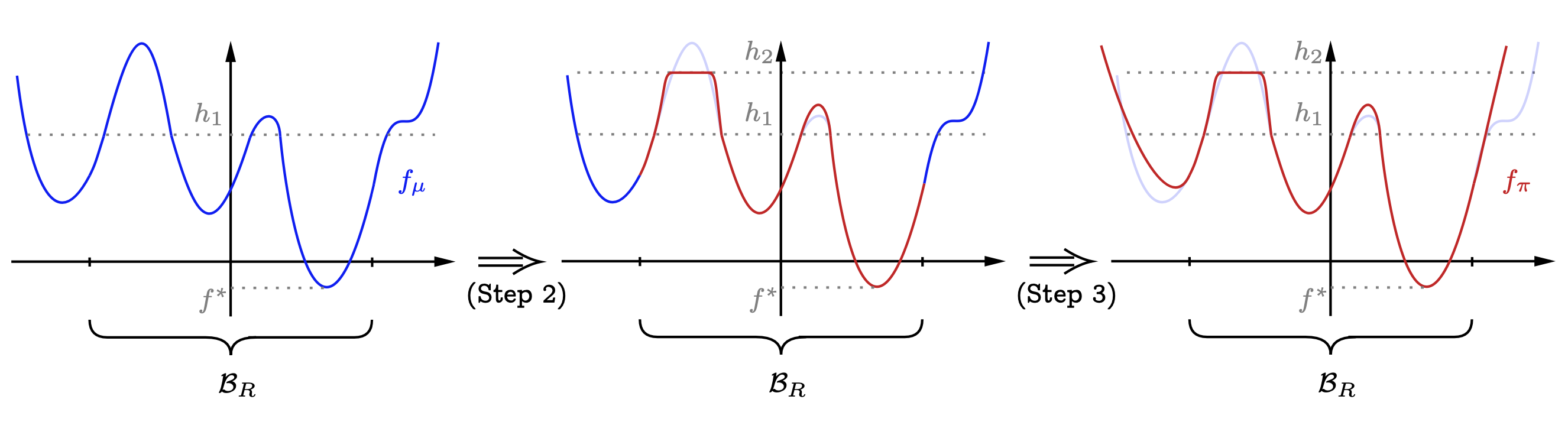}
  \caption{Construction of $f_{\pi}$.}
  \label{fig:ub}
\end{figure}

By smoothing the above construction appropriately, we can prove the smoothness of $f_{\pi}$ and prove a bound for the \Poincare constant of $\pi$. Besides, the value of $f_{\pi}$ and $\grad f_{\pi}$ can be efficiently calculated given query access to $f_{\mu}$ and $\grad f_{\mu}$. Then we can apply the averaged Langevin algorithm in \cite{BCE+22} to sample from $\pi$ and the output distribution is also close to $\mu$ in total variation distance.

Note that in Step 1, we use a grid-based approximation to find $\wh f^*$, which is an optimization task. However, the comparison between \Cref{thm:main-ub,thm:main-opt-lb} shows that our sampling algorithm has lower complexity than solving the optimization problem itself. This is because the required precision for this optimization in our task is very low (see \Cref{prop:Z-and-fmin}). We make a more thorough discussion on the theme ``sampling versus optimization'' in \Cref{sec:sampling-vs-opt}.

%Therefore, even though our sampling algorithm first requires solving an optimization problem, the sampling task is actually simpler than the optimization. 

% The idea of truncation in step 2 has also been used in a recent work~\cite{Cha24}, with different implementation and purpose. Their goal is to ensure the effectiveness of importance sampling. They truncate regions where the potential function is small, transforming the original distribution into a log-concave one to facilitate approximate sampling. In contrast, we truncate regions with large potential values (i.e., very small density), and to keep $\pi$ close enough to $\mu$, our truncation must be more sophisticated.
%\htodo{Do we need this comparison with the technique in \cite{Cha24}? If need, please check whether this is correct.}
%\ctodo{I think it is not necesssary.}

\section{Preliminaries}

\paragraph{Notations}
In this paper, all logarithms refer to the natural logarithms with base $e$. 
For two distributions $\mu$ and $\nu$ over $\bb R^d$ with density function $p_{\mu}$ and $p_{\nu}$, the total variation distance is defined as $\DTV(\mu,\nu)=\frac{1}{2}\int_{\bb R^d}\abs{p_\mu(x)-p_\nu(x)}\dd x$. The Kullback-Leibler divergence (KL divergence) is defined as $\!{KL}(\mu\|\nu)=\E[X\sim \mu]{\log \frac{p_\mu(X)}{p_\nu(X)}}$. 

Unless otherwise specified, the distributions we consider in this paper are all under \Cref{assump:moment}~and~\ref{assump:smooth}. We say a distribution $\mu$ is $L$-log-smooth if its potential function is $L$-smooth.  We always assume $LM\geq d$ because from \Cref{lem:lb-LM}, $LM$ is lower bounded by $d$ as long as $\grad f(0)=0$.

We use the notation $\+N(u,\Sigma)$ to denote a Gaussian distribution over $\bb R^d$ with mean $u\in \bb R^{d}$ and covariance matrix $\Sigma\in \bb R^{d\times d}$. We use $\!{Id}_d$ to denote the identity matrix in $\bb R^{d\times d}$ and let $\|A\|_{\!{op}}$ denote the operator norm of a matrix $A$.

% \paragraph{Balls in $d$ dimension}

For a ball centered at a point $v\in \bb R^d$ with radius $R$, we denote it as $\+B_R(v)$. When $v=0$, we abbreviate $\+B_R(v)$ as $\+B_R$.
%  Recall that $\!{vol}(\+B_R) = \frac{(\pi R^2)^{\frac{d}{2}}}{\Gamma\tp{\frac{d}{2}+1}}$.
%  We also find the following bounds for $\!{vol}(\+B_R)$ useful.

% \begin{proposition}\label{prop:dballvolbound}
% $\tp{\frac{e\pi R^2}{d}}^{\frac{d}{2}}\le\!{vol}\tp{\+B_R} \le \tp{\frac{2e\pi R^2}{d}}^{\frac{d}{2}}$.
 
% \end{proposition}

\paragraph{The mollifier}\label{sec:mollifier}

Define the function $q_{\!{mol}}\colon \bb R\to [0,1]$ as 
$
q_{\!{mol}}(z) = 
    \begin{cases}
        0, & z<0,\\
        6z^5-15z^4+10z^3, &z\in [0,1],\\
        1, & z>1.                
    \end{cases}
$
It is a mollifier between $0$ and $1$, and it is easy to see that $q_{\!{mol}}$ has the following properties.
\begin{proposition}
    The following holds for $q_{\!mol}$.
    \begin{itemize}
        \item For $z\le 0$, $q_{\!{mol}}(z)=0$; for $z\ge 1$, $q_{\!{mol}}(z)=1$; and for $z\in (0,1)$, $q_{\!{mol}}(z)\in (0,1)$.
        \item $q'_{\!{mol}}(0)=q'_{\!{mol}}(1)=0$ and for each $z\in [0,1]$, $\abs{q'_{\!{mol}}(z)}<\infty$.
        \item $q''_{\!{mol}}(0)=q''_{\!{mol}}(1)=0$ and for each $z\in [0,1]$, $\abs{q''_{\!{mol}}(z)}<\infty$.
    \end{itemize}
\end{proposition}

% \paragraph{The Markov's inequality} For any non-negative random variable $X$ and any positive real number $s$, the Markov's inequality states that 
% \[
%     \Pr{X\geq s}\leq \frac{\E{X}}{s}.
% \]

% For each $q\in (1,\infty)$, the Renyi divergence between two distributions $\nu$ and $\pi$ is defined as
% \[
%     R_q(\nu\|\pi) = \frac{1}{q-1}\ln\tp{\bigg\|\frac{\dd \nu}{\dd \pi}\bigg\|^q_{L^q(\pi)}}
% \]

\paragraph{The Langevin dynamics and Ornstein-Uhlenbeck process}
The Langevin dynamics is a continuous-time process $\ab\{X_t\}_{t\geq 0}$ described in the following stochastic differential equation:
\begin{equation}
    \d X_t = - \grad f(X_t)\d t + \sqrt{2}\d B_t, \label{eq:LD}
\end{equation}
where $f:\bb R^d\to \bb R$ is a differentiable function and $\ab\{B_t\}_{t\geq 0}$ is the standard Brownian motion. The Ornstein-Uhlenbeck process (OU process) is a special case of \Cref{eq:LD} with $f(x) = \frac{\|x\|^2}{2}$. It is well known that the law of $X_t$ always converges to the standard Gaussian distribution in the OU process.

\paragraph{The \Poincare inequality}
We say a distribution $\mu$ satisfies the \Poincare inequality with a constant $C>0$ if for all $f\in C^1(\bb R^d)$ with $\E{f^2}< \infty$, it holds that
\begin{equation}\label{eqn:def-PI}
    \Var[\mu]{f} \le \frac{1}{C}\cdot \E[\mu]{\|\grad f\|^2}.
\end{equation}
We use $C_{\!{PI}}$ to denote the largest $C>0$ so that \eqref{eqn:def-PI} holds. We also call $C_{\!{PI}}$ the \Poincare constant of $\mu$. 

% \subsection{The log-Sobolev inequality}

% We say a distribution $\mu$ satisfies the log-Sobolev inequality with a constant $C>0$ if for all $f\in C^1(\bb R^d)$ with $\E{f^2}< \infty$, it holds that
% \begin{equation}\label{eqn:def-LSI}
%     \Ent{f^2} \le \frac{2}{C}\cdot \E{\norm{\grad f}^2},
% \end{equation}
% where $\Ent{g}$, the entropy of a function $g$, is defined as
% \[
%     \Ent{g} \defeq \E{g\log g}-\E{g}\E{\log g}.
% \]
% We use $C_{\!{LSI}}$ to denote the largest $C>0$ so that \eqref{eqn:def-LSI} holds. We also call $C_{\!{LSI}}$ the log-Sobolev constant of $\mu$. 

\input{lower-bound}

\input{upper-bound}

\input{smoothness}

\input{optimization}

\section{Conclusion and open problems}

In the paper, we studied the query complexity of sampling from $L$-log-smooth distributions with the second moment at most $M$. For this family of distributions, we established a $\tp{\frac{LM}{d\eps}}^{\Theta(d)}$ query complexity bound. It is an interesting question to explore the correct constant in the exponent. 
\begin{problem}
    Determine the infimum $c>0$ such that the query complexity of sampling from an $L$-log-smooth distribution with the second moment at most $M$ is $\tp{\alpha\cdot \frac{LM}{d\eps}}^{cd}$ for some universal constant $\alpha>0$.
\end{problem}
In light of our lower bound proof, we conjecture that the correct constant is $\frac{1}{2}$. In fact, our algorithm for the upper bound compromised significantly on this constant during the truncation of target distribution. Our approach also has the drawback of only applying to bounding the total variation distance. An optimal algorithm might rely on directly establishing functional inequalities (e.g. weak \Poincare inequality~\cite{HMRW24}) for the target distribution.

\begin{problem}
    Establish tight functional inequalities for $L$-log-smooth distributions with the second moment at most $M$.
\end{problem}

One of the motivation of this work is to understand the extent to which the quasi-polynomial time algorithm of~\cite{HZD+24} applies. As we investigated in \Cref{sec:OU-smooth}, finding a criterion for being $\+O(1)$-smooth along the trajectory of the OU process is a challenging task.

\begin{problem}
    Understand to what extent a diffusion process (not restricted to the OU process in ~\cite{HZD+24}) can maintain the $\+O(1)$-smoothness of the initial distribution. Can this condition result in efficient algorithm as well?
\end{problem}

Finally, can we characterize the query complexity for sampling from general non-log-concave distributions. Even in the case of mixtures of Gaussians, this remains a challenging problem.

\begin{problem}
    Characterize the condition under which sampling from a mixture of Gaussians has sub-exponential query complexity.    
\end{problem}

\section*{Acknowledgements}

The authors would like to thank Zongchen Chen for bringing the Hubbard-Stratonovich transform of Ising model into our attention, in particular its connection with multi-modal sampling. We also thank the anonymous reviewers for their valuable comments and suggestions.
\bibliographystyle{alpha}
\bibliography{arxiv}
\appendix

\input{appendix}

\end{document}

%% file: lower-bound.tex
\section{The lower bound}\label{sec:lb}
In this section, we provide a sample complexity lower bound proof with regard to the error tolerance $\eps$. In fact, we prove \Cref{thm:main}, which is a formal version of \Cref{thm:main-lb}. %All distributions constructed in this section will satisfy \Cref{assump:moment} and \Cref{assump:smooth}.

% A sampling algorithm here is a procedure with query access to $f$, $\grad f$ or $\grad^2 f$, which outputs a sample point based on a series of queries. We aim to prove the following theorem in this note.
% \htodo{Here $c$ is the constant in \Cref{lem:disjointcap}.}

Recall that we use $\@D_{L,M}$ to denote the collection of distributions which are $L$-log-smooth and have second moment at most $M$. 
\begin{theorem}\label{thm:main}
    % There exist a universal constant $C>0$ such that 
    For any $L,M>0$ satisfying $LM\ge d$ and for any $\eps\in(0,1/200)$, $d\geq 5$, if a sampling algorithm $\+A$ always terminates within 
\[
    \frac{\eps (d-2)^{\frac{3}{2}}}{8}\cdot \tp{\frac{9}{256} \cdot \frac{LM}{d\eps} \cdot \frac{1}{\log \frac{LM}{d\eps}}}^{\frac{d-1}{2}}
    % \frac{\eps}{4}\cdot \tp{\frac{C\cdot LM}{2 d\eps} \cdot \frac{1}{\log \frac{LM}{d\eps}}}^{\frac{d-1}{2}}% \approx \frac{\eps}{4}\exp\set{\frac{d}{2}\cdot \Omega\tp{\log \frac{LM}{d\eps}}}
\]
queries on every input instance in $\@D_{L,M}$, then there must exist some distribution $\mu\in \@D_{L,M}$ such that when the underlying instance is $\mu$, the distribution of $\+A$'s output, denoted as $\tilde \mu$, is $\eps$ away from $\mu$ in total variation distance, i.e., $\DTV(\mu,\tilde \mu)\geq \eps$.
\end{theorem}
%\htodo{Actually we consider those distributions with second moment $O(M)$ and $O(L)$-smooth. Not exactly $M$ and $L$.}

%\mn{When $d$ is even, $\Gamma\tp{\frac{d}{2}+1} = \tp{\frac{d}{2}}!$. When $d$ is odd, $\Gamma\tp{\frac{d}{2}+1} = \frac{\sqrt{\pi}}{2^d}\cdot \frac{d!}{\tp{\frac{d-1}{2}}!}$.}

\subsection{The base instance}

We first construct a base distribution $\mu_0$. Let $R= \tp{\frac{M}{\eps}}^{\frac{1}{2}}$, and let
\[
    \mathfrak{g}_{[\frac{R}{4},\frac{R}{2}]}(x)=q_{\!{mol}}\tp{\frac{ \|x\|^2- \frac{R^2}{16}}{\frac{R^2}{4} - \frac{R^2}{16}}} \quad \mbox{and} \quad \mathfrak{g}_{[R,2R]}(x)=q_{\!{mol}}\tp{\frac{\|x\|^2-R^2}{4R^2-R^2}},
\]
% Let $\alpha\in (0,2)$ be a constant to be determined later.
With constant function $h_1 \equiv \log \tp{\!{vol}(\+B_{3 R})} + \log \frac{1}{\eps}$ and function $h_0(x)=\frac{d\|x\|^2}{2M} + \frac{d}{2}\log \frac{2\pi M}{d}$, define the function $f_0$ as
\[
    \forall x\in \bb R^d, f_0(x) = \begin{cases} h_0(x), & \|x\|\leq \frac{R}{4} \\
    \mathfrak{g}_{[\frac{R}{4},\frac{R}{2}]}(x)\cdot h_1 + \tp{1-\mathfrak{g}_{[\frac{R}{4},\frac{R}{2}]}(x)}\cdot h_0(x), & \frac{R}{4}<\|x\| \leq \frac{R}{2}\\
    h_1, & \frac{R}{2} <\|x\|\leq R\\
    \mathfrak{g}_{[R,2R]}(x)\cdot h_0(x) + \tp{1-\mathfrak{g}_{[R,2R]}(x)}\cdot h_1, & R<\|x\|\leq 2R \\
    h_0(x), & \|x\|>2R
    \end{cases}.
\]

%\mn{Here we use a constant function $h_1$ rather than using Gaussian directly. This is crucial.}
Consider the distribution $\mu_0$ with density $p_{\mu_0} \propto \exp\tp{-f_0(x)}$ and its normalizing factor $Z_0 = \int_{\bb R^d} \exp\tp{-f_0(x)} \d x$.
\begin{lemma}\label{lem:Z_0}
    The normalizing constant $1-16\eps \leq Z_0 \leq  1+\eps$. 
\end{lemma}
\begin{proof}
    On one hand, from Markov's inequality,
    $$
        Z_0\geq \int_{\+B_{\frac{R}{4}}} e^{-f_0(x)} \dd x = \int_{\+B_{\frac{R}{4}}} e^{-h_0(x)} \dd x = 1- \Pr[X\sim \+N\tp{0,\frac{M}{d}\cdot \!{Id}_d}]{\|X\|^2\geq \frac{R^2}{16}} \geq 1-16\eps.
    $$
    On the other hand, 
    $$
        Z_0\leq \int_{\bb R^d} e^{-h_0(x)} \dd x + \vol (\+B_{2R})\cdot e^{-h_1} \leq 1 + \eps\cdot \frac{\vol (\+B_{2R})}{\vol(\+B_{3R})}\leq 1+\eps.
    $$
    
    % Note that $h_1=\log \frac{1}{\eps} + \frac{d}{2}\log\tp{\pi \alpha^2 R^2} - \log \Gamma\tp{\frac{d}{2}+1}$.
    
    % Note that for each fixed $x\in \bb R^d$, $f_0$ is a non-decreasing function wrt $\alpha$. So $Z_0$ is decreasing as $\alpha$ increases. When $\alpha=1$,
\end{proof}

% We fixed the value of $\alpha$ to be the one in \Cref{lem:Z_0}.
\begin{lemma}\label{lem:propertymu0}
    The distribution $\mu_0$ is $\+O(L)$-log-smooth and has second moment $\+O(M)$.
\end{lemma}
%\htodo{$O(M)$ and $O(L)$, not exactly $M$ and $L$} 
\begin{proof}
    We first calculate the second moment of $\mu_0$. We have
    $$
        \E[\mu_0]{\|X\|^2} \leq \frac{\E[X\sim \+N\tp{0,\frac{M}{d}\cdot \!{Id}_d}]{\|X\|^2} + \vol (\+B_{2R})\cdot e^{-h_1}\cdot 4R^2}{Z_0} \leq \frac{M + \eps\cdot \frac{\vol (\+B_{2R})}{\vol (\+B_{3R})}\cdot 4R^2}{Z_0} \leq \frac{3M}{Z_0}\leq 6M,
    $$
    where the last inequality is due to \Cref{lem:Z_0} and the fact that $\eps<\frac{1}{200}$.
    % \htodo{We may need $\eps<\frac{1}{200}$.}

    For the smoothness, we only need to check $\|\grad^2 f_0(x)\|$ for those $x$ with $\|x\|\in (\frac{R}{4},\frac{R}{2}]$ and $\|x\|\in (R,2R]$ since clearly $f_0\in C^2(\bb R^d)$. %\ctodo{need $f_0\in C^2$ here.}

    First, for $\|x\|\in (\frac{R}{4},\frac{R}{2}]$,
    \[
        \grad f_0(x) = \grad \mathfrak{g}_{[\frac{R}{4},\frac{R}{2}]}(x) \cdot (h_1 - h_0(x) ) + (1-\mathfrak{g}_{[\frac{R}{4},\frac{R}{2}]}(x))\cdot \grad h_0(x), 
    \]
    and
    \begin{align*}
        \grad^2 f_0(x) &= \underbrace{\grad^2 \mathfrak{g}_{[\frac{R}{4},\frac{R}{2}]}(x) \cdot (h_1 - h_0(x))}_{\mbox{(a)}} - \underbrace{\tp{\grad \mathfrak{g}_{[\frac{R}{4},\frac{R}{2}]}(x) \cdot \grad h_0(x)^{\top} +  \grad h_0(x)\cdot \grad \mathfrak{g}_{[\frac{R}{4},\frac{R}{2}]}(x)^{\top}}}_{\mbox{(b)}}\\
        &\quad + \underbrace{ \tp{1-\mathfrak{g}_{[\frac{R}{4},\frac{R}{2}]}(x)}\cdot \grad^2 h_0(x)}_{\mbox{(c)}}.
        \end{align*}
    Recall that $LM\geq d$, so it is easy to get $0\preceq \mbox{(c)}\preceq L\cdot \!{Id}_d$. Since 
    \[
        \grad \mathfrak{g}_{[\frac{R}{4},\frac{R}{2}]}(x) = \frac{2x}{\frac{R^2}{4} - \frac{R^2}{16}} \cdot q_{\!{mol}}'\tp{\frac{ \|x\|^2- \frac{R^2}{16}}{\frac{R^2}{4} - \frac{R^2}{16}}},
    \]
    we have
    \[
        \mbox{(b)} = \frac{4d \cdot xx^{\top}}{M\tp{\frac{R^2}{4} - \frac{R^2}{16}}} \cdot q_{\!{mol}}'\tp{\frac{ \|x\|^2- \frac{R^2}{16}}{\frac{R^2}{4} - \frac{R^2}{16}}}.
    \]
    % \mn{For $x,y\in \bb R^d$, we can show that $\|x\|\|y\|\!{Id}_d - xy^T\succeq 0$: for any $z\in \bb R^d$,
    % \begin{align*}
    %     &\phantom{{}={}}z^T\tp{\|x\|\|y\|\!{Id}_d - xy^T}z \\
    %     & =\|x\|\|y\|\|z\|^2 - (z^Tx)(y^Tz)\\
    %     &\geq \|x\|\|y\|\|z\|^2 - \|x\|\|y\|\|z\|^2\\
    %     &= 0.
    % \end{align*}}
    Therefore, $-\+O(L)\cdot\!{Id}_d \preceq \mbox{(b)} \preceq \+O(L)\cdot\!{Id}_d$.
    By direct calculation,
    \[
        \grad^2 \mathfrak{g}_{[\frac{R}{4},\frac{R}{2}]}(x) = \frac{4xx^{\top}}{\tp{\frac{R^2}{4} - \frac{R^2}{16}}^2} \cdot q''_{\!{mol}}\tp{\frac{ \|x\|^2- \frac{R^2}{16}}{\frac{R^2}{4} - \frac{R^2}{16}}} + \frac{2\!{Id}_d}{\frac{R^2}{4} - \frac{R^2}{16}} \cdot q'_{\!{mol}}\tp{\frac{ \|x\|^2- \frac{R^2}{16}}{\frac{R^2}{4} - \frac{R^2}{16}}}
    \]
    and 
    \begin{align}
        \abs{h_1 - h_0(x)} &= \abs{\log \frac{1}{\eps} + \frac{d}{2}\log\tp{9\pi R^2} - \log \Gamma\tp{\frac{d}{2}+1} - \frac{d\|x\|^2}{2M} - \frac{d}{2}\log \frac{2\pi M}{d}} \notag\\
        &= \abs{\log \frac{1}{\eps} + \frac{d}{2}\log\frac{9d}{2\eps} - \log \Gamma\tp{\frac{d}{2}+1} - \frac{d\|x\|^2}{2M}} \label{eq:1}
        % &\leq \log \frac{1}{\eps} + \frac{d}{2}\log\tp{9\pi R^2} - \log \Gamma\tp{\frac{d}{2}+1} + \frac{R^2 d}{8M} + \frac{d}{2}\log \frac{2\pi M}{d}.
    \end{align}
    % \htodo{Here we need $R^2=\frac{M}{\eps}\log \frac{M}{\eps}$?}
    From Stirling's formula, we know that for any $d>0$,
    \[
        \log\sqrt{\pi d} + \frac{d}{2}\log \frac{d}{2e} \leq \log \Gamma\tp{\frac{d}{2}+1}\leq \log\sqrt{\pi d} + \frac{d}{2}\log \frac{d}{2e} + 1.
    \]
    % \[
    %     \log \Gamma\tp{\frac{d}{2}+1} \leq \begin{cases}
    %         \log\sqrt{\pi d} + \frac{d}{2}\log \frac{d}{2e} + 1, & d \mbox{ is even}\\
    %         \log \sqrt{\frac{2\pi d}{d-1}} + \frac{d+1}{2}\log \frac{d}{2e} + 1, & d \mbox{ is odd}
    %     \end{cases},
    % \]
    % and 
    % \[
    %     \log \Gamma\tp{\frac{d}{2}+1} \geq \begin{cases}
    %         \log\sqrt{\pi d} + \frac{d}{2}\log \frac{d}{2e} , & d \mbox{ is even}\\
    %         \log \sqrt{\frac{2\pi d}{d-1}} + \frac{d+1}{2}\log \frac{d}{2e} , & d \mbox{ is odd}
    %     \end{cases}.
    % \]
    Back to \Cref{eq:1}, we have
    \begin{align*}
        \abs{h_1 - h_0(x)}&\leq \abs{\log \frac{1}{\eps} + \frac{d}{2}\log\frac{9e}{\eps} - \frac{d\|x\|^2}{2M} -\log\sqrt{\pi d}} + 1 \\
        &\leq \log \frac{1}{\eps} + \frac{d}{2}\log\frac{9e}{\eps} + \frac{dR^2}{8M} + \log\sqrt{\pi d} +1.
    \end{align*}
    Since $LM\geq d$, we have $-\+O(L)\cdot\!{Id}_d \preceq \mbox{(a)}\preceq \+O(L)\cdot\!{Id}_d$.
    
    For $\|x\|\in (R,2R]$, 
    \[
        \grad^2 f_0(x) = \grad^2 \mathfrak{g}_{[R,2R]}(x)(h_0(x)-h_1) + 2\grad \mathfrak{g}_{[R,2R]}(x) \grad h_0(x)^{\top} + \mathfrak{g}_{[R,2R]}(x) \cdot \grad^2 h_0(x).
    \]
    The remaining calculations are similar. %\ctodo{Maybe say more here.}
\end{proof}

\subsection{Perturb the base instance}\label{sec:hardinstance}
We then construct instances via perturbing $\mu_0$. Let $r_1= \sqrt{\frac{d}{L}\log \frac{LM}{d\eps}}$, $r_2=\sqrt{2}r_1$. Let $h_2\defeq h_1 - \gamma$, where $\gamma$ is a value to be determined later.
% \htodo{Assume the value of $\eps,L,M,d$ satisfy $4r_2\leq R$ and $r_1=\Omega(1)$.}

Note that when $\eps<1/200$, we have $4r_2\leq R$. For a point $v\in \bb R^d$ with $\|v\|=\frac{3R}{4}$, let $\mathfrak{g}_v(x) = q_{\!{mol}}\tp{\frac{\|x-v\|^2-r_1^2}{r_2^2-r_1^2}}$ and $f_v(x)=\mathfrak{g}_v(x)f_0(x)+(1-\mathfrak{g}_v(x))h_2$. This means that, outside the ball $\+B_{r_2}(v)$, $f_v\equiv f_0$, and inside the ball $\+B_{r_1}(v)$, $f_v\equiv h_2$. Define density of the distribution $\mu_v$ over $\bb R^d$ as $p_{\mu_v}\propto e^{-f_v}$. Let $Z_v = \int_{\bb R^d} e^{-f_v(x)} \dd x$.

\begin{lemma}\label{lem:gamma}
    There exists a $\gamma>0$ such that the following holds at the same time:
    \begin{itemize}
        \item  $\int_{\+B_{r_2}(v)} \tp{e^{-f_v(x)} - e^{-h_1}}  \dd x = 9\eps$;
        %\item $9 \tp{\frac{3R}{r_2}}^d \leq e^{\gamma}\leq 18\tp{\frac{3R}{r_1}}^d $;
        % \item $\eps \leq \DTV(\mu_0,\mu_v)\leq 10\eps$;
        \item $\frac{9\eps e^{h_1}}{\!{vol}(\+B_{r_2})}\le e^\gamma-1 \le \frac{9\eps e^{h_1}}{\!{vol}(\+B_{r_1})}$.
        \item $Z_0\leq Z_v\leq 1+10\eps$.
    \end{itemize}
    % There exists a $\gamma$ with $\eps\abs{h_1-\gamma} \leq ()$, such that $\eps \leq \DTV(\mu_0,\mu_v)\leq 10\eps$ for each $\|v\| = \frac{3R}{4}$. This $\gamma$ also satisfies that $Z_0\leq Z_v\leq 1+10\eps$
\end{lemma}

Before proving the lemma, let us examine the information it brings. Recall that $f_0(x)\equiv h_1$ when $x\in \+B_{r_2}(v)$ and we would like to perturb $h_1$ by amount of $\gamma$ to obtain $f_v$ so that there will be $\Theta(\eps)$ more probability mass in $\+B_{r_2}(v)$. For fixed $r_1$ and $r_2$, the lemma says that the order of $\gamma$ is roughly proportional to $h_1$. 

\begin{proof}[Proof of \Cref{lem:gamma}]
    Consider the value $\int_{\+B_{r_2}(v)} \tp{e^{-f_v(x)} - e^{-h_1}}\dd x$. It is continuous and increasing in $\gamma$ when $\gamma\ge 0$. When $\gamma = 0$, $\int_{\+B_{r_2}(v)} \tp{e^{-f_v(x)} - e^{-h_1}}\dd x = 0$. When $\gamma \to \infty $, this value goes to $\infty$. So we can find a $\gamma$ such that $\int_{\+B_{r_2}(v)} \tp{e^{-f_v(x)} - e^{-h_1}}\dd x = 9\eps$ holds exactly.

    For such a $\gamma$, we have
    \[
        Z_v = \int_{\bb R^d} e^{-f_v(x)} \dd x \leq \int_{\bb R^d} e^{-f_0(x)} \dd x + \int_{\+B_{r_2}(v)} \tp{ e^{-f_v(x)} - e^{-h_1}}\dd x = Z_0+9\eps \leq 1+10\eps.
    \]
    Also
    $$
        Z_v = \int_{\bb R^d} e^{-f_v(x)} \dd x \geq \int_{\bb R^d} e^{-f_0(x)} \dd x = Z_0.
    $$
    
    % Then we caculate $\DTV(\mu_0,\mu_v)$. On one hand, since $\+B(v,r_1)\subseteq \+B(v,r_2)\subseteq \set{x\in \bb R^d:\ \|x\| \in (\frac{R}{2},R]}$,
    % \begin{align*}
    %     \DTV(\mu_0,\mu_v) &= \frac{1}{2}\int_{\bb R^d} \abs{\frac{e^{-f_0(x)}}{Z_0} - \frac{e^{-f_v(x)}}{Z_v}} \dd x \\ 
    %     &\geq \frac{1}{2}\int_{\+B(v,r_2)} \abs{\frac{e^{-f_0(x)}}{Z_0} - \frac{e^{-f_v(x)}}{Z_v}} \dd x \\
    %     &=  \frac{1}{2}\int_{\+B(v,r_2)}  \abs{\frac{e^{-f_0(x)}}{Z_v} - \frac{e^{-f_v(x)}}{Z_v}} - \abs{\frac{e^{-f_0(x)}}{Z_0} -  \frac{e^{-f_0(x)}}{Z_v}}  \dd x  \\
    %     &= \frac{1}{2}\tp{\frac{9\eps}{Z_v} - e^{-h_1}\cdot \abs{\frac{1}{Z_0} - \frac{1}{Z_v}} \cdot \vol\tp{\+B(v,r_2)}} \\
    %     &\geq \eps.
    % \end{align*}

    % On the other hand, \begin{align*}
    %     \DTV(\mu_0,\mu_v) &= \frac{1}{2}\int_{\bb R^d} \abs{\frac{e^{-f_0(x)}}{Z_0} - \frac{e^{-f_v(x)}}{Z_v}} \dd x \\ 
    %     &\leq \frac{1}{2}\int_{\bb R^d }  \abs{\frac{e^{-f_0(x)}}{Z_v} - \frac{e^{-f_v(x)}}{Z_v}} + \abs{\frac{e^{-f_0(x)}}{Z_0} -  \frac{e^{-f_0(x)}}{Z_v}}  \dd x \\
    %     &\leq \frac{1}{2}\tp{\frac{9\eps}{Z_v} + Z_0\cdot \abs{\frac{1}{Z_0} - \frac{1}{Z_v}}}\\
    %     &\leq 10\eps.
    % \end{align*}
    
    It remains to calculate $e^\gamma$. We have that
    \[
        \vol\tp{\+B_{r_1}(v)} \cdot e^{-h_1} \tp{e^\gamma - 1}\leq \int_{\+B_{r_2}(v)} \tp{e^{-f_v(x)} - e^{-h_1}}  \dd x =9\eps \leq \vol\tp{\+B_{r_2}(v)} \cdot e^{-h_1} \tp{e^\gamma - 1}.
    \]
\end{proof}

\begin{corollary}\label{cor:gamma-bound}
    For our choice of $h_1$, $r_1$ and $r_2$, it holds that
    \[
        9\tp{\frac{3R}{r_2}}^d \leq e^{\gamma}\leq 18\tp{\frac{3R}{r_1}}^d.
    \]
\end{corollary}
\begin{proof}
    Recall that $e^{h_1}=\eps^{-1}\vol\tp{\+B_{3R}}$ and $\vol\tp{\+B_r} = \frac{\tp{\pi R^2}^{\frac{d}{2}}}{\Gamma\tp{\frac{d}{2}+1}}$. We have
    \[
        e^\gamma-1 \leq  9\tp{\frac{3R}{r_1}}^d \mbox{ and } e^\gamma-1 \geq 9 \tp{\frac{3R}{r_2}}^d.
    \]
    Therefore
    \[
        18\tp{\frac{3R}{r_1}}^d \geq e^\gamma\geq 9 \tp{\frac{3R}{r_2}}^d.
    \]
\end{proof}

\subsection{Properties of the perturbed distributions}
For every $v$ with $\|v\|= \frac{3R}{4}$, we first analyze the smoothness and second moment of the distribution $\mu_v$.

\begin{lemma}\label{lem:moment}
    For $\|v\|=\frac{3R}{4}$, $\E[\mu_v]{\|X\|^2}= \+O\tp{M}$.
\end{lemma}
\begin{proof}
    Direct calculation gives
    \begin{align*}
        \E[\mu_v]{\|X\|^2} & \leq \frac{\E[X\sim \+N\tp{0,\frac{M}{d}\cdot \!{Id}_d}]{\|X\|^2} + \vol (B_{2R})\cdot e^{-h_1}\cdot 4R^2 + \int_{\+B_{r_2}(v)} \tp{e^{-f_v(x)} - e^{-h_1}} \|X\|^2 \dd x}{Z_v} \\
        &\leq \frac{3M + R^2\cdot 9\eps}{Z_v} \leq 24 M.
    \end{align*}
    % \htodo{Here only $\tilde O(M)$. It seems that we cannot guarantee the $O(L)$-smooth of $f_0$ and $O(M)$ second moment of $\mu_v$. Or equivalently, we can choose $M=\frac{M_0}{\log \frac{M_0}{\eps}}$. We require $LM_0\geq d\log\frac{M_0}{\eps}$.}
\end{proof}

\begin{lemma}\label{smooth}
  For $\|v\|=\frac{3R}{4}$, the function $f_v$ is $\+O(L)$-smooth. 
\end{lemma}
\begin{proof}
    We only need to consider those $x\in \+B_{r_2}(v)\setminus \+B_{r_1}(v)$. For such $x$, $f_0(x)=h_1$. Therefore,
    \[
        \grad^2 f_v(x) = \gamma\cdot \grad^2 g_v(x).
    \]
    Note that
    \[
        \grad g_v(x) = \frac{2(x-v)}{r_2^2-r_1^2}\cdot q_{\!{mol}}'\tp{\frac{\|x-v\|^2-r_1^2}{r_2^2-r_1^2}}
    \]
    and
    \[
        \grad^2 g_v(x) = \frac{4(x-v)(x-v)^{\top}}{\tp{r_2^2-r_1^2}^2}\cdot q''_{\!{mol}}\tp{y=\frac{\|x-v\|^2-r_1^2}{r_2^2-r_1^2}} + \frac{2\!{Id}_d}{r_2^2-r_1^2}\cdot q'_{\!{mol}}\tp{\frac{\|x-v\|^2-r_1^2}{r_2^2-r_1^2}}.
    \]
    From \Cref{cor:gamma-bound}, 
    \[
        \gamma \leq \log 18 + \frac{d}{2}\log\frac{9LM}{\eps d} - \frac{d}{2}\log\log \frac{LM}{d\eps}.
    \]
    Therefore, $-\+O(L)\cdot\!{Id}_d \preceq \grad^2 f_v(x) \preceq \+O(L)\cdot\!{Id}_d $.
\end{proof}

We remark that the constants hidden in the $\+O(\cdot)$ in the above two lemmas are universal constants and do not depend on $d$ and $\eps$.

\begin{lemma}\label{lem:TV}
    For $u,v\in \bb R^d$ such that $\|v\|=\|u\|=\frac{3R}{4}$ and $\+B_{r_2}(u)\cap \+B_{r_2}(v)=\emptyset$, $\DTV\tp{\mu_u,\mu_v}> 4\eps$.
\end{lemma}
\begin{proof}
    By the definition of total variation distance, 
    \begin{align*}
        \DTV(\mu_u,\mu_v) & = \frac{1}{2}\int_{\bb R^d} \abs{\frac{e^{-f_u(x)}}{Z_u} - \frac{e^{-f_v(x)}}{Z_v}} \dd x \\
        \mr{$Z_u=Z_v$} & =\frac{1}{2Z_v} \tp{\int_{\+B_{r_2}(u)}\abs{e^{-f_u(x)} - e^{-h_1}} \dd x + \int_{\+B_{r_2}(v)}\abs{e^{-f_v(x) }- e^{-h_1}} \dd x} \\
        &= \frac{9\eps}{Z_v} > 4\eps.
    \end{align*}
\end{proof}

\subsection{The number of disjoint $\+B_{r_2}(v)$'s}

\begin{lemma}\label{lem:disjointcap}
    Suppose $d\geq 5$. There exist $n=\frac{(d-1)\sqrt{d-2}}{2}\cdot \tp{\frac{3R}{8\sqrt{2}r_2}}^{d-1}$ vectors $v_1,v_2,\dots,v_n \in \bb R^d$ such that
    \begin{itemize}
        \item for each $i\in[n]$, $\|v_i\| = \frac{3R}{4}$;
        \item for each $i,j\in [n]$, if $i\ne j$, then $\+B_{r_2}(v_i)\cap \+B_{r_2}(v_j)=\emptyset$.
    \end{itemize}
\end{lemma}
\begin{proof}
    Let $S$ be the sphere $\set{x\in \bb R^d: \|x\| = \frac{3R}{4}}$. For two vectors $x,y\in \bb R^d$, let $\theta(x,y)$ represent the angle between $x$ and $y$. We first try to find $n$ disjoint caps $C_1,C_2,\dots,C_n$ on $S$. Denoting $v_i$ as the central vector of cap $C_i$, $C_i=\set{x\in \bb R^d: \|x\|=\frac{3R}{4}, \cos(\theta(x,v_i))\geq \ell}$ with $\ell=\frac{\sqrt{(\tp{\frac{3R}{4}}^2-2r_2^2)}}{\frac{3R}{4}}$. 
    
    In contrast, suppose we can only find $n'<\frac{(d-2)^{\frac{3}{2}}}{2}\cdot \tp{\frac{3R}{8\sqrt{2}r_2}}^{d-1}$ such disjoint caps $\set{C_i}_{1\leq i\leq n'}$ with central vectors $\set{v_i}_{1\leq i\leq n'}$. Then for any $w\in S$, there exist $i\in[n']$ and $x\in S$ such that $\cos(\theta(x,v_i))\geq \ell$ and $\cos(\theta(x,w))\geq \ell$. Otherwise we can find a new cap with center $w$.
    
    Via \Cref{lem:cos}, we know that $\cos(\theta(w,v_i))\geq \ell^2 - \tp{1-\ell^2}  \ell'$ with $\ell'=\frac{\tp{\frac{3R}{4}}^2-4r_2^2}{\tp{\frac{3R}{4}}^2}$. 
    This means we can find $n'$ larger caps with central vectors $\set{v_i}_{1\leq i\leq n'}$ and angle $\arccos(\ell')$ to cover the sphere.
    From \cite{L11}, however, the area of a cap with angle $\theta = \arccos(\ell')$ is $\frac{\Gamma\tp{\frac{d-1}{2}}}{\sqrt{\pi}\Gamma\tp{\frac{d}{2}}} \int_{0}^{\theta} \sin^{d-2}(\phi)\d \phi$ times of the total sphere. We have
    % \ctodo{Find the constant $c$.}
    \[
         \int_{0}^{\theta} \sin^{d-2}(\phi)\d \phi \leq \frac{1}{\ell'} \int_{0}^{\theta} \sin^{d-2}(\phi)\cos (\phi)\d \phi = \frac{1}{\ell'} \int_0^{\sin(\theta)} s^{d-2} \d s = \frac{1}{\ell'} \frac{\tp{\sin \theta}^{d-1}}{d-1}.
    \]
    Since $\sin \theta = \sqrt{1-\tp{\ell'}^2}\leq \frac{2\sqrt{2}r_2}{\frac{3R}{4}}$, the ratio between the area of this larger cap and the sphere can be bounded by
    \begin{align*}
        \frac{\Gamma\tp{\frac{d-1}{2}}}{\sqrt{\pi}\Gamma\tp{\frac{d}{2}}} \int_{0}^{\theta} \sin^{d-2}(\phi)\d \phi & \leq \frac{1}{\ell' (d-1)}\cdot \frac{\Gamma\tp{\frac{d-1}{2}}}{\sqrt{\pi}\Gamma\tp{\frac{d}{2}}}\cdot \tp{\frac{2\sqrt{2}r_2}{\frac{3R}{4}}}^{d-1} \\
        \mr{$R\geq 4r_2$}
        &\leq \frac{9}{5\sqrt{\pi}(d-1)} \cdot \frac{\Gamma\tp{\frac{d-1}{2}}}{\Gamma\tp{\frac{d}{2}}} \cdot \tp{\frac{8\sqrt{2}r_2}{3R}}^{d-1} \\
        \mr{Gautschi's inequality}
        &\leq \frac{9\sqrt{2}}{5\sqrt{\pi}}\cdot \frac{1}{(d-1)\sqrt{d-2}}\cdot \tp{\frac{8\sqrt{2}r_2}{3R}}^{d-1}\\
        &\leq \frac{2}{(d-1)\sqrt{d-2}} \tp{\frac{8\sqrt{2}r_2}{3R}}^{d-1}.
    \end{align*}
    
    % So there exists some universal constant $c'$ such that the ratio between the area of this larger cap and the sphere is no larger than $\tp{\frac{c' r_2}{R}}^{d-1}$. 
    % By choosing $c= \frac{1}{c'}$, 
    This will lead to a conflict since $n'$ such caps cannot cover the sphere. Therefore, we can find $n= \frac{(d-1)\sqrt{d-2}}{2}\cdot \tp{\frac{3R}{8\sqrt{2}r_2}}^{d-1}$ such $C_i$'s.
    
    Furthermore, from \Cref{lem:cosinBall}, $\+B_{r_2}(v_i)\cap \+B_{r_2}(v_j)=\emptyset$.
\end{proof}

\subsection{Proof of the lower bound}
\begin{theorem}[\Cref{thm:main} restated]
    % There exist a universal constant $C>0$ such that 
    For any $L,M>0$ satisfying $LM\ge d$ and for any $\eps\in(0,1/200)$, $d\geq 5$, if a sampling algorithm $\+A$ always terminates within 
    \[
        \frac{\eps (d-2)^{\frac{3}{2}}}{8}\cdot \tp{\frac{9}{256} \cdot \frac{LM}{d\eps} \cdot \frac{1}{\log \frac{LM}{d\eps}}}^{\frac{d-1}{2}}
        % \frac{\eps}{4}\cdot \tp{\frac{C\cdot LM}{2 d\eps} \cdot \frac{1}{\log \frac{LM}{d\eps}}}^{\frac{d-1}{2}}
        % % \approx \frac{\eps}{4}\exp\set{\frac{d}{2}\cdot \Omega\tp{\log \frac{LM}{d\eps}}}
    \]
    queries on every input instance in $\@D_{L,M}$, then there must exist some distribution $\mu\in \@D_{L,M}$ such that when the underlying instance is $\mu$, the distribution of $\+A$'s output, denoted as $\tilde \mu$, is $\eps$ away from $\mu$ in total variation distance, i.e., $\DTV(\mu,\tilde \mu)\geq \eps$.
\end{theorem}
\begin{proof}
    Let $v_1,v_2,\dots,v_n$ be the $n$ vectors in \Cref{lem:disjointcap}. For each $i\in [n]$, construct a distribution $\mu_{v_i}$ with density $p_i \propto e^{-f_{v_i}}$ as described in \Cref{sec:hardinstance}. From the discussion in previous sections, we can assume that every $\mu_{v_i}$ as well as the base instance are $L$-log-smooth and have second moment at most $M$. For simplicity, we write $\mu_{v_i}$ as $\mu_i$.

    We use $\Pr[\mu]{\cdot}$ and $\E[\mu]{\cdot}$ to denote the probability and expectation when the underlying instance is some distribution $\mu$. Let $\+E_{k,i}$ be the event that the algorithm $\+A$ queries a value in zone $\+B_{r_2}(v_i)$ in the $k$-th query and it is the first time that $\+A$ queries the points in $\+B_{r_2}(v_i)$. If the algorithm terminates before the $k$-th query, we regard $\+E_{k,i}$ as an impossible event.

    Assume $\+A$ always terminates in $N$ queries for some 
    $$
        N< \frac{\eps (d-2)^{\frac{3}{2}}}{8}\cdot \tp{\frac{9}{256}\cdot \frac{LM}{d\eps} \cdot \frac{1}{\log \frac{LM}{d\eps}}}^{\frac{d-1}{2}} < \frac{\eps}{4}\cdot \frac{(d-1)\sqrt{d-2}}{2}\cdot \tp{\frac{3R}{8\sqrt{2}r_2}}^{d-1} = \frac{\eps n}{4}.
    $$ 
    Let $\+E_{i}$ be the event that $\+A$ queries the points in $\+B_{r_2}(v_i)$ at least once. Then
    \begin{align*}
        \sum_{i=1}^n \Pr[\mu_0]{\+E_i}& = \sum_{i=1}^n \sum_{k=1}^N \Pr[\mu_0]{\+E_{k,i}} = \sum_{k=1}^N \sum_{i=1}^n \Pr[\mu_0]{\+E_{k,i}} \leq N<\frac{\eps n}{4}.
    \end{align*}
    So there exists some $i_0,j_0\in[n]$ and $i_0\neq j_0$ such that $\Pr[\mu_0]{\+E_{i_0}}<\frac{\eps}{3}$ and $\Pr[\mu_0]{\+E_{j_0}}<\frac{\eps}{3}$. Otherwise $\sum_{i=1}^n \Pr[\mu_0]{\+E_i}\geq \frac{\eps(n-1)}{3}\geq \frac{\eps n}{4}$ since $n\geq 4$ when $d\geq 5$. From union bound, $\Pr[\mu_0]{\ol{\+E_{i_0}}\cap \ol{\+E_{j_0}}}> 1-\frac{2\eps}{3}$. 
    
    We know that on $\bb R^d\setminus \+B_{r_2}(v_i)$, $f_{v_i}(x)=f_0(x)$ for each $i\in[n]$. Therefore, via coupling arguments,
    \[
        \Pr[\mu_{i_0}]{\ol{\+E_{i_0}}\cap \ol{\+E_{j_0}}} = \Pr[\mu_{j_0}]{\ol{\+E_{i_0}}\cap \ol{\+E_{j_0}}} = \Pr[\mu_0]{\ol{\+E_{i_0}}\cap \ol{\+E_{j_0}}}> 1-\frac{2\eps}{3}.
    \]

    Let $\+E=\+E_{i_0}\cup \+E_{j_0}$. Let $\tilde \mu_0^{\ol{\+E}}$ be the output distribution of $\+A$ with input distribution $\mu_0$ when $\ol{\+E}$ happens. 
    Since
    \begin{align*}
        4\eps\leq \DTV(\mu_{i_0},\mu_{j_0}) \leq \DTV(\mu_{i_0},\tilde \mu_0^{\ol{\+E}}) + \DTV(\mu_{j_0},\tilde \mu_0^{\ol{\+E}}),
    \end{align*}
    we have either $\DTV(\mu_{i_0},\tilde \mu_0^{\ol{\+E}})>2\eps$ or $\DTV(\mu_{j_0},\tilde \mu_0^{\ol{\+E}})>2\eps$. W.l.o.g., assume $\DTV(\mu_{i_0},\tilde \mu_0^{\ol{\+E}})>2\eps$. Assume the output distribution of $\+A$ is $\tilde \mu_{i_0}$ when the input is $\mu_{i_0}$. Let $\tilde \mu_{i_0}^{\+E}$ and $\tilde \mu_{i_0}^{\ol{\+E}}$ be $\tilde \mu_{i_0}$ conditioned on $\+E$ and $\ol{\+E}$ respectively and denote their density functions as $\tilde p_{i_0}^{\+E}$, $\tilde p_{i_0}^{\ol{\+E}}$ and $\tilde p_{i_0}$. Then $\tilde \mu_{i_0}^{\ol{\+E}}=\tilde \mu_0^{\ol{\+E}}$. 
    Then we have
    \begin{align*}
        \DTV(\mu_{i_0},\tilde \mu_{i_0}) &= \frac{1}{2}\int_{\bb R^d} \abs{p_{i_0}(x) - \tilde p_{i_0} (x)} \dd x\\
        &=\frac{1}{2}\int_{\bb R^d} \abs{p_{i_0}(x) - \Pr[\mu_{i_0}]{\+E}\tilde p_{i_0}^{\+E}(x) - \Pr[\mu_{i_0}]{\ol{\+E}}\tilde p_{i_0}^{\ol{\+E}}(x)} \dd x \\
        &\geq \Pr[\mu_{i_0}]{\ol{\+E}}\cdot \frac{1}{2} \int_{\bb R^d}\abs{p_{i_0}(x) - \tilde p_{i_0}^{\ol{\+E}}(x)} \dd x  - \Pr[\mu_{i_0}]{\+E}\cdot \frac{1}{2} \int_{\bb R^d}\abs{p_{i_0}(x) - \tilde p_{i_0}^{\+E}(x)} \dd x \\
        & = \Pr[\mu_{i_0}]{\ol{\+E}}\cdot \DTV(\mu_{i_0}, \tilde \mu_{i_0}^{\ol{\+E}}) - \Pr[\mu_{i_0}]{\+E}\cdot \DTV(\mu_{i_0}, \tilde \mu_{i_0}^{\+E}) \\
        & = \Pr[\mu_{i_0}]{\ol{\+E}}\cdot \DTV(\mu_{i_0}, \tilde \mu^{\ol{\+E}}_0) - \Pr[\mu_{i_0}]{\+E}\cdot \DTV(\mu_{i_0}, \tilde \mu_{i_0}^{\+E}) \\
        &\geq \tp{1-\frac{2\eps}{3}}\cdot 2\eps - \frac{2\eps}{3}\cdot 1\\
        &>\eps.
    \end{align*}
    This means, on instance $\mu_{i_0}$, the algorithm $\+A$ will fail to output a distribution which is $\eps$-close to $\mu_{i_0}$ in total variation distance.
\end{proof}

%% file: upper-bound.tex
\section{The upper bound}\label{sec:ub}

In this section we prove \Cref{thm:main-ub}. We construct a distribution $\pi$  that satisfies $\DTV(\mu,\pi)\le \frac{\eps}{2}$, has smoothness close to that of $\mu$, is of bounded moment, and whose \Poincare constant is at least $\approx \tp{\frac{LM}{d\eps}}^{-O(d)}$. Then we call known Langevin-based algorithm to sample from $\pi$, whose sample complexity is directly related to the \Poincare constant. Our strategy for the construction of $\pi$ is as follows.

\begin{itemize}
    \item First, observe the following comparison result. Let $p_1(x)$ and $p_2(x)$ be the densities of two distributions supported on $\bb R^n$. If $1/C\le \frac{p_1(x)}{p_2(x)}\le C$ for every $x\in \bb R^n$, then the ratio of their \Poincare constants is at least $C^{-\+O(1)}$. Therefore, we only need to construct a distribution $\pi$ with appropriate smoothness, whose density is pointwise close to a suitable Gaussian. The range of ratios in the densities that we can tolerate is of the order $\tp{\frac{LM}{d\eps}}^{\+O(d)}$.
    \item Clearly the density $p_\mu(x)\propto e^{-f_\mu(x)}$ of $\mu$ does not satisfy our requirement due to the possible existence of certain regions with extremely small probability. The value of $f_\mu$ may be very large in these regions. On the other hand, the measure of these regions under $\mu$ is small, so we can \emph{truncate} $f_\mu$ appropriately to ensure its value is well upper bounded without affecting the measure $\mu$ much. We then use the truncated function to define the distribution $\pi$.
    \item In order to truncate $f_\mu$ appropriately, we need to estimate its minimum value $f^*\defeq \min_{x\in\bb R^d} f_\mu(x)$ and the partition function $Z_\mu \defeq \int_{\bb R^n} \exp\tp{-f_\mu(x)} \d x$ within a certain accuracy. To this end, we divide a compact set containing most mass of $\mu$ into cubes and approximate $\mu$ in each cube respectively using queries to $f_\mu$.
\end{itemize}
%\htodo{We only use queries to $f_\mu$, no $\grad f_{\mu}$ in the estimation?}

We will give the construction of $\pi$ in \Cref{sec:construction-of-pi} and prove its properties in \Cref{sec:properties-of-pi}. Then we show how to estimate the key parameters in our construction in \Cref{sec:estimate-of-pi}. Finally, we combine everything and prove \Cref{thm:main-ub} in \Cref{sec:proof-of-ub}.

\subsection{The construction of $\pi$}\label{sec:construction-of-pi}

The purpose of this section is to construct a distribution $\pi$ whose density function is close to that of $\gamma\sim\+N\tp{0,\frac{M}{\eps d}\!{Id}_d}$ \emph{pointwise} and $\DTV(\mu,\pi)\le \frac{\eps}{2}$. We assume the density of $\pi$ is \emph{proportional to} $\exp\tp{-f_\pi(x)}$. Let $f_\gamma\colon \bb R^d\to\bb R$ be the function $x\mapsto \frac{\eps d\norm{x}^2}{2M}+\frac{d}{2}\log\frac{2\pi M}{d\eps}$. Then the density of $\gamma$ is \emph{equal to} $\exp\tp{-f_\gamma(x)}$.

We use $Z_\pi \defeq \int_{\bb R^d} \exp\tp{-f_\pi(x)}\d x$ and $Z_\mu\defeq \int_{\bb R^d} \exp\tp{-f_\mu(x)} \d x$ to denote the two normalizing factors. Then the density of $\pi$ and $\mu$ are
\[
    p_\pi(x) = \exp\tp{-f_\pi(x)} / Z_\pi \mbox{ and } p_{\mu}(x) = \exp\tp{-f_{\mu}(x)} / Z_\mu
\]
respectively.

Note that the second moment of a random variable $X$ with law $\mu$ is at most $M$. By Markov's inequality, for $R=\sqrt{\frac{32M}{\eps}}$,
\begin{equation}\label{eqn:markov-mu}
    \Pr{\norm{X}^2>R^2} \le \frac{\E{\norm{X}^2}}{R^2}\le \frac{M}{R^2} = \frac{\eps}{32},
\end{equation}
meaning outside a ball of radius $\Theta(R)$, the mass of $\mu$ is $O(\eps)$. We let $f_\pi(x) = f_\gamma(x) - \log\eps$ for $x\in \bb R^d \setminus \+B_{2R}$. We will then construct a function $f_\pi^{\le 2R}$ with support $\+B_{2R}$ and define
\[
    f_\pi(x) = (1-\mathfrak{g}_{[R,2R]}(x)) \cdot f_\pi^{\le 2R}(x) + \mathfrak{g}_{[R,2R]}(x) \cdot \tp{f_\gamma(x)-\log\eps},
\]
where $\mathfrak{g}_{[R,2R]}\defeq q_{\!{mol}}\tp{\frac{\norm{x}^2-R^2}{(2R)^2-R^2}}$ is the smooth function interpolating $f_\pi^{\le 2R}$ and $f_\gamma - \log\eps$ in the region $\norm{x}\in [R,2R]$. 

%Let $Z_\pi\defeq \int_{\bb R^n}$ The construction of $f_\pi^{\le 2R}$ should meet the 
As discussed before, for $x\in \+B_R$, ideally $p_{\pi}(x) = \exp\tp{-f_{\pi}^{\le 2R}(x)} / Z_\pi$ should be close to $p_{\mu}(x) = \exp\tp{-f_\mu(x)} / Z_\mu$ with those points of extremely small probability smoothly truncated. As a result, we first assume that we can find an approximation of $Z_\mu$, denoted as $\wh Z_\mu$. In general calculating a good approximation for $Z_\mu$ is computationally equivalent to sampling from $p_\mu$. However, as our target is a bound for the \Poincare constant of order exponential in $d$, our requirement for the accuracy of the approximation is very loose. We also assume an approximation $\wh f^*$ of the minimum $f^*\defeq \inf_{x\in \+B_{2R}} f_\mu(x)$. In fact, we will prove the following proposition in \Cref{sec:estimate-of-pi}.

\begin{proposition} \label{prop:Z-and-fmin}
    Within $\+O\tp{\frac{LM}{\eps d}}^d$ queries to $f_\mu(x)$, one can find
    \begin{itemize}
        \item a number $\wh Z_\mu$ satisfying $\frac12 e^{-d}\le \frac{\wh Z_\mu}{Z_\mu}\le 1$, and
        \item a number $\wh f^*$ satisfying $f^*\le \wh f^*\le f^*+d$.
    \end{itemize}    
\end{proposition}

Then we turn to truncate the small value of $\exp(-f_\mu(x))$ or equivalently the large value of $f_\mu(x)$ in $\+B_{2R}$. Define two constants
\[
    h_1 \defeq \wh f^* +\log\!{vol}(\+B_{2R}) + \frac{d}{2}\log L+\log\frac{4}{\eps},\; h_2 \defeq h_1+\frac{d}{2}\log\frac{LM}{d\eps}.
\]
We remark that $h_1$ is our threshold for the truncation. The term $\frac{d}{2}\log L$ term is used to guarantee that $\log\!{vol}(\+B_{2R}) + \frac{d}{2}\log L$ is nonnegative and therefore $h_1\ge f^*$. In order to keep the truncated function smooth, we define a \emph{soft threshold} $h_2$ above $h_1$. 

Define the interpolation function $\mathfrak{g}_{[h_1,h2]}(x) \defeq q_{\!{mol}}\tp{\frac{h_2-f_{\mu}(x)}{h_2-h_1}}$. We define
\[
    \ol{f_{\pi}^{\le 2R}}(x)\defeq \mathfrak{g}_{[h_1,h_2]}(x)\cdot f_\mu(x) + (1-\mathfrak{g}_{[h_1,h_2]}(x))\cdot h_2
\]
In other words, for those $x$ with $f_\mu(x)\le h_1$, $\ol{f_{\pi}^{\le 2R}}(x) = f_\mu(x)$; for those $x$ with $f_\mu(x)\ge h_2$, $\ol{f_{\pi}^{\le 2R}}(x) = h_2$; for those $x$ with $f_\mu(x)\in [h_1,h_2]$, the value of $\ol{f_{\pi}^{\le 2R}(x)}$ is smoothly interpolated between $h_1$ and $h_2$. The function $f_\pi^{\le 2R}$ is illustrated in \Cref{fig:ub}. %\ctodo{A figure here.}

Finally, we \emph{approximately normalize} $\exp\tp{-\ol{f_{\pi}^{\le 2R}}(x)}$ into a ``probability'' by dividing our estimate $\wh Z_\mu$, namely that for every $x\in \bb R^d$, let
\[
    f^{\le 2R}_\pi(x) \defeq \ol{f^{\le 2R}_\pi}(x) + \log \wh Z_{\mu}.
\]

\subsection{Properties of $\pi$}\label{sec:properties-of-pi}

In this section we prove some useful properties of the distribution $\pi$ just constructed. We begin with three useful technical lemmas in \Cref{sec:ub-tech}. Then we prove key properties of $\pi$, including showing its closeness to $\mu$ in terms of total variation distance in \Cref{sec:ub-closeness}, bounding its \Poincare constant in \Cref{sec:ub-poincare} and analyzing its smoothness in \Cref{sec:ub-smooth}.

\subsubsection{Technical lemmas}\label{sec:ub-tech}

Recall that $p_\mu(x)\propto \exp\tp{-f_\mu(x)}$ is the density of the $L$-log-smooth distribution $\mu$ and $Z_\mu = \int_{\bb R^d} p_\mu(x)\d x$ is the normalizing factor. The first lemma says that $p_\mu(x)$ has an upper bound since $\mu$ is $L$-log-smooth. 

\begin{lemma}\label{lem:mu-bound}
    $\forall x\in \bb R^d,\;p_\mu(x) \le \tp{\frac{2\pi}{L}}^{-\frac{d}{2}}$.
\end{lemma}
\begin{proof}
    Let $x^* = \argmax_{x\in\bb R^d} \exp\tp{-f_\mu(x)}$. Since $f_\mu$ is $L$-smooth, for each $y\in \bb R^d$,
    \begin{align*}
        f_\mu(y)
        &\le f_\mu(x^*) + \grad f_\mu(x^*)\top (y-x^*) + \frac{L}{2}\norm{y-x^*}^2\\
        \mr{$\grad f_\mu(x^*)=0$}
        &=f_\mu(x^*)+ \frac{L}{2}\norm{y-x^*}^2.
    \end{align*}
    On the other hand, 
    \begin{align*}
        1
        &=Z_\mu^{-1} \int_{\bb R^d} \exp\tp{-f_\mu(y)} \d y\\
        &\ge Z_{\mu}^{-1}\int_{\bb R^d} \exp\tp{-f_\mu(x^*)-\frac{L}{2}\norm{y-x^*}^2} \d y\\
        &= p_\mu(x^*)\cdot\int_{\bb R^d} \exp\tp{-\frac{L}{2}\norm{y-x^*}^2}\d y.
    \end{align*}
    Since $\int_{\bb R^d} \exp\tp{-\frac{L}{2}\norm{y-x^*}^2}\d y = \tp{\frac{2\pi}{L}}^{\frac{d}{2}}$, we conclude the proof. 
\end{proof}

The second lemma shows that the function $\ol{f^{\le 2R}_{\pi}}$, our truncation for $f_\mu$, does not change the mass in the $\+O(R)$-ball much. 

\begin{lemma}\label{lem:Z_R-close-to-one}
    Assume $d\geq 3$. The following holds.
    \[
        Z_{\mu}^{-1}\int_{\+B_{2R}} \exp\tp{-\ol{f^{\le 2R}_{\pi}}(x)}\d x \le 1+\frac{\eps}{32}, \mbox{ and } Z_{\mu}^{-1}\int_{\+B_R} \exp\tp{-\ol{f^{\le 2R}_{\pi}}(x)}\d x \ge 1-\frac{\eps}{16}.
    \]
\end{lemma}
\begin{proof}
    Let $\+L =\set{x\in \+B_{2R}\cmid f_\mu(x)\ge h_1}$ be the set of points in $\+B_{2R}$ where the truncation occurs. Clearly
    \begin{align*}
        Z_{\mu}^{-1}\int_{\+B_{2R}} \exp\tp{-\ol{f^{\le 2R}_\pi}(x)}\d x
        &\leq Z_{\mu}^{-1}\int_{\+B_{2R}\setminus \+L} \exp\tp{-f_\mu(x)}\d x + Z_{\mu}^{-1}\int_{\+L} \exp\tp{-h_1}\d x\\
        &\le Z_{\mu}^{-1}\int_{\bb R^d} \exp\tp{-f_\mu(x)}\d x + Z_\mu^{-1}\cdot \!{vol}(\+B_{2R})\cdot \exp(-h_1)\\
        \mr{by \Cref{lem:mu-bound}, $Z_\mu^{-1}\exp\tp{-f^*} \le \tp{\frac{2\pi}{L}}^{-\frac{d}{2}}$}
        &\le 1+\frac{\eps}{4}\cdot (2\pi)^{-\frac{d}{2}}\\
        \mr{$d\geq 3$} 
        &\le 1+\frac{\eps}{32}.
    \end{align*}
    For the lower bound, we can calculate that
    \begin{align*}
        Z_\mu^{-1}\int_{\+B_R} \exp\tp{-\ol{f^{\le 2R}_\pi}(x)}\d x
        &\ge Z_\mu^{-1}\int_{\+B_R\setminus \+L} \exp\tp{-\ol{f^{\le 2R}_\pi}(x)}\d x\\
        \mr{$\ol{f^{\le 2R}_\pi}(x) = f_\mu(x)$ for $x\in\+B_{2R}\setminus\+L$} 
        &=Z_\mu^{-1}\int_{\+B_R\setminus\+L} \exp\tp{-f_\mu(x)} \d x\\
        &=Z_\mu^{-1}\int_{\+B_R} \exp\tp{-f_\mu(x)} \d x - Z_\mu^{-1}\int_{\+L} \exp\tp{-f_\mu(x)} \d x.
    \end{align*}
    By Markov's inequality,
    \begin{equation}\label{eqn:1st}
        Z_\mu^{-1}\int_{\+B_R} \exp\tp{-f_\mu(x)} \d x = \Pr[X\sim\mu]{X\in \+B_R}\ge 1-\frac{M}{R^2} = 1-\frac{\eps}{32}.
    \end{equation}
    By our definition of $h_1$, 
    \begin{equation}\label{eqn:2nd}
        Z_\mu^{-1}\int_{\+L} \exp\tp{-h_1}\le \!{vol}(\+B_{2R})\cdot Z_\mu^{-1}\exp\tp{-h_1}\le \frac{\eps}{4}\cdot (2\pi)^{-\frac{d}{2}}\le \frac{\eps}{32}.
    \end{equation}
    Combining~\eqref{eqn:1st} and~\eqref{eqn:2nd} finishes the proof.
\end{proof}

Recall that our definition for $f^{\le 2R}_{\pi}$ is an \emph{approximately normalized} $\ol{f^{\le 2R}_{\pi}}$ using our estimate $\wh Z_\mu$ for $Z_\mu$. The following lemma states that provided the estimate is accurate enough, $Z_\pi$ is close to $1$.

\begin{lemma}\label{lem:Zpi-close-to-one}
    Assume $d\geq 3$. It holds that
    \begin{itemize}
        \item $1-\frac{\eps}{16} \le Z_\pi \cdot \frac{\wh Z_\mu}{Z_\mu} \le 1+\frac{\eps}{16}$.
        \item $\frac12\le Z_\pi\le 4e^d$.
    \end{itemize}
\end{lemma}
\begin{proof}
    On the one hand, from \Cref{lem:Z_R-close-to-one},
    \[
        Z_{\pi} \geq \int_{\+B_R} \exp\tp{-f_\pi(x)} \d x = \frac{Z_\mu}{\wh Z_{\mu}} \cdot Z_\mu^{-1}\int_{\+B_R} \exp\tp{-\ol{f^{\le 2R}_\pi}(x)} \d x \geq \frac{Z_\mu}{\wh Z_{\mu}} \cdot \tp{1 - \frac{\eps}{16}}.
    \]
    On the other hand, 
    \begin{align*}
        Z_{\pi} &\leq \int_{\bb R^d\setminus \+B_R} e^{-f_\gamma(x)+\log\eps} \d x + \int_{\+B_{2R}} \exp\tp{-f^{\le 2R}_\pi(x)} \d x\\
        &\le \eps\Pr[X\sim \+N(0,\frac{M}{d\eps}\cdot \!{Id}_d)]{ \|X\|^2 \geq R^2} + \frac{Z_\mu}{\wh Z_{\mu}} \cdot Z_\mu^{-1}\int_{\+B_{2R}} e^{-\ol{f^{\le 2R}_\pi}(x)} \d x \\
        \mr{\Cref{lem:Z_R-close-to-one}}
        &\leq \frac{\eps}{32} + \frac{Z_\mu}{\wh Z_{\mu}} \cdot \tp{1+\frac{\eps}{32}}.
    \end{align*}
    The lemma then follows from \Cref{prop:Z-and-fmin}.
\end{proof} 

\subsubsection{Distance between $\pi$ and $\mu$}\label{sec:ub-closeness}

We now prove that the total variation distance between $\pi$ and $\mu$ is at most $\frac{\eps}{2}$. 

\begin{lemma} \label{lem:pi-mu-close}
    Assume $d\geq 3$. We have $\DTV(\pi,\mu)\le \frac{\eps}{2}$.
\end{lemma}
\begin{proof}
    We still let $\+L = \set{x\in \+B_{2R}\cmid f_{\mu}(x)\ge h_1}$ denote those points that have been truncated. Clearly
    \[
        \DTV(\pi,\mu) = \frac{1}{2}\Big(\underbrace{\int_{\bb R^d\setminus \+B_R}  \abs{p_\pi(x)-p_\mu(x)} \d x}_{\mbox{(a)}} + \underbrace{\int_{ \+B_R\setminus \+L }  \abs{p_\pi(x)-p_\mu(x)} \d x}_{\mbox{(b)}} + \underbrace{\int_{\+L\cap \+B_R}  \abs{p_\pi(x)-p_\mu(x)} \d x}_{\mbox{(c)}}\Big).
    \]
    We then bound terms (a), (b) and (c) respectively. For (a), we have
    \begin{align*}
        \mbox{(a)}
        &\le \int_{\bb R^d\setminus \+B_R} p_\mu(x) \d x + \int_{\bb R^d\setminus \+B_R} p_\pi(x)\d x\\
        &= \Pr[X\sim\mu]{\norm{X}^2>R^2} +\tp{1-Z_\pi^{-1} \int_{\+B_R} \exp\tp{-f^{\le 2R}_\pi(x)} \d x}\\
        \mr{\eqref{eqn:markov-mu} and definition of $f^{\le 2R}_\pi$}
        &\le \frac{\eps}{32} + 1-Z_{\pi}^{-1} \cdot \frac{Z_\mu}{\wh Z_\mu}\cdot Z_\mu^{-1}\cdot \int_{\+B_R} \exp\tp{-\ol{f^{\le 2R}_\pi}(x)} \d x\\
        \mr{\Cref{lem:Z_R-close-to-one}}
        &\le\frac{\eps}{32} + 1-\frac{Z_\mu}{\wh Z_\mu}\cdot Z_\pi^{-1}\tp{1-\frac{\eps}{16}}\\
        \mr{\Cref{lem:Zpi-close-to-one}}
        &\le \frac{\eps}{32}+\frac{\eps}{8} = \frac{5\eps}{32}.
    \end{align*}
    By our construction, for $x\in \+B_R\setminus\+L$, we have that $f_\pi(x) = f_{\mu}(x)+\log\wh Z_\mu$. Therefore, for the term (b), we have
    \begin{align*}
        \mbox{(b)}
        &=\int_{\+B_R\setminus\+L}\abs{Z_\mu^{-1}\exp\tp{-f_{\mu}(x)}-Z_\pi^{-1}\cdot \wh Z_\mu^{-1}\exp\tp{-f_{\mu}(x)}} \d x\\
        &=\abs{1-\frac{Z_\mu}{\wh Z_\mu}\cdot Z_\pi^{-1}} \cdot Z_\mu^{-1}\int_{\+B_R\setminus\+L} \exp\tp{-f_{\mu}(x)} \d x\\
        &\le \abs{1-\frac{Z_\mu}{\wh Z_\mu}\cdot Z_\pi^{-1}}\\
        \mr{\Cref{lem:Zpi-close-to-one}}
        &\le \frac{\eps}{8}.
    \end{align*}
    Finally, for the term (c), we have
    \begin{align*}
        \mbox{(c)}
        &\le \int_{\+L\cap \+B_R} p_\mu(x)\d x + \int_{\+L\cap \+B_R} p_\pi(x) \d x\\
        &\le Z_\mu^{-1}\int_{\+B_R} \exp\tp{-h_1}\d x + Z_\pi^{-1}\int_{\+L\cap \+B_R} \exp\tp{-h_1-\log \wh Z_\mu}\d x\\
        &\le \tp{1+\frac{Z_\mu}{\wh Z_\mu\cdot Z_\pi}}\cdot Z_\mu^{-1}\int_{\+B_R} \exp\tp{-h_1}\d x\\
        \mr{\Cref{lem:Zpi-close-to-one}}
        &\le 3\cdot\!{vol}(\+B_R)\cdot Z_\mu^{-1} \exp\tp{-h_1}\\
        \mr{Definition of $h_1$ and \Cref{lem:mu-bound}}
        &\le\frac{3\eps}{4}\cdot\tp{2\pi}^{-\frac{d}{2}} \le \frac{3\eps}{32}.
    \end{align*}
    In total, we have $\DTV(\pi,\mu) \le \frac{5\eps}{32}+\frac{\eps}{8}+\frac{3\eps}{32}<\frac{\eps}{2}$.
\end{proof}

\subsubsection{The \Poincare constant of $\pi$}\label{sec:ub-poincare}

In this section we bound the \Poincare constant of $\pi$. Recall that the density $p_\gamma$ of $\gamma\sim \+N\tp{0,\frac{M}{\eps d}\!{Id}_d}$ is $p_\gamma(x) = \exp\tp{-f_\gamma(x)}$ where $f_\gamma(x) = \frac{\eps d\norm{x}^2}{2M}+\frac{d}{2}\log\frac{2\pi M}{d\eps}$.  We will show that $p_\pi$ is close to $p_\gamma$ pointwise. 

\begin{lemma}\label{lem:fclose}
    Assume $d\geq 3$. For every $x\in \bb R^d$, $\abs{f_\gamma(x)-f_\pi(x)} \le \+O\tp{d\log\frac{LM}{d\eps}}$.
\end{lemma}
\begin{proof}
    Outside $\+B_{2R}$, we have
    \[
        \abs{f_\pi(x) - f_\gamma(x)} = \log\frac{1}{\eps} = \+O\tp{d\log\frac{LM}{d\eps}}.
    \]
    For $x\in \+B_{2R}$, or equivalently $\norm{x}^2\le \frac{128M}{\eps}$, 
    % $\abs{f_\gamma(x)} = \+O\tp{d\log\frac{LM}{d\eps}}$. 
    \[
        \frac{d}{2}\log\frac{2\pi M}{d\eps} \leq f_\gamma(x) \leq 64d + \frac{d}{2}\log\frac{2\pi M}{d\eps}.
    \]
    
    % It remains to show that $\abs{f_\pi(x)} = \+O\tp{d\log\frac{LM}{d\eps}}$ as well.
    It remains to bound $f_\pi(x)$ inside $\+B_{2R}$. Note that $f_\pi(x) = f^{\le 2R}_\pi(x)$ for $x\in \+B_R$ and $f_\pi(x)$ is an interpolation of $f^{\le 2R}_\pi(x)$ and $f_\gamma(x)-\log\eps$ for $x\in \+B_{2R}\setminus \+B_R$, we only need to bound $f^{\le 2R}_{\pi}$.
    % we only need to verify that $f^{\le 2R}_{\pi} = \+O\tp{d\log\frac{LM}{d\eps}}$ for $x\in \+B_{2R}$. 

    Recall that ${f^{\le 2R}_{\pi}(x)} = {\ol{f^{\le 2R}_{\pi}}(x)+\log \wh Z_\mu}= {\ol{f^{\le 2R}_{\pi}}(x)+\log Z_\mu}+{\log\frac{\wh Z_\mu}{Z_\mu}}$. By \Cref{prop:Z-and-fmin}, $-d -1 \leq \log\frac{\wh Z_\mu}{Z_\mu} \leq 0$.
    % we only need to bound $\abs{\ol{f^{\le 2R}_{\pi}}(x)+\log Z_\mu}$.
    By our construction, for all $x\in \+B_{2R}$, 
    \[
        f^*+\log Z_{\mu} \leq \ol{f^{\le 2R}_{\pi}}(x)+\log Z_{\mu} \leq h_2 + \log Z_{\mu}.
    \]
    From \Cref{lem:mu-bound}, $f^*+\log Z_{\mu} \geq \frac{d}{2}\log\frac{2\pi}{L}$. On the other hand, $h_2 + \log Z_{\mu} \leq f^*+\log Z_{\mu} + d + \log \!{vol}(\+B_{2R}) + \frac{d}{2}\log\frac{L^2M}{d\eps} + \log \frac{4}{\eps}$. Since
    \[
        \!{vol}(\+B_R)\cdot e^{-f^* - \log Z_{\mu}} = \int_{\+B_R} e^{-f^* - \log Z_{\mu}} \dd x \geq \Pr[X\sim \mu]{X\in \+B_R} \geq 1-\frac{\eps}{32}, 
    \]
    we have $f^*+\log Z_{\mu} \leq \log \!{vol}(\+B_R) + 1$. Therefore, 
    \begin{align*}
        \ol{f^{\le 2R}_{\pi}}(x)+\log Z_{\mu} \leq h_2 + \log Z_{\mu} &\leq \log \!{vol}(\+B_R) + \log \!{vol}(\+B_{2R}) + d+1 + \frac{d}{2}\log\frac{L^2M}{d\eps} + \log \frac{4}{\eps} \\
        \mr{\Cref{prop:Gamma}}
        &\leq \log \frac{4}{\eps} + d+1 + \frac{d}{2}\log \frac{ L^2 M}{d\eps} + \frac{d}{2}\log \frac{64e\pi M}{d\eps} + \frac{d}{2}\log \frac{4\cdot 64e\pi M}{d\eps} \\
        & \leq \log \frac{4}{\eps} + d+ 1 + d\log \frac{8 L M}{d\eps} + \frac{d}{2}\log \frac{4\cdot 64 e^2\pi^2 M}{d\eps}.
        % & = \+O\tp{d\log\frac{LM}{d\eps}}.
    \end{align*}

    Combining the above calculations, for $x\in \+B_{2R}$,
    \begin{align*}
        f_\pi(x) - f_\gamma(x) &\leq f^{\le 2R}_{\pi}(x) - f_\gamma(x) + \log \frac{1}{\eps}\\
        &\leq \ol{f^{\le 2R}_{\pi}}(x)+\log Z_{\mu} - f_\gamma(x) + \log\frac{\wh Z_\mu}{Z_\mu} + \log \frac{1}{\eps}\\
        &\leq \log \frac{4}{\eps} + d+ 1 + d\log \frac{8 L M}{d\eps} + \frac{d}{2}\log \frac{4\cdot 64 e^2\pi^2 M}{d\eps} - \frac{d}{2}\log\frac{2\pi M}{d\eps} + \log \frac{1}{\eps}\\
        &= \+O\tp{d\log \frac{LM}{d\eps}}
    \end{align*}
    and 
    \begin{align*}
        f_\pi(x) - f_\gamma(x) &\geq f^{\le 2R}_{\pi}(x) - f_\gamma(x) \\
        &= \ol{f^{\le 2R}_{\pi}}(x)+\log Z_{\mu} - f_\gamma(x) + \log\frac{\wh Z_\mu}{Z_\mu} \\
        &\geq \frac{d}{2}\log\frac{2\pi}{L} - 64d - \frac{d}{2}\log\frac{2\pi M}{d\eps} - d - 1\\
        &= - \+O\tp{d\log \frac{LM}{d\eps}}.
    \end{align*}
\end{proof}

Since $\abs{\log p_\pi(x) - \log p_\gamma(x)} = \abs{f_\gamma(x)-f_\pi(x)-\log Z_\pi} \le \abs{f_\gamma(x)-f_\pi(x)}+\abs{\log Z_\pi}$, by \Cref{lem:Zpi-close-to-one}, we have the following corollary.

\begin{corollary}\label{cor:pclose}
    For every $x\in\bb R^d$, $\tp{\frac{LM}{d\eps}}^{-\+O\tp{d}}\le \frac{p_\pi(x)}{p_\gamma(x)} \le \tp{\frac{LM}{d\eps}}^{\+O\tp{d}}$.
\end{corollary}

Then we come to the bound for the \Poincare constant of $\pi$. 

\begin{lemma}\label{lem:pi-PI}
    $C_{\!{PI}}(\pi) = \frac{2d\eps}{M}\cdot \tp{\frac{LM}{d\eps}}^{-\+O(d)}$.
\end{lemma}

\begin{proof}
    By definition, $C_{\!{PI}}(\pi) = \inf_{h\colon\bb R^d \to \bb R} \frac{\E[\pi]{\|\grad h\|^2}}{\Var[\pi]{h^2}}$. For each $h\colon\bb R^d \to \bb R$,
    \begin{align*}
        \frac{\E[\pi]{\|\grad h\|^2}}{\Var[\pi]{h}} 
        & = \frac{2\int_{\bb R^d} \|\grad h(x)\|^2 p_\pi(x) \dd x}{\int_{\bb R^d\times \bb R^d} \tp{h(x)-h(y)}^2 p_\pi(x)p_\pi(y) \dd x \dd y} \\
        \mr{\Cref{cor:pclose}}
        &\geq \tp{\frac{LM}{d\eps}}^{\+O(d)}\cdot \frac{2\int_{\bb R^d} \|\grad h(x)\|^2 p_\gamma(x) \dd x}{\int_{\bb R^d\times \bb R^d} \tp{h(x)-h(y)}^2 p_\gamma(x)p_\gamma(y) \dd x \dd y}\\
        &= \tp{\frac{LM}{d\eps}}^{\+O(d)}\cdot \frac{\E[\gamma]{\|\grad h\|^2}}{\Var[\gamma]{h}}.
    \end{align*}
    Since $C_{\!{PI}}(\gamma) = \inf_{h\colon\bb R^d \to \bb R} \frac{\E[\gamma]{\|\grad h\|^2}}{\Var[\gamma]{h}} = \frac{2d\eps}{M}$ \cite{HE76}, we know that $C_{\!{PI}}(\pi) \geq \frac{2d\eps}{M}\cdot \tp{\frac{LM}{d\eps}}^{\+O(d)}$.

\end{proof}

\subsubsection{The smoothness and first moment of $\pi$} \label{sec:ub-smooth}

In this section, we prove the smoothness property and bound the first moment of $\pi$. These properties are important in the algorithm to sample from $\pi$ in \Cref{sec:proof-of-ub}. Remember that we assumed $\grad f_\mu(0) = 0$. 

\begin{lemma}\label{lem:smooth1}
    % We have $\grad \ol{f^{\le 2R}_\pi}(0)=0$ and for each $x\in \+B_{2R}$, $\| \grad^2 \ol{f^{\le 2R}_\pi}(x) \| =  \+O\tp{\frac{L^3R^4}{\tp{h_2-h_1}^2}}$ and $\| \grad \ol{f^{\le 2R}_\pi}(x) \| = \+O\tp{\frac{L^2R^3}{h_2-h_1}}$.  
    We have $\grad \ol{f^{\le 2R}_\pi}(0)=0$ and for any $x,y \in \+B_{2R}$, $\| \grad \ol{f^{\le 2R}_\pi}(x) \| = \+O\tp{\frac{L^2R^3}{h_2-h_1}}$ and $\| \grad \ol{f^{\le 2R}_\pi}(x) - \grad \ol{f^{\le 2R}_\pi}(y) \| =  \+O\tp{\frac{L^3R^4}{\tp{h_2-h_1}^2}}\cdot \|x-y\|$.  
\end{lemma}
\begin{proof}
        By the definition of $\ol{f^{\le 2R}_\pi}$, for each $x,y\in \+B_{2R}$, direct calculation gives
        \begin{align*}
            \grad \ol{f^{\le 2R}_\pi}(x) &= \grad \mathfrak{g}_{[h_1,h_2]}(x) 
            \cdot \tp{f_\mu(x) - h_2} + \grad f_\mu(x) \cdot \mathfrak{g}_{[h_1,h_2]}(x)
            % ,\\
            % \grad^2 \ol{f^{\le 2R}_\pi}(x) &= \grad^2 \mathfrak{g}_{[h_1,h_2]}(x)\cdot \tp{f_\mu(x) - h_2} + \mathfrak{g}_{[h_1,h_2]}(x)\cdot \grad^2 f_\mu(x)\\
            % &\quad\quad +\grad \mathfrak{g}_{[h_1,h_2]}(x) \grad f_\mu(x)^{\top} + \grad f_\mu(x)\cdot \grad \mathfrak{g}_{[h_1,h_2]}(x)^{\top}
        \end{align*}
        and 
        \begin{align}
            \grad \ol{f^{\le 2R}_\pi}(x) - \grad \ol{f^{\le 2R}_\pi}(y) &= \tp{\grad \mathfrak{g}_{[h_1,h_2]}(x) - \grad \mathfrak{g}_{[h_1,h_2]}(y) }\cdot \tp{f_\mu(x) - h_2} + \grad \mathfrak{g}_{[h_1,h_2]}(y) \cdot (f_\mu(x)-f_{\mu}(y)) \notag \\
            &\quad\quad +\mathfrak{g}_{[h_1,h_2]}(x) \tp{\grad f_\mu(x) - \grad f_\mu(y)} + \grad f_\mu(y)\cdot \tp{\mathfrak{g}_{[h_1,h_2]}(x) - \mathfrak{g}_{[h_1,h_2]}(y)}. \notag
        \end{align}
        By the definition of $\mathfrak{g}_{[h_1,h_2]}$, we have
        \[
            \grad \mathfrak{g}_{[h_1,h_2]}(x) = \frac{-\grad f_\mu(x)}{h_2-h_1} \cdot q_{\!{mol}}'\tp{\frac{ h_2 - f_\mu(x) }{ h_2 - h_1 }}.
        \]
        % and
        % \[
        %     \grad^2 \mathfrak{g}_{[h_1,h_2]}(x) = \frac{\grad f_\mu(x)\cdot \grad f_\mu(x)^{\top}}{(h_2-h_1)^2} \cdot q_{\!{mol}}''\tp{\frac{ h_2 - f_\mu(x) }{ h_2 - h_1 }} - \frac{\grad^2 f_{\mu}(x)}{h_2-h_1} \cdot q_{\!{mol}}'\tp{\frac{ h_2 - f_\mu(x) }{ h_2 - h_1 }}.
        % \]
        It is easy to see $\grad \ol{f^{\le 2R}_\pi}(0)=0$. Since $f$ is $L$-smooth and $\grad f_\mu(0)=0$, for $x\in \+B_{2R}$, $\|\grad f_\mu(x) \| \leq L\|x\| \leq 2LR$ and $\|\grad f_\mu(x) - \grad f_\mu(y) \|\leq L\|x-y\|$. Recall that $q'_{\!{mol}}$ is always $O(1)$. We have $\|\grad \mathfrak{g}_{[h_1,h_2]}(x)\| =\+O\tp{\frac{LR}{h_2-h_1}}$ and 
        \begin{align*}
            \|\grad \mathfrak{g}_{[h_1,h_2]}(x) - \grad \mathfrak{g}_{[h_1,h_2]}(y) \| &\leq \norm{\frac{\grad f_\mu(x)-\grad f_\mu(y)}{h_2-h_1}}\cdot q_{\!{mol}}'\tp{\frac{ h_2 - f_\mu(x) }{ h_2 - h_1 }} \\
            &\quad + \frac{\|\grad f_{\mu}(y)\|}{h_2-h_1} \cdot \abs{q_{\!{mol}}'\tp{\frac{ h_2 - f_\mu(x) }{ h_2 - h_1 }} - q_{\!{mol}}'\tp{\frac{ h_2 - f_\mu(y) }{ h_2 - h_1 }}}\\
            &\leq \+O\tp{\frac{L}{h_2-h_1}}\cdot \|x-y\| + \+O\tp{\frac{LR}{h_2-h_1}\cdot \frac{LR}{h_2-h_1}}\cdot \|x-y\|\\
            &= \+O\tp{\frac{L^2R^2}{(h_2-h_1)^2}}\cdot \|x-y\|.
        \end{align*}
        Consequently, $\abs{f_\mu(x)-f_\mu(y)} \leq \+O(LR)\|x-y\|$ and $\abs{ \mathfrak{g}_{[h_1,h_2]}(x) - \mathfrak{g}_{[h_1,h_2]}(y) } \leq \+O\tp{\frac{LR}{h_2-h_1}}\cdot \|x-y\|$.
        
         Let $x^* = \arg\min_{x\in \+B_{2R}} f_\mu(x)$. We have for any $x\in \+B_{2R}$
         \begin{equation}\label{eq:x*}
             f^*\leq f_\mu(x) \leq f_\mu(x^*) + \grad f_\mu(x^*)\cdot (x-x^*) + \frac{L}{2}\cdot \|x-x^*\|^2 \leq f^* + 16LR^2. 
         \end{equation}
        Recall the definition of $h_2 =  \wh f^* + \log \!{vol}(\+B_{2R}) + \frac{d}{2}\log L + \log \frac{4}{\eps} + \frac{d}{2}\log\frac{LM}{d\eps}$. According to \Cref{prop:Z-and-fmin} and~\eqref{eq:x*}, for $x\in \+B_{2R}$, 
        \[
             f_\mu(x)-h_2\geq - \frac{d}{2}\log \frac{LM}{d\eps} - \log\frac{4}{\eps} + \log\Gamma\tp{\frac{d}{2}+1} - \frac{d}{2}\log \frac{4\cdot 32\pi LM}{\eps} - d,
        \]
        and
        \[
             f_\mu(x)-h_2\leq - \frac{d}{2}\log \frac{LM}{d\eps} - \log\frac{4}{\eps} + \log\Gamma\tp{\frac{d}{2}+1} - \frac{d}{2}\log \frac{4\cdot 32\pi LM}{\eps} + 16LR^2.
        \]    
        Therefore, $\abs{f_\mu(x)-h_2}= \+O(LR^2)$.

        Combining all above, we have
        \[
             \| \grad \ol{f^{\le 2R}_\pi}(x) \| \leq 2LR \cdot \+O\tp{\frac{\abs{f_\mu(x)-h_2}}{h_2-h_1}}  + 2LR = \+O\tp{\frac{L^2R^3}{h_2-h_1}},
        \]
        and 
        \begin{align*}
            \|\grad \ol{f^{\le 2R}_\pi}(x) - \grad \ol{f^{\le 2R}_\pi}(y)\| &\leq \+O\tp{\frac{L^3R^4}{\tp{h_2-h_1}^2} + \frac{L^2R^2}{h_2-h_1} + L } = \+O\tp{\frac{L^3R^4}{\tp{h_2-h_1}^2}}\cdot \|x-y\|.
        \end{align*}
\end{proof}

\begin{lemma}\label{lem:f2smooth}
    % The function $f_\pi$ is $\+O\tp{\frac{L^3R^4}{\tp{h_2-h_1}^2} + \frac{d\eps}{M}}$-smooth.
    The function $f_\pi$ is $\+O\tp{\frac{L^3R^4}{\tp{h_2-h_1}^2}}$-smooth.
\end{lemma}
\begin{proof}
    We divide $\bb R^d$ into three parts, $\+B_R$, $\+B_{2R}\setminus \+B_{R}$ and $\bb R^d \setminus \+B_{2R}$. Since our construct guarantees that $f_\pi$ and $\grad f_{\pi}$ are continuous functions, to prove the smoothness of $f_{\pi}$, we only need to bound $\|\grad f_{\pi}(x) - \grad f_{\pi}(y)\|$ for those $x,y$ from the same part. For $x,y$ from different parts, for example, if $x\in \+B_R$ and $y\in \+B_{2R}\setminus \+B_{R}$, we can find a $z$ at the intersection of this two parts such that $\|x-y\|=\|x-z\|+\|z-y\|$ and bounding $\|\grad f_{\pi}(x) - \grad f_{\pi}(y)\|$ can be transformed to bounding $\|\grad f_{\pi}(x) - \grad f_{\pi}(z)\|$ and $\|\grad f_{\pi}(z) - \grad f_{\pi}(y)\|$ respectively.

    By construction, for $x,y\in \+B_R$, $\norm{\grad f_\pi(x)-\grad f_\pi(y)}$ is bounded by \Cref{lem:smooth1}. For $x,y\in \bb R^d \setminus \+B_{2R}$, we know $\|\grad f_\pi(x)-\grad f_\pi(y)\|\leq \frac{\eps d}{M}\cdot \|x-y\|$ and $\frac{\eps d}{M} = \+O\tp{\frac{L^3R^4}{\tp{h_2-h_1}^2}}$. It remains to deal with those $x,y\in \+B_{2R}\setminus \+B_{R}$.

    For $x,y\in \+B_{2R}\setminus \+B_{R}$,
    \[
        \grad f_\pi(x) = - \grad f^{\le 2R}_\pi(x) \cdot \mathfrak{g}_{[R,2R]}(x) - \grad \mathfrak{g}_{[R,2R]}(x) \cdot f^{\le 2R}_\pi(x) + \grad \mathfrak{g}_{[R,2R]}(x) \cdot (f_{\gamma}(x)-\log \eps) + \grad f_{\gamma}(x) \cdot  \mathfrak{g}_{[R,2R]}(x)
    \] 
    and
    \begin{align}
        \|\grad f_\pi(x) - \grad f_\pi(y)\| &\leq  \norm{\grad f^{\le 2R}_\pi(x) - \grad f^{\le 2R}_\pi(y)}\cdot \abs{\mathfrak{g}_{[R,2R]}(x)} + \norm{\grad f^{\le 2R}_\pi(y)}\cdot \abs{\mathfrak{g}_{[R,2R]}(x) - \mathfrak{g}_{[R,2R]}(y)} \notag \\
        &\quad + \abs{f^{\le 2R}_\pi(x) - f^{\le 2R}_\pi(y)}\cdot \|\grad \mathfrak{g}_{[R,2R]}(x)\|  + \norm{\grad \mathfrak{g}_{[R,2R]}(x)}\cdot \abs{f_{\gamma}(x) - f_{\gamma}(y)} \notag  \\
        &\quad + \norm{\grad f_{\gamma}(x) - \grad f_{\gamma}(y)}\cdot \abs{\mathfrak{g}_{[R,2R]}(x)} + \norm{\grad f_{\gamma}(y)}\cdot \abs{\mathfrak{g}_{[R,2R]}(x) - \mathfrak{g}_{[R,2R]}(y)}.\notag  \\
        &\quad + \norm{\grad \mathfrak{g}_{[R,2R]}(x) - \grad \mathfrak{g}_{[R,2R]}(y)} \cdot \abs{f^{\le 2R}_\pi(y) -f_{\gamma}(y)+\log \eps} . \label{eq:grad2}
    \end{align}
    For the first term in \Cref{eq:grad2}, we know from \Cref{lem:smooth1} and the fact $\mathfrak{g}_{[R,2R]}(x)=\+O(1)$ that $\norm{\grad f^{\le 2R}_\pi(x) - \grad f^{\le 2R}_\pi(y)}\cdot \abs{\mathfrak{g}_{[R,2R]}(x)} = \+O\tp{\frac{L^3R^4}{\tp{h_2-h_1}^2}}\cdot \|x-y\|$. For the fifth term, similarly, we have $\norm{\grad f_{\gamma}(x) - \grad f_{\gamma}(y)}\cdot \abs{\mathfrak{g}_{[R,2R]}(x)} = \+O\tp{\frac{\eps d}{M}}\cdot \|x-y\|$.

    By the definition of $\mathfrak{g}_{[R,2R]}$, we have $\grad \mathfrak{g}_{[R,2R]}(x) = \frac{2x}{(2R)^2 - R^2} \cdot q'_{\!{mol}}\tp{\frac{\|x\|^2 - R^2}{(2R)^2 - R^2}} = \+O\tp{\frac{1}{R}}$ for $x\in \+B_{2R}$. Therefore, we can bound the second term in \Cref{eq:grad2} by $\norm{\grad f^{\le 2R}_\pi(y)}\cdot \abs{\mathfrak{g}_{[R,2R]}(x) - \mathfrak{g}_{[R,2R]}(y)} = \+O\tp{\frac{L^2R^2}{h_2-h_1}}\cdot \|x-y\|$ and also bound the third term by $\abs{f^{\le 2R}_\pi(x) - f^{\le 2R}_\pi(y)}\cdot \|\grad \mathfrak{g}_{[R,2R]}(x)\| = \+O\tp{\frac{L^2R^2}{h_2-h_1}}\cdot \|x-y\|$. 

    Since $\grad f_{\gamma}(x) = \frac{\eps d x}{M}$, for any $x\in \+B_{2R}$, $\norm{\grad f_{\gamma}(x)} = \+O\tp{\frac{\eps d R }{M}}= \+O\tp{\frac{d }{R}}$. Then both the fourth and the sixth term in \Cref{eq:grad2} can be bounded by $\+O\tp{\frac{d }{R^2}}\cdot \|x-y\|$.

    Still by the definition of $\mathfrak{g}_{[R,2R]}$, for $x\in \+B_{2R}$,
    \begin{align*}
        \|\grad \mathfrak{g}_{[R,2R]}(x) - \grad \mathfrak{g}_{[R,2R]}(y)\| &= \frac{2\|x-y\|}{(2R)^2 - R^2} \cdot \abs{q'_{\!{mol}}\tp{\frac{\|x\|^2 - R^2}{(2R)^2 - R^2}}} \\
        &\quad + \frac{2\|y\|}{(2R)^2 - R^2} \cdot \abs{q'_{\!{mol}}\tp{\frac{\|x\|^2 - R^2}{(2R)^2 - R^2}} - q'_{\!{mol}}\tp{\frac{\|y\|^2 - R^2}{(2R)^2 - R^2}}} \\
        &\leq \+O\tp{\frac{1}{R^2}}\cdot \|x-y\|.
    \end{align*}
    Therefore, from \Cref{lem:fclose} the last term in \Cref{eq:grad2} can be bounded by  $\+O\tp{\frac{d\log \frac{LM}{d\eps}}{R^2}}\cdot \|x-y\|$.

    % Combining the above equation and \Cref{lem:fclose}, we have $\norm{\grad^2 \mathfrak{g}_{[R,2R]}(x) \cdot \tp{f_\gamma(x) - f^{\le 2R}_\pi(x) - \log\eps}} \leq \+O\tp{\frac{d\eps}{M}\cdot \log\frac{LM}{d\eps}}$. From \Cref{lem:smooth1}, we have $\norm{\grad \mathfrak{g}_{[R,2R]}(x) \grad f^{\le 2R}_\pi(x)^\top},\norm{\grad f^{\le 2R}_\pi(x)\grad \mathfrak{g}_{[R,2R]}(x)^\top} \leq \+O\tp{\frac{L^2R^2}{h_2-h_1}}$. For $\norm{\grad \mathfrak{g}_{[R,2R]}(x) \grad f_\gamma(x)^\top}$ and $\norm{ \grad f_\gamma(x)\grad \mathfrak{g}_{[R,2R]}(x)^\top}$, they can be bounded by $\+O\tp{\frac{d\eps}{M}}$.
    In total, $\grad f_\pi(x)$ is $ \+O\tp{\frac{L^3R^4}{\tp{h_2-h_1}^2}}$-Lipschitz.
\end{proof}

The following result is a corollary of the above lemmas.
\begin{corollary}\label{coro:gap}
    The function $f_{\pi}$ satisfies $f_{\pi}(0) - \min_{x\in \bb R^d} f_{\pi}(x) = \+O\tp{\frac{L^3R^6}{\tp{h_2-h_1}^2}}$.
\end{corollary}
\begin{proof}
    By the definition of $f_\pi$, it is increasing as $\|x\|$ increases outside $\+B_{2R}$. Therefore, $f_{\pi}(0) - \min_{x\in \bb R^d} f_{\pi}(x) = f_{\pi}(0) - \min_{x\in \+B_{2R}} f_{\pi}(x)$. Since $\grad f_{\pi}(0) = \grad \ol{f^{\le 2R}_\pi}(0)=0$ and $f_{\pi}$ is $\+O\tp{\frac{L^3R^4}{\tp{h_2-h_1}^2}}$-smooth from \Cref{lem:f2smooth}, for arbitrary $x\in \+B_{2R}$, $f_{\pi}(0) - f_{\pi}(x) \leq \+O\tp{\frac{L^3R^6}{\tp{h_2-h_1}^2}}$ and consequently, $f_{\pi}(0) - \min_{x\in \bb R^d} f_{\pi}(x) = \+O\tp{\frac{L^3R^6}{\tp{h_2-h_1}^2}}$.
\end{proof}

Then we bound the first moment of $\pi$.
\begin{lemma}\label{lem:f2moment}
    The first moment of $\pi$ is bounded by $\+O(\sqrt{M})$.
\end{lemma}
\begin{proof}
    By the definition of $\pi$ and $f_\gamma$, we have
    \begin{align*}
        \E[X \sim \pi]{\|X\|}^2 &\leq \E[X \sim \pi]{\|X\|^2} \\
        & \le \frac{1}{Z_{\pi}} \tp{\int_{\bb R^d} \|x\|^2\cdot \exp\tp{-f_\gamma(x)+\log \eps} \dd x + \frac{Z_{\mu}}{\wh Z_{\mu}}\cdot \frac{1}{Z_{\mu}}\cdot \int_{\+B_{2R}} \|x\|^2\cdot \exp\tp{-\ol{f^{\le 2R}_\pi}(x)} \dd x } \\
        \mr{\Cref{lem:Zpi-close-to-one}}
        & \leq \frac{M}{Z_{\pi}} + \frac{2}{Z_{\mu}} \tp{\int_{\bb R^d} \|x\|^2\cdot \exp\tp{-f_\mu(x)} \dd x + \int_{\+B_{2R}} \|x\|^2\cdot \exp\tp{-h_1} \dd x } \\
        &\leq \frac{M}{ Z_{\pi}} + 2M + \frac{2}{Z_\mu}\cdot \int_{\+B_{2R}} (2R)^2\cdot \exp\tp{-h_1} \dd x  \\
        \mr{Definition of $h_1$} &\leq \frac{M}{Z_{\pi}} + 2M + \frac{8R^2}{Z_\mu}\cdot \exp\tp{-\wh f^* - \frac{d}{2}\log L - \log \frac{4}{\eps}} \\
        \mr{\Cref{lem:mu-bound}}
        &\leq \frac{M}{Z_{\pi}} + 2M + 2\eps R^2 \cdot (2\pi)^{-\frac{d}{2}}\\
        \mr{\Cref{lem:Zpi-close-to-one}}
        &= \+O\tp{M}.
    \end{align*}
\end{proof}

% \begin{lemma}\label{lem:f2moment}
%     The second moment of $\pi$ is bounded by $\+O(M)$.
% \end{lemma}
% \begin{proof}
%     By the definition of $\pi$ and $f_\gamma$, we have
%     \begin{align*}
%         \E[X \sim \pi]{\|X\|^2} 
%         & \le \frac{1}{Z_{\pi}} \tp{\int_{\bb R^d} \|x\|^2\cdot \exp\tp{-f_\gamma(x)+\log \eps} \dd x + \frac{Z_{\mu}}{\wh Z_{\mu}}\cdot \frac{1}{Z_{\mu}}\cdot \int_{\+B_{2R}} \|x\|^2\cdot \exp\tp{-\ol{f^{\le 2R}_\pi}(x)} \dd x } \\
%         \mr{\Cref{prop:Z-and-fmin}}
%         & \leq \frac{M}{Z_{\pi}} + \frac{2}{Z_{\mu}} \tp{\int_{\bb R^d} \|x\|^2\cdot \exp\tp{-f_\mu(x)} \dd x + \int_{\+B_{2R}} \|x\|^2\cdot \exp\tp{-h_1} \dd x } \\
%         &\leq \frac{M}{ Z_{\pi}} + 2M + \frac{2}{Z_\mu}\cdot \int_{\+B_{2R}} (2R)^2\cdot \exp\tp{-h_1} \dd x  \\
%         \mr{Definition of $h_1$} &\leq \frac{M}{Z_{\pi}} + 2M + \frac{8R^2}{Z_\mu}\cdot \exp\tp{-\wh f^* - \frac{d}{2}\log L - \log \frac{4}{\eps}} \\
%         \mr{\Cref{lem:mu-bound}}
%         &\leq \frac{M}{Z_{\pi}} + 2M + 2\eps R^2 \cdot (2\pi)^{-\frac{d}{2}}\\
%         \mr{\Cref{lem:Zpi-close-to-one}}
%         &= \+O\tp{M}.
%     \end{align*}
% \end{proof}

\subsection{Estimate $f^*$ and $Z_\mu$}\label{sec:estimate-of-pi}

In this section, we prove \Cref{prop:Z-and-fmin}, namely to show that how to get estimators $\wh f^*$ and $\wh Z_{\mu}$ satisfying
\[
    f^* \leq \wh f^* \leq f^* + d \quad \mbox{and}\quad 
    \frac12 e^{-d}\le \frac{\wh Z_{\mu}}{Z_{\mu}} \le 1.
\]
The idea is to discretize $\bb R^d$ into cubes of side length $\ell$ and use information in each cube to construct the estimation. Let $\ell = \frac{1}{64}\sqrt{\frac{d\eps}{L^2 M}}$ and $R_0 = 2R + \sqrt{d \cdot \ell^2} $. Let $\+Z_{R_0} = \+B_{R_0}\cap \ell \bb Z^d$ be the collection of vertices of the cubes in $\+B_{R_0}$.
%\set{x\in \+B_{R_0}:\ x / \ell\in \bb Z^d}$.

\begin{lemma}\label{lem:cubes}
    There are at most $\tp{\frac{2^{10}\cdot 5 LM}{d\eps}}^{d}$ cubes with side length $\ell$ in $\+B_{R_0}$ whose vertices are all in $\+Z_{R_0}$. 
\end{lemma}
\begin{proof}
    From \Cref{cor:dballvolbound}, $\!{vol}(\+B_{R_0}) = \frac{\tp{\pi R_0^2}^{\frac{d}{2}}}{\Gamma\tp{ \frac{d}{2}+1} } \leq \tp{\frac{2\pi e R_0^2}{d}}^{\frac{d}{2}}$. The volume of a cube with side length $\ell$ is $\tp{\frac{d\eps}{2^{12}\cdot L^2 M}}^{\frac{d}{2}}$. For the cubes whose vertices are all in $\+Z_{R_0}$, the overlapping area is $0$. Therefore, the total number of such cubes is no larger than 
    \begin{align*}
        \tp{2\pi e R_0^2 \cdot \frac{2^{12}\cdot L^2 M}{\eps d^2}}^{\frac{d}{2}} &= \tp{2\pi e \cdot \frac{2^{12}\cdot L^2 M}{\eps d^2}\cdot \tp{4R^2 + d\ell^2 + 4R\cdot \sqrt{d\ell^2}}}^{\frac{d}{2}}\\
        &= \tp{2\pi e \cdot \frac{2^{12}\cdot L^2 M}{\eps d^2}\cdot \tp{\frac{4\cdot 32 M}{\eps} + \frac{d^2\eps}{64^2\cdot L^2 M} + \frac{\sqrt{2}d}{4L}}}^{\frac{d}{2}}\\
        &= \tp{2\pi e \cdot \tp{\frac{2^{19}L^2 M^2}{\eps^2 d^2} + \frac{1}{4} + \frac{2^8\sqrt{2} LM}{\eps d}}}^{\frac{d}{2}}\\
        &\leq \tp{\frac{2^{10}\cdot 5 LM}{d\eps}}^{d}.
    \end{align*}
    % $\tp{2\pi e R_0^2 \cdot \frac{8\cdot 128 L^2 M}{\eps d^2}}^{\frac{d}{2}} = \tp{ 2\pi e \tp{\frac{8\cdot 128^2 M^2 L^2}{\eps^2 d^2} + 1 }}^{\frac{d}{2}} \leq \tp{\frac{2^{10}\cdot 5 LM}{d\eps}}^{d}$.
\end{proof}

We consider those cubes with side length $\ell$ and all vertices in $\+Z_{R_0}$. Let $n$ be the total number of such cubes. From \Cref{lem:cubes}, $n\leq \tp{\frac{2^{10}\cdot 5 LM}{d\eps}}^{d}$. Denote these cubes as $C_1,C_2,\dots,C_n$ and let $v_1,v_2,\dots, v_n$ be the center of these cubes. We first show that these cubes well cover the ball $\+B_{2R}$.
\begin{lemma}\label{lem:cubecover}
    For each $x\in \+B_{2R}$, there exists $i\in[n]$ such that $x\in C_i$.
    % $\|x - v_i\|\leq 2\sqrt{d}\ell$.
\end{lemma}
\begin{proof}
    % Assume in contradiction that we cannot find such a $v_i$ for some point $y\in \+B_{R'}$. 
    For each point $x\in \+B_{2R}$, we define $\ol x\in \bb R^d$ as
    $
        \forall j\in[d],\  \ol x(j) = \begin{cases}
            \lfloor \frac{x(j)}{\ell} \rfloor \cdot \ell, &\mbox{ if } x(j)\geq 0\\
            \lceil \frac{x(j)}{\ell} \rceil\cdot \ell, &\mbox{ if } x(j)< 0
        \end{cases}.
    $
    Consider the following cube
    \[
        C_x = \set{y\in \bb R^d:\ \forall j \in [d], y(j)\in \begin{cases}
            [\ol x(j), \ol x(j) + \ell], &\mbox{ if }\ol x(j)\ge 0\\
            [\ol x(j) - \ell, \ol x(j)], &\mbox{ if }\ol x(j)< 0
        \end{cases}}.
    \]
    It is obvious that $x\in C_x$ and for each vertex $y\in C_x$,
    \begin{align*}
        \|y\|^2 &\leq \sum_{j=1}^d \tp{\ol x(j)^2 + \ell^2 + 2\ell \cdot \abs{\ol x(j)}} \\
        &= \|\ol x\|^2 + d\ell^2 + 2\sqrt{d\ell^2}\cdot \frac{\sum_{j=1}^n \abs{\ol x(j)}}{\sqrt{d}}\\
        \mr{Cauchy-Schwartz inequality}&\leq \|\ol x\|^2 + d\ell^2 + 2\sqrt{d\ell^2}\cdot \|\ol x\|
        \\ 
        &\leq 4R^2 + d\ell^2 + 2\sqrt{d\ell^2}\cdot 2R
        \\ 
        &\leq R_0^2.
    \end{align*}
    Therefore $y\in \+Z_{R_0}$ and there exists $i\in [n]$ such that $C_i=C_x$.

    % We have $\ol x\in \+B_{2R} \cap \+Z_{R_0}$ and $\|\ol x - x\|^2 \leq d\ell^2$. Consider the following cube
    % \[
    %     C_x = \set{y\in \bb R^d:\ \forall j \in [d], y(j)\in \begin{cases}
    %         [\ol x(j), \ol x(j) + \ell], &\mbox{ if }\ol x(j)< 0\\
    %         [\ol x(j) - \ell, \ol x(j)], &\mbox{ if }\ol x(j)\geq 0
    %     \end{cases}}.
    % \]
    % It is obvious that there exists $i\in [n]$ such that $C_i=C_x$ and $\|x - v_i\|\leq \|x - \ol x\| + \|\ol x - v_i\| \leq 2\sqrt{d}\ell$.
\end{proof}

In the following, we will write $v_x$ for $v_i$ where $i$ is the unique $i$ such that $x\in C_i$.
%For each $x\in \+B_{2R}$, we denote $v_x$ as the vector $v_i$ where $i$ is the smallest index in $[n]$ such that $x\in C_i$. 
% For $x\in \+B_{2R} \setminus \tp{\bigcup_{j\in[n]} C_j}$, let $v_i$ be the closest point to $x$ in $\set{v_j}_{j\in[n]}$. 
% We denote the unique $v_i$ associated with each $x$ as $v_x$. 
%Let $\+J = \set{i\in[n]: \exists x\in \+B_{2R}, v_x = v_i}$. 
Let $\+J\defeq \set{i\in [n]\cmid C_i\cap \+B_{2R}\ne\emptyset}$. To estimate $f^*$ and $Z_{\mu}$, we query the value of each $f_\mu(v_i)$ and assign
\begin{equation}\label{eqn:hat}
    \wh f^* = \min_{i\in \+J} f_\mu(v_i) + \frac{d}{2},\quad \wh Z_{\mu} = \sum_{i=1}^n \!{vol}(C_i) \cdot \exp\tp{-f_\mu(v_i) - \frac{d}{2}}.
\end{equation}

%\mn{Here we add a multiplier $\frac{1}{\eps}$ to $\hat Z_{\mu}$ artificially to make $\hat Z_{\mu}$ large enough.}
Then the query complexity to determine $\wh f^*$ and $\wh Z_{\mu}$ is at most $\tp{\frac{2^{10}\cdot 5 LM}{d\eps}}^{d}$. In the following, we show that our construction satisfies the accurancy requirement, namely that

% The following lemma shows that such constructions indeed satisfy \Cref{assump:Z-and-fmin}.
 \begin{lemma}\label{lem:estimate}
     The construction of $\wh f^*$ and $\wh Z_{\mu}$ in~\eqref{eqn:hat} satisfies
     \[
    f^* \leq \wh f^* \leq f^* + d \quad \mbox{and}\quad 
    \frac12 e^{-d}\le \frac{\wh Z_{\mu}}{Z_{\mu}} \le 1.
    \]
 \end{lemma}
 \begin{proof}
    Since the function $f_\mu$ is $L$-smooth and $\grad f_\mu(0)=0$, for each $x\in \+B_{R_0}$, $\|\grad f_\mu(x)\|\leq L\|x\| \leq LR_0$. From \Cref{lem:cubecover}, we always have $\|x - v_x\| \leq \sqrt{d}\ell$. Then by the definition of $L$-smooth,
    \begin{align*}
        f_\mu(v_x) &\leq f_\mu(x) + \grad f_\mu(x)^{\top} (x-v_x) + \frac{L}{2} \|x - v_x\|^2\\
        & \leq f_\mu(x) + \|\grad f_\mu(x)\| \cdot \|x-v_x\| + \frac{L}{2} \|x - v_x\|^2 \\
        &\leq f_\mu(x) + LR_0\cdot \sqrt{d}\ell + \frac{L}{2}\cdot d\ell^2 \leq f_\mu(x) + \frac{d}{2}
    \end{align*}
    % \htodo{The universal constants here might be inaccurate.}
    and similarly
    \begin{align*}
        f_\mu(x) & \leq f_\mu(v_x) + \grad f_\mu(v_x)^{\top} (v_x - x) + \frac{L}{2} \|x - v_x\|^2 \leq f_\mu(v_x) + \frac{d}{2}.
    \end{align*}
    Therefore, $f^* = \min_{x\in \+B_{2R}} f_\mu(x) \leq \min_{x\in \+B_{2R}} f_\mu(v_x) + \frac{d}{2} = \wh f^*$ and $f^* \geq \min_{x\in \+B_{2R}} f_\mu(v_x) - \frac{d}{2} = \wh f^* - d$.  From the same calculation, we know for each $x\in C_i$, $f_\mu(v_i) - \frac{d}{2} \leq f_\mu(x)\leq f_\mu(v_i) + \frac{d}{2}$. For $\wh Z_{\mu}$, we have
    \[
        Z_{\mu} \geq \sum_{i=1}^n \int_{C_i} \exp\tp{-f_\mu(x)} \d x \geq \sum_{i=1}^n \!{vol}(C_i)\cdot \exp\tp{-f_\mu(v_i) - \frac{d}{2}} = \wh Z_{\mu}.
    \]
    On the other hand, since $\+B_R\subseteq \bigcup_{i\in[n]} C_i$, we have $\int_{\bb R^d \setminus \tp{\bigcup_{i\in[n]} C_i}} \exp\tp{-f_\mu(x)} \d x < \int_{\bigcup_{i\in[n]} C_i} \exp\tp{-f_\mu(x)} \d x$. Therefore,
    \[
        Z_{\mu}\leq 2\sum_{i=1}^n \int_{C_i} \exp\tp{-f_\mu(x)} \d x \leq 2\sum_{i=1}^n \!{vol}(C_i)\cdot \exp\tp{-f_\mu(v_i) + \frac{d}{2}} = 2e^d\cdot \wh Z_{\mu}.
    \]
\end{proof}

\subsection{Proof of \Cref{thm:main-ub}} \label{sec:proof-of-ub}

From previous sections we know that the distribution $\pi$ is $\+O\tp{\frac{L^3R^4}{\tp{h_2-h_1}^2}}$-log-smooth (\Cref{lem:f2smooth}), has its first moment bounded by $\+O(\sqrt{M})$ (\Cref{lem:f2moment}), satisfies $\DTV\tp{\pi,\mu}= \frac{\eps}{2}$ (\Cref{lem:pi-mu-close}) and satisfies $C_{\!{PI}}\ge \frac{2d\eps}{M}\cdot \tp{\frac{LM}{d\eps}}^{-\+O(d)}$ (\Cref{lem:pi-PI}). Moreover, we can query $f_\pi(x)$ and $\grad f_\pi(x)$ efficiently, provided query access to $f_\mu$ and $\grad f_\mu$. 

Therefore, we can use the algorithm in~\cite{BCE+22} to sample from $\pi$ (see also~\cite[Chapter 11]{Che24}). Let $N$ be the total steps and $h$ be the step size. To sample from a target distribution $\nu\propto e^{-f}$, their algorithm acts as follows:
\begin{itemize}
    \item[1.] Pick a time $t_0\in[0,Nh]$ uniformly at random.
    \item[2.] Let $k_0$ be the largest integer such that $k_0 h<t$. For each $t<t_0$ and $k\leq k_0$, the process evolves as 
    \begin{equation}
        X_t = X_{kh} - (t-kh) \grad f(X_{kh}) + \sqrt{2}(B_t-B_{kh}), \label{eq:LMC2}
    \end{equation} 
    where $\set{B_t}_{t\geq 0}$ is the standard Brownian motion.
    \item[3.] Output $X_{t_0}$.
\end{itemize}

\begin{theorem}[A direct consequence of Corollary 8 in \cite{BCE+22}]\label{thm:LMCforPI}
    Let $\set{\mu_t}_{t\geq 0}$ denote the law of the interpolation \Cref{eq:LMC2} of LMC. Assume the potential function $f$ is $\+L$-smooth and the target distribution $\nu\propto e^{-f}$ satisfies the \Poincare inequality with constant $\alpha>0$. If $\!{KL}(\mu_0 \| \nu)\leq K_0$, choosing step size $h=\frac{\sqrt{K_0}}{2\+L\sqrt{dN}}$, then for $N\geq \max\set{\frac{32^2 \alpha^{-2} \+L^2 dK_0}{\delta^4}, \frac{9K_0}{d}}$ and $\ol \mu_{Nh}\defeq \frac{\int_0^{Nh} \mu_t \d t}{Nh}$, $\DTV(\ol \mu_{Nh},\nu)\leq \delta$.
\end{theorem}

To get a convergence guarantee for our target distribution $\pi$, it remains to find an initial distribution $\mu_0$ such that $\!{KL}(\mu_0 \| \pi)$ is bounded. \Cref{lem:f2smooth} shows that $f_{\pi}$ is $\+L$-smooth for $\+L=\+O\tp{\frac{L^3R^4}{\tp{h_2-h_1}^2}}$. By choosing $\mu_0$ as $\+N\tp{0, \frac{\!{Id}_d}{2\+L}}$, we can bound $\!{KL}(\mu_0 \| \pi)$ using the following lemma.
\begin{lemma}[A direct corollary of Lemma 32 in \cite{CEL+24}]\label{lem:initial}
    Suppose $\grad f(0)= 0$ and $f$ is $\+L$-smooth. Let $m = \E[X\sim \nu]{\|X\|}$ be the first moment of $\nu \propto e^{-f}$. Then for $\mu_0=\+N\tp{0, \frac{\!{Id}_d}{2\+L}}$,
    \[
        \log \tp{\sup \frac{\d \mu_0}{\d \nu}} \leq 2+\+L + f(0)-\min_{x\in \bb R^d} f(x) + \frac{d}{2}\log \tp{4m^2\+L}.
    \]
\end{lemma}

Combining \Cref{coro:gap}, \Cref{lem:f2moment} and \Cref{lem:initial}, $\!{KL}(\mu_0\|\pi)$ can be bounded by $\!{poly}(d,M,L,\eps^{-1})$. Therefore, we can choose $\delta = \frac{\eps}{2}$ in \Cref{thm:LMCforPI} to sample from a distribution whose total variation distance is at most $\frac{\eps}{2}$ to $\pi$ with $\!{poly}(L,M, d,\eps^{-1})\cdot \tp{\frac{LM}{\eps d}}^{\+O(d)}$ queries to $f_\mu$ and $\grad f_{\mu}$.

%% file: smoothness.tex
\section{The smoothness conditions}\label{sec:OU-smooth}

In this section, we will compare \Cref{assump:smooth} with the smoothness assumption in \cite{HZD+24} and prove \Cref{thm:main-smooth}.

Recall that we assume the target distribution $\mu$ with density $p_{\mu}$ to be $L$-log-smooth. This assumption is typically essential for bounding the discretization error of sampling algorithms. In many works based on denoising diffusion probabilistic models (DDPMs) (e.g. \cite{CCL+23,LLT23,CLL23,HZD+24}), they further assume that the distributions during the OU process starting from $\mu$ are also $L$-log-smooth. 
% \ctodo{What is the $L$-smoothness assumption in those ``score function estimation $\implies$ good sampler'' works}

The definition for the OU process is as follows. Suppose we start from a random point $X_0\sim \mu$. The OU process $\set{X_t}_{t\geq 0}$ evolves with the following equation
\[
    \d X_t = -X_t \dd t + \sqrt{2}\d B_t
\]
where $\set{B_t}_{t>0}$ is the standard Brownian motion. The solution of the above equations is
\begin{equation}
    X_t = e^{-t}X_0 + \sqrt{2}\cdot e^{-t}\int_{0}^t e^s \d B_s. \label{eq:OU}
    % X_t = e^{-t}X_0 + \sqrt{1-e^{-2t}} Z_t, \label{eq:OU}
\end{equation}
From direct calculation, we know that $
\sqrt{2}\cdot e^{-t}\int_{0}^t e^s \d B_s \sim \+N\tp{0, (1-e^{-2t})I_d}$.
% where $Z_t$ is drawn from a standard Gaussian distribution.

There have been many convergence guarantees for the DDPMs.  Let $\mu_t$ be the distribution of $X_t$ and $p_t$ be the corresponding density function. Assuming the second-moment of $\mu$ is bounded by $M$ and $\log p_t$ being $L$-smooth for any $t\geq 0$, the work of \cite{HZD+24} proposed an algorithm that guarantees with high probability, the output distribution is $\tilde{\+O}(\eps)$-close to the target distribution in KL divergence, requiring at most $\exp\set{\+O\tp{L^3\cdot \log^3\frac{Ld+M}{\eps}} \cdot \max\ab\{\log\log Z^2,1\}}$ queries \footnote{Here $Z$ is the maximum norm of particles appeared in their algorithm.}. 
% This implies that as long as the smoothness condition of $\log p_t$ is satisfied with constant $L$, a quasi-polynomial sampling algorithm exists. 

It has been known that a smooth $\log p_{\mu}$ can imply the smoothness of $\log p_t$ in some specific cases, for example when $t$ is small (\cite{CLL23}) or when $\mu$ is strongly log-concave (\cite{LPSR21}). Via the techniques in Lemma~12 of \cite{CLL23}, one can prove an $O(d)$ upper bound of $\|\grad^2 \log p_t\|_{\!{op}}$ in expectation when $t=\Omega(1)$. However, this bound does not offer much utility for the algorithm of \cite{HZD+24}, as their results will be super-exponential under an $O(d)$-smoothness bound.

If the smoothness bound for $\log p_{\mu}$ is large, the bound in \cite{HZD+24} will be poor. Fortunately, we can assume $\log p_{\mu}$ to be $\+O(1)$-smooth without loss of generality. This is because we can always scale the domain to adjust the distribution's smoothness bound, while not changing the product of the smoothness bound and the second moment $M$ (\Cref{lem:LM}). Therefore, with the initial distribution being $\+O(1)$-log-smooth and the second moment being polynomial in $d$, if the $\log p_t$'s also remain $\+O(1)$-smooth, quasi-polynomial sampler exists. 
% So the relationship between the smoothness of the initial distribution and that of the distributions during the OU process is worth studying. 

% Hence previous works do not pay much attention to distinguishing between these two smoothness conditions.

Therefore, we are interested in the conditions for the $O(1)$-smoothness to be kept during the entire OU process. Our results in \Cref{thm:main-lb} indicates an exponential lower bound even for $O(1)$-log-smooth distributions. Then we know that for those hard instances constructed in \Cref{sec:lb}, $\log p_t$ cannot always be $O(1)$-smooth during the OU process. Otherwise a quasi-polynomial sampler exists from \cite{HZD+24}. 
% we found that the smoothness bound for $\log p_t$ can be significantly worse than that of the initial $\log p_{\mu}$.
In \Cref{subsec:stitched}, we introduce a family of $\+O(1)$-log-smooth distribution, which we refer to as the stitched Gaussian distributions and is a simplified version of those hard instances in \Cref{sec:lb}. As the OU process evolves, the bound on the smoothness of stitched Gaussians can become $\omega(1)$ at certain time $t$. 
% Given this counterexample, the relationship between the smoothness of the initial distribution and that of the distributions during the OU process is worth studying. 

To see the case for other non-log-concave distributions, in \Cref{subsec:mix}, we considered a class of classical multi-modal distributions, the mixture of Gaussians, and provide analyses of their smoothness properties with different parameter settings. Although mixture of Gaussians appear to be quite similar to the stitched Gaussians, our results demonstrate that they exhibit fundamentally different behaviors in terms of smoothness. To be specific, we show that for a mixture of two Gaussians with mean $u_1$ and $u_2$, if their covariance matrices $\Sigma_1=\Sigma_2\succeq \Omega(1)\cdot \!{Id}_d$, the smoothness of the distributions are almost determined by $\|u_1-u_2\|$ and the $\log p_t$'s will inherit the $\+O(1)$-smoothness is $\log p_{\mu}$ is $\+O(1)$-smooth. In contrast, when the covariance matrices differ, even with $\|u_1-u_2\|=o(1)$, the initial $\log p_{\mu}$ is not $L$-smooth for any $L=o(d)$. 

We also explore the mixture of multiple Gaussians in \Cref{subsec:mix}. For those cases with more components, the analysis becomes more complex. Even when all covariance matrices are the same, it is challenging to derive a concise rule to characterize the relationship between smoothness and the distances between the means of the Gaussians. We give an example where the centers of components are far apart, yet the mixture distribution remains $\+O(1)$-log-smooth, and this smoothness is preserved during the OU process.

An overview of the main results of this section is given in \Cref{tab:result-comp}, where we show the smoothness bounds of $\log p_{\mu}$ and corresponding $\log p_t$ in different cases, as well as whether the $\+O(1)$-smoothness property is preserved during the OU process.

\begin{table*}[htbp]
	\centering
	\caption{The Comparison between the Smoothness Bounds}
	\label{tab:result-comp}
  \begin{threeparttable}
\begin{tabular}{m{1.7cm}<{\centering}m{2.7cm}<{\centering}m{3.9cm}<{\centering}m{3.7cm}<{\centering}m{2cm}<{\centering}}
	\toprule
	& {Parameters \textcolor{red}{\tnote{1}}} & {Smoothness Bound for $\log p_{\mu}$} & {Smoothness Bound for $\log p_t$} & {Keep $\+O(1)$-smooth?}\\
	\midrule
	Stitched Gaussian & \shortstack{$m=2$ \\ $\Sigma_1=\Sigma_2=\!{Id}_d$ \\ $\|u_1-u_2\|^2 = s = \Omega(d)$} & \shortstack{$\+O(1)$\\ (\Cref{lem:stitchsmooth})} & \shortstack{ $\Omega\tp{e^{-2t}s - 1}$ \textcolor{red}{\tnote{2}}
        %$\+O(\max\ab\{1,e^{-2t}s\})$
    \\ for $t>\frac{\log 10}{2}$\\ (\Cref{thm:stitched2})} & No \\
    \hline 
    Mixture of Gaussians & \shortstack{$m=2$\\
    $\Sigma_1=\Sigma_2\succeq \Omega(1)\cdot \!{Id}_d$ \\
    $\|u_1-u_2\|^2=s$} & \shortstack{$\+O(\max\ab\{1,s\})$\\ (\Cref{lem:2-same})} & \shortstack{$\+O(\max\ab\{1,e^{-2t}s\})$\\ (\Cref{lem:2-same})} & Yes \\
    \hline 
    Mixture of Gaussians & \shortstack{$m=2$\\
    $\Sigma_1=\frac{\!{Id}_d}{2},\Sigma_2= \!{Id}_d$ \\
    $\|u_1-u_2\|=o(1)$} & \shortstack{$\Omega(d)$\\ (\Cref{lem:2-diff})} & - & NA \\
    \hline 
    Mixture of Gaussians & \shortstack{$m=2^d$\\
    $\Sigma_1=\cdots=\Sigma_m=J$ \\
    $\|u_i-u_j\|^2= \+O(d)$ \textcolor{red}{\tnote{3}}} & \shortstack{$\+O(1)$\\ (\Cref{lem:mixture1})} & \shortstack{$\+O(1)$\\ (\Cref{cor:mixture2})} & Yes \\
	\bottomrule
\end{tabular}
\begin{tablenotes}
	\footnotesize
    \item[\textcolor{red}{1}] Note that both the stitched Gaussian and the mixture of Gaussians are constructed based on Gaussian distributions with different parameters. We assume the Gaussian distributions used in the construction are $\+N(u_1,\Sigma_1),\+N(u_2,\Sigma_2), \dots, \+N(u_m,\Sigma_m)$ respectively.
    \item[\textcolor{red}{2}] We say the smoothness bound for a function $f$ is $\Omega(c)$ if there exists some $x\in \bb R^d$ such that $\norm{\grad^2 f(x)}_{\!{op}} = \Omega(c)$.
	\item[\textcolor{red}{3}] Here $J$ is a symmetric matrix with $\delta_{\!{Id}_d}\preceq J \preceq (1-\delta)\!{Id}_d$ for some constant $\delta\in (0,1/2)$. For the detailed construction of $\ab\{u_i\}_{i\in[m]}$, see \Cref{eq:HS-mix}.
  \end{tablenotes}
\end{threeparttable}
\end{table*}

\subsection{The stitched Gaussian distributions}\label{subsec:stitched}

In this section,  we show that the smoothness bound of $\log p_t$ can differ significantly from that of $\log p_{\mu}$ on the \emph{stitched Gaussian distributions}.
% we prove that the smoothness bound for $\log p$ and $\log p_t$ can vary significantly when $p$ is a stitched Gaussian distribution. 

The stitched Gaussian distributions are a class of distributions constructed by interpolating between multiple Gaussian components. Recall that $q_{\!{mol}}$ is the mollifier defined in \Cref{sec:mollifier}. Let $u\in \bb R^d$ be a vector satisfying $\|u\|^2 \geq 100 d$. Define $\mathfrak{g}_{u}(x) = q_{\!{mol}}\tp{10\tp{\frac{\|x-u\|}{\|u\|} - 0.4}}$ and 
\[
    f_{\mu}(x) = \mathfrak{g}_{u}(x)\cdot \frac{\|x\|^2}{2} + (1-\mathfrak{g}_{u}(x))\cdot \frac{\|x-u\|^2}{2}.
\]  
The specific type of stitched Gaussians we consider here has a density function of $p_{\mu}\propto e^{-f_{\mu}}$.

% When the information is clear from context, we abbreviate $\mathfrak{g}_{\left[\frac{2\|u\|}{5},\frac{\|u\|}{2}\right]}$ as $\mathfrak{g}$ for simplicity.
Let $r=\|u\|$. We can divide $\bb R^d$ into two parts $\+B_{\frac{r}{2}}(u) = \set{x\in \bb R^d: \| x - u\| \leq 0.5 \|u\|}$ and $\ol{\+B_{\frac{r}{2}}(u)} = \bb R^d \setminus \+B_{\frac{r}{2}}(u)$. By definition of $\mathfrak{g}_{u}$, for $x\in \ol{\+B_{\frac{r}{2}}(u)}$, $f_{\mu}(x) = \frac{\|x\|^2}{2}$. Furthermore, inside $\+B_{\frac{r}{2}}(u)$, when $x\in \+B_{\frac{2r}{5}}(u) =\set{x\in \bb R^d: \| x - u\| \leq 0.4 \|u\|}$, $f_{\mu}(x) = \frac{\|x-u\|^2}{2}$. In $\+B_{\frac{r}{2}}(u)\setminus \+B_{\frac{2r}{5}}(u)$, these two Gaussians are ``stitched together'', which means the density function transitions smoothly from a Gaussian distribution centered at $0$ to another Gaussian distribution centered at $u$.

\subsubsection{The smoothness of $\log p_{\mu}$}\label{subsec:stitched1}

We first prove that $\log p_{\mu}$ is indeed $\+O(1)$-smooth.

\begin{lemma}\label{lem:stitchsmooth}
    The function $\log p_{\mu}$ is $\+O(1)$-smooth for any $u\in \bb R^d$.
\end{lemma}
\begin{proof}
    For $x\in \+B_{\frac{2r}{5}}(u)$, we know that $f_{\mu}(x) = \frac{\|x-u\|^2}{2}$. Therefore $\grad^2 \log p_{\mu}(x) = \grad^2 f_{\mu}(x) = \!{Id}_d$. Similarly, for $x\in \ol{\+B_{\frac{r}{2}}(u)}$, $f_{\mu}(x)=\frac{\|x\|^2}{2}$ and $\grad^2 \log p_{\mu}(x) = \grad^2 f_{\mu}(x) = \!{Id}_d$. So it only remains to deal with those $x\in \+B_{\frac{r}{2}}(u) \setminus \+B_{\frac{2r}{5}}(u)$.

    % Let $\+B_{\frac{r}{2}}(u)\setminus \+B_{\frac{2r}{5}}(u) = \+B_{\frac{r}{2}}(u)\setminus \+B_{\frac{2r}{5}}(u) = \set{x\in \bb R^d:0.4 \|u\| < \| x - u\| \leq 0.5 \|u\|}$.
    For $x\in \+B_{\frac{r}{2}}(u)\setminus \+B_{\frac{2r}{5}}(u)$, 
    \begin{align*}
        \grad f_{\mu}(x) = \mathfrak{g}_{u}(x)\cdot x + \frac{\|x\|^2}{2} \cdot \grad \mathfrak{g}_{u}(x) - \frac{\|x - u\|^2}{2}\cdot \grad \mathfrak{g}_{u}(x)+ (1-\mathfrak{g}_{u}(x))\cdot (x-u)
    \end{align*}
    and consequently,
    \begin{align}
        \grad^2 f_{\mu}(x) &= \underbrace{x\cdot \grad \mathfrak{g}_{u}(x)^{\top}}_{(a_1)} + \underbrace{\grad \mathfrak{g}_{u}(x)\cdot  x^{\top}}_{(b_1)} + \underbrace{\mathfrak{g}_{u}(x)\cdot \!{Id}_d}_{(c_1)} + \underbrace{\frac{\|x\|^2}{2}\cdot \grad^2 \mathfrak{g}_{u}(x)}_{(d_1)} \notag \\
        &\quad - \underbrace{(x-u)\cdot \grad \mathfrak{g}_{u}(x)^{\top}}_{(a_2)} - \underbrace{\grad \mathfrak{g}_{u}(x)\cdot (x-u)^{\top}}_{(b_2)} + \underbrace{(1-\mathfrak{g}_{u}(x))\cdot \!{Id}_d}_{(c_2)} -\underbrace{\frac{\|x - u\|^2}{2}\cdot \grad^2 \mathfrak{g}_{u}(x)}_{(d_2)}. \label{eq:s1}
    \end{align}

    By the definition of $\mathfrak{g}_{u}(x)$, we have
     \[
        % \grad \mathfrak{g}_{u}(x) = \frac{10(x-u)}{\|x-u\|\cdot \|u\|}\cdot  \frac{\dd q}{\dd y}\Bigg|_{y=10\tp{\frac{\|x-u\|}{\|u\|}-0.4}}
        \grad \mathfrak{g}_{u}(x) = \frac{10(x-u)}{\|x-u\|\cdot \|u\|}\cdot  q'_{\!{mol}}\tp{10\tp{\frac{\|x-u\|}{\|u\|}-0.4}}
        % \grad h\tp{10\tp{\frac{\|x-2u\|}{\|u\|}-1.5}}
     \]
     and 
     \begin{align*}
        %  \grad^2 \mathfrak{g}_{u}(x) &= \frac{100(x-u)(x-u)^{\top}}{\|x-u\|^2\cdot \|u\|^2}\cdot  \frac{\dd^2 q}{\dd y^2}\Bigg|_{y=10\tp{\frac{\|x-u\|}{\|u\|}-0.4}} \\
        %  &\quad + \frac{10}{\|u\|}\tp{\frac{\!{Id}_d}{\|x-u\|} - \frac{(x-u)(x-u)^{\top}}{\|x-u\|^3}} \cdot  \frac{\dd q}{\dd y}\Bigg|_{y=10\tp{\frac{\|x-u\|}{\|u\|}-0.4}}.
        \grad^2 \mathfrak{g}_{u}(x) &= \frac{100(x-u)(x-u)^{\top}}{\|x-u\|^2\cdot \|u\|^2}\cdot  q''_{\!{mol}}\tp{10\tp{\frac{\|x-u\|}{\|u\|}-0.4}} \\
        &\quad + \frac{10}{\|u\|}\tp{\frac{\!{Id}_d}{\|x-u\|} - \frac{(x-u)(x-u)^{\top}}{\|x-u\|^3}} \cdot  q''_{\!{mol}}\tp{10\tp{\frac{\|x-u\|}{\|u\|}-0.4}}.
     \end{align*}

     Then we calculate the terms in \Cref{eq:s1} one by one. We have 
     \[ 
        (a_1) = \frac{10x(x-u)^{\top}}{\|x-u\|\cdot \|u\|}\cdot  q'_{\!{mol}}\tp{10\tp{\frac{\|x-u\|}{\|u\|}-0.4}}.
    \]
    Since $x\in 
    \+B_{\frac{r}{2}}(u)\setminus \+B_{\frac{2r}{5}}(u)$, we have $\|x-u\|=\Theta\tp{\|u\|}$ and $\|x\|=\Theta\tp{\|u\|}$. Recall that $q'_{\!{mol}}=\+O(1)$. So we have $\+O(1)\cdot \!{Id}_d\mge (a_1)\mge -\+O(1)\cdot \!{Id}_d$. We can also prove such bounds for $(a_2), (b_1)$ and $(b_2)$ in the same way.

    For the terms $(c_1)$ and $(c_2)$, we know that $\mathfrak{g}_{u}(x)\in [0,1]$ for all $x\in \+B_{\frac{r}{2}}(u)\setminus \+B_{\frac{2r}{5}}(u)$. Therefore, we have $\+O(1)\cdot \!{Id}_d\mge (c_1)\mge 0$ and $\+O(1)\cdot \!{Id}_d\mge (c_2)\mge 0$.

    For $x\in \+B_{\frac{r}{2}}(u)\setminus \+B_{\frac{2r}{5}}(u)$, $\|x\|^2=\Theta(\|u\|^2)$ and $\|x-u\|^2=\Theta(\|u\|^2)$. Since $q''_{\!{mol}}$ and $q'_{\!{mol}}$ are all $\+O(1)$, we have that $\+O\tp{\frac{1}{\|u\|^2}}\cdot \!{Id}_d\mge \grad^2 \mathfrak{g}_{u}(x) \mge 0$. Therefore, $\+O(1)\cdot \!{Id}_d\mge (d_1)\mge 0$ and $\+O(1)\cdot \!{Id}_d\mge (d_2)\mge 0$.

    Combining all these together, we know that $\log p_{\mu}$ is $\+O(1)$-smooth.
\end{proof}

\subsubsection{The smoothness of $\log p_t$}\label{subsec:stitched2}

We then show that for arbitrary $u\in \bb R^d$ with $\|u\|^2\geq 100d$, when $t$ satisfies $e^{-2t}<0.1$, there exists $x_0\in \bb R^d$ such that $\|\grad^2 \log p_t(x_0)\|_{\!{op}}\geq \Omega\tp{e^{-2t}\|u\|^2-1}$.
% $\log p_t(x_0) \not \mle L\cdot \!{Id}_d$ for any $L=o(\max\ab\{e^{-2t}\|u\|^2,1\})$. 
Combining this with the results in \Cref{subsec:stitched1}, it indicates that the smoothness bound for $\log p_t$ can be much larger than that for $\log p_{\mu}$. In other words, the smoothness of the initial distribution does not necessarily imply the smoothness during the OU process. 

We first see how the distribution evolves during the process. From \Cref{eq:OU}, for any $x\in \bb R^d$,
\[
    p_t(x) \propto \int_{\bb R^d} p\tp{\frac{y}{e^{-t}}} \cdot e^{-\frac{\|x-y\|^2}{2(1-e^{-2t})}} \dd y.
\]
For a fixed $x_0\in \bb R^d$, consider the distribution $\nu_t$ with density $q_t(y)\propto p_{\mu}\tp{\frac{y}{e^{-t}}} \cdot e^{-\frac{\|x_0-y\|^2}{2(1-e^{-2t})}}$ for each $y\in \bb R^d$. To calculate the Hessian of $\log p_t$, we use \Cref{prop:stitched-decomp}.
\begin{proposition}[Corollary of Lemma 22 in \cite{CLL23}]\label{prop:stitched-decomp}
    We have $\grad^2 \log p_t(x_0) = \frac{1}{(1-e^{-2t})^2}\cdot \!{Cov}_{Y\sim \nu_t}[Y] - \frac{\!{Id}_d}{1-e^{-2t}}$.
\end{proposition}

In the remaining part of this section, we aim to prove the following lemma.
\begin{theorem}\label{thm:stitched2}
    For arbitrary $u\in \bb R^d$ with $\|u\|^2\geq 100d$, when $t$ satisfies $e^{-2t}<0.1$ and when $x_0=\frac{e^{-t}u}{2}$, $\|\!{Cov}_{Y\sim \nu_t}[Y]\|_{\!{op}} = \Omega\tp{e^{-2t}\|u\|^2}$ and consequently, $\|\grad^2 \log p_t(x_0)\|_{\!{op}}\geq \Omega\tp{e^{-2t}\|u\|^2-1}$.
    % for any $L=o(\max\ab\{e^{-2t}\|u\|^2,1\})$, $\!{Cov}_{Y\sim \nu_t}[Y] \not\mle L \cdot \!{Id}_d$ and consequently, $\grad^2 \log p_t(x_0)\not \mle L\cdot \!{Id}_d$.
\end{theorem}

We will always assume $x_0=\frac{e^{-t}u}{2}$ in the following analyses. Before calculating $\!{Cov}_{Y\sim \nu_t}[Y]$, we see some basic properties of the distribution $\nu_t$. 

Define $c_t = \frac{d}{2}\log \tp{2\pi \sigma_t^2}$ where $\sigma_t^2 =  e^{-2t} (1-e^{-2t})$. Let $\+N_1$ and $\+N_2$ represent the two Gaussian distributions $\+N\tp{\frac{e^{-3t}u}{2}, \sigma_t^2 \cdot \!{Id}_d}$ and $\+N\tp{e^{-t}\tp{1-\frac{e^{-2t}}{2}}u, \sigma_t^2 \cdot \!{Id}_d}$ respectively. Let $h_1$ and $h_2$ be the potential function of these two Gaussian distributions. That is
\[
    h_1(y) = \frac{\norm{y-\frac{e^{-3t}u}{2}}^2}{2e^{-2t}(1-e^{-2t})} + c_t \quad \mbox{and}\quad h_2(y) = \frac{\norm{y- \tp{e^{-t}\tp{1-\frac{e^{-2t}}{2}}u }}^2}{2e^{-2t}(1-e^{-2t})} + c_t.
\]
Define $f_t:\bb R^d \to \bb R$ as 
\begin{align*}
    f_t(y) &= \mathfrak{g}_{u}\tp{\frac{y}{e^{-t}}}\cdot h_1(y) + \tp{1-\mathfrak{g}_{u}\tp{\frac{y}{e^{-t}}}}\cdot  h_2(y). 
\end{align*}

\begin{lemma}\label{lem:stitched0}
    With $x_0=\frac{e^{-t}u}{2}$, we have $q_t(y) \propto e^{-f_t(y)}$.
\end{lemma}

The proof of \Cref{lem:stitched0} is provided in \Cref{subsec:proof2}. Recall that $\+B_{\frac{r}{2}}(u) = \set{x\in \bb R^d: \| x - u\| \leq 0.5 \|u\|}$ and $\+B_{\frac{2r}{5}}(u) = \set{x\in \bb R^d: \| x - u\| \leq 0.4 \|u\|}$. For a set $S\subseteq \bb R^d$ and a real number $c\neq 0$, let $c\cdot S = \set{x\in \bb R^d: \frac{x}{c}\in S}$ be the scaled set. Therefore, outside the ball $e^{-t}\cdot \+B_{\frac{r}{2}}(u)$, $\mathfrak{g}_{u}\tp{\frac{y}{e^{-t}}}=1$ and $f_t(y)\equiv h_1(y)$. Inside $e^{-t}\cdot \+B_{\frac{r}{2}}(u)$, the potential function $f_t$ is the interpolating of $h_1$ and $h_2$ and when $y\in e^{-t}\cdot \+B_{\frac{2r}{5}}(u)$, $f_t(y)\equiv h_2(y)$. 

We claim that the density $q_t$ can be decomposed into a Gaussian density $e^{-h_1}$ plus some density function $p_{\gamma_t}$ which is only supported on $e^{-t}\cdot \+B_{\frac{r}{2}}(u)$. That is, we can find $\delta_t\in (0,1)$ and a distribution $\gamma_t$ with density $p_{\gamma_t}$ such that for any $y\in \bb R^d$, $ q_t(y) = (1-\delta_t) e^{-h_1(y)} + \delta_t\cdot p_{\gamma_t}(y)$. The existence of such $\delta_t$ and $\gamma_t$ is guaranteed by the properties given in the following two lemmas. The proofs of \Cref{lem:stitched3,lem:normalizing} are given in \Cref{subsec:proof2}.

\begin{lemma}\label{lem:stitched3}
    For any $y\in e^{-t}\cdot \+B_{\frac{r}{2}}(u)$, $e^{-h_1(y)} \leq e^{-f_t(y)}$.
\end{lemma}
Let $Z_t=\int_{\bb R^d} e^{-f_t(y)} \dd y$ be the normalizing factor.
\begin{lemma}\label{lem:normalizing}
    When $e^{-2t}<0.1$ and $\|u\|^2\geq 100d$, we have $1.8\leq Z_t\leq 2$.
\end{lemma}

Furthermore, we can prove that $\delta_t=\Theta(1)$ for any $t$ satisfying $e^{-2t}<0.1$.

\begin{lemma}\label{lem:stitched4}
    We have $0.3 < \delta_t < 0.7$ when $e^{-2t}<0.1$ and $\|u\|^2\geq 100d$.
\end{lemma}

We prove \Cref{lem:stitched4} in \Cref{subsec:proof2}. Equipped with these lemmas, we are now ready to prove \Cref{thm:stitched2}.
\begin{proof}[Proof of \Cref{thm:stitched2}]
    We first decompose $\!{Cov}_{Y\sim \nu_t}[Y]$.
    By definition,
    \begin{align*}
        \!{Cov}_{Y\sim \nu_t}[Y] &= \E[Y\sim \nu_t]{(Y-\E[\nu_t]{Y})(Y-\E[\nu_t]{Y})^{\top}} \\
        &= (1-\delta_t)\cdot \E[Y\sim \+N_1]{(Y-\E[\nu_t]{Y})(Y-\E[\nu_t]{Y})^{\top}} + \delta_t\cdot \E[Y\sim \gamma_t]{(Y-\E[\nu_t]{Y})(Y-\E[\nu_t]{Y})^{\top}} \\
        &= (1-\delta_t)\cdot \E[Y\sim \+N_1]{(Y-\E[\+N_1]{Y})(Y-\E[\+N_1]{Y})^{\top}} \\
        &\quad + (1-\delta_t) \cdot \tp{\E[\+N_1]{Y} - \E[\nu_t]{Y}}\tp{\E[\+N_1]{Y} - \E[\nu_t]{Y}}^{\top} \\
        &\quad  + \delta_t \cdot \E[Y\sim \gamma_t]{(Y-\E[\gamma_t]{Y})(Y-\E[\gamma_t]{Y})^{\top}} + \delta_t \cdot \tp{\E[\gamma_t]{Y} - \E[\nu_t]{Y}}\tp{\E[\gamma_t]{Y} - \E[\nu_t]{Y}}^{\top} \\
        &= (1-\delta_t)\cdot\!{Cov}_{\+N_1}[Y] + (1-\delta_t) \cdot \tp{\E[\+N_1]{Y} - \E[\nu_t]{Y}}\tp{\E[\+N_1]{Y} - \E[\nu_t]{Y}}^{\top}\\
        &\quad + \delta_t\cdot\!{Cov}_{\gamma_t}[Y] + \delta_t \cdot \tp{\E[\gamma_t]{Y} - \E[\nu_t]{Y}}\tp{\E[\gamma_t]{Y} - \E[\nu_t]{Y}}^{\top}.
    \end{align*}
    Since $q_t(y) = (1-\delta_t)\cdot e^{-h_1(y)} + \delta_t \cdot p_{\gamma_t}(y)$, the expectation $\E[\nu_t]{Y} = (1-\delta_t)\cdot \E[\+N_1]{Y} + \delta_t\cdot \E[\gamma_t]{Y}$. Then we have 
    \[
        \E[\nu_t]{Y} - \E[\+N_1]{Y} = \delta_t\cdot \tp{\E[\gamma_t]{Y} - \E[\+N_1]{Y}} 
    \]
    and
    \[
        \E[\nu_t]{Y} - \E[\gamma_t]{Y} = (1-\delta_t)\cdot \tp{\E[\+N_1]{Y} - \E[\gamma_t]{Y}}.
    \]
    Therefore,
    \begin{align*}
        \!{Cov}_{Y\sim \nu_t}[Y] &= (1-\delta_t)\cdot\!{Cov}_{\+N_1}[Y] + \delta_t\cdot\!{Cov}_{\gamma_t}[Y] \\
        &\quad + \tp{\delta_t^2(1-\delta_t) + (1-\delta_t)^2\delta_t} \cdot \tp{\E[\+N_1]{Y} - \E[\gamma_t]{Y}}\tp{\E[\+N_1]{Y} - \E[\gamma_t]{Y}}^{\top} \\
        & = (1-\delta_t)\cdot \!{Cov}_{\+N_1}[Y] + \delta_t\cdot\!{Cov}_{\gamma_t}[Y] + \delta_t(1-\delta_t) \cdot \tp{\E[\gamma_t]{Y} - \frac{e^{-3t}u}{2}}\cdot \tp{\E[\gamma_t]{Y} - \frac{e^{-3t}u}{2}}^\top.
    \end{align*}

    We then show that the matrix $\tp{\E[\gamma_t]{Y} - \frac{e^{-3t}u}{2}}\cdot \tp{\E[\gamma_t]{Y} - \frac{e^{-3t}u}{2}}^\top$ cannot be bouned by $L\cdot \!{Id}_d$ for any $L=o(e^{-2t}\|u\|^2)$. To show this, we only need to prove $\norm{\E[\gamma_t]{Y} - \frac{e^{-3t}u}{2}}^2 = \Omega\tp{e^{-2t}\|u\|^2}$. Recall that $\gamma_t$ is supported only on $ e^{-t}\cdot \+B_{\frac{r}{2}}(u) = \set{y\in \bb R^d: \| y - e^{-t}u\| \leq 0.5 e^{-t} \|u\|}$. For each $y\in e^{-t}\cdot \+B_{\frac{r}{2}}(u)$, with $e^{-2t}<0.1$,
    \[
        \norm{ y- \frac{e^{-3t}u}{2} } \geq \|y\| - \norm{\frac{e^{-3t}u}{2}} \geq \frac{e^{-t}}{2} \|u\| - \frac{e^{-3t}}{2} \|u\| > \frac{e^{-t}}{4} \|u\|.
    \]
    So $\norm{\E[\gamma_t]{Y} - \frac{e^{-3t}u}{2}}^2$ can be lower bounded by 
    \[
        \inf_{y\in  e^{-t}\cdot \+B_{\frac{r}{2}}(u)} \norm{y - \frac{e^{-3t}u}{2}}^2  \geq \frac{e^{-2t}}{16} \|u\|^2 = \Omega\tp{e^{-2t}\|u\|^2}.
    \]
    
    Since $\!{Cov}_{\+N_1}[Y]\mge 0$, $\!{Cov}_{\gamma_t}[Y]\mge 0$ and from \Cref{lem:stitched4}, $\delta_t=\Theta(1)$, this indicates that $\|\!{Cov}_{Y\sim \nu_t}[Y]\|_{\!{op}} = \Omega\tp{e^{-2t}\|u\|^2}$.
    % this indicates that for any $L=o(e^{-2t}\|u\|^2)$, $\!{Cov}_{Y\sim \nu_t}[Y] \not\mle L \cdot \!{Id}_d$. 
    The remaining of the lemma then follows from \Cref{prop:stitched-decomp}.
\end{proof}

\subsubsection{Proofs for the supporting lemmas in \Cref{subsec:stitched2}}\label{subsec:proof2}
\begin{proof}[Proof of \Cref{lem:stitched0}]
    By definition, for any $y\in \bb R^d$,
    \begin{align}
        q_t(y) &\propto \exp\set{-\mathfrak{g}_{u}\tp{\frac{y}{e^{-t}}}\cdot \tp{\frac{\|y\|^2}{2e^{-2t}} +\frac{\|y-x_0\|^2}{2(1-e^{-2t})}} - \tp{1-\mathfrak{g}_{u}\tp{\frac{y}{e^{-t}}}}\cdot \tp{\frac{\|y-e^{-t}\cdot u\|^2}{2e^{-2t}} + \frac{\|y-x_0\|^2}{2(1-e^{-2t})}} } \notag \\
        &= \exp\left\{-\mathfrak{g}_{u}\tp{\frac{y}{e^{-t}}}\cdot \frac{\|y\|^2 + e^{-2t}\|x_0\|^2 - e^{-2t}(y^{\top}x_0 + x_0^{\top}y)}{2e^{-2t}(1-e^{-2t})} \right. \notag \\
        & \quad  \left. - \tp{1-\mathfrak{g}_{u}\tp{\frac{y}{e^{-t}}}}\cdot \frac{\|y\|^2 - (1-e^{-2t})e^{-t}(u^{\top}y + y^{\top}u) - e^{-2t}(x_0^{\top}y + y^{\top}x_0) + (1-e^{-2t})e^{-2t}\|u\|^2 + e^{-2t}\|x_0\|^2}{2e^{-2t}(1-e^{-2t})} \right\} \notag  \\
        &= \exp\left\{-\mathfrak{g}_{u}\tp{\frac{y}{e^{-t}}}\cdot \tp{\frac{\|y-e^{-2t}x_0\|^2}{2e^{-2t}(1-e^{-2t})} + \frac{\|x_0\|^2}{2}} \right. \notag \\
        & \quad \left.- \tp{1-\mathfrak{g}_{u}\tp{\frac{y}{e^{-t}}}}\cdot \tp{\frac{\|y- \tp{e^{-t}(1-e^{-2t})u + e^{-2t}x_0}\|^2}{2e^{-2t}(1-e^{-2t})} + \frac{\|x_0 - e^{-t}u\|^2}{2}}\right\}.\label{eq:stitched1}
    \end{align}
    When we choose $x_0 = \frac{e^{-t}u}{2}$, we can further simplify \Cref{eq:stitched1} as
    \begin{align*}
        q_t(y) & \propto \exp\left\{-\mathfrak{g}_{u}\tp{\frac{y}{e^{-t}}}\cdot \tp{\frac{\norm{y-\frac{e^{-3t}u}{2}}^2}{2e^{-2t}(1-e^{-2t})} + \frac{\|e^{-t}u\|^2}{8}} \right. \notag \\
        & \quad \left.- \tp{1-\mathfrak{g}_{u}\tp{\frac{y}{e^{-t}}}}\cdot \tp{\frac{\norm{y- \tp{e^{-t}\tp{1-\frac{e^{-2t}}{2}}u }}^2}{2e^{-2t}(1-e^{-2t})} + \frac{\|e^{-t}u\|^2}{8}}\right\} \notag  \\
        &\propto \exp\left\{-\mathfrak{g}_{u}\tp{\frac{y}{e^{-t}}}\cdot \frac{\norm{y-\frac{e^{-3t}u}{2}}^2}{2e^{-2t}(1-e^{-2t})}  - \tp{1-\mathfrak{g}_{u}\tp{\frac{y}{e^{-t}}}}\cdot \frac{\norm{y- \tp{e^{-t}\tp{1-\frac{e^{-2t}}{2}}u }}^2}{2e^{-2t}(1-e^{-2t})} \right\}  \notag \\
        &\propto \exp\left\{-\mathfrak{g}_{u}\tp{\frac{y}{e^{-t}}}\cdot \tp{\frac{\norm{y-\frac{e^{-3t}u}{2}}^2}{2e^{-2t}(1-e^{-2t})} + c_t} - \tp{1-\mathfrak{g}_{u}\tp{\frac{y}{e^{-t}}}}\cdot \tp{\frac{\norm{y- \tp{e^{-t}\tp{1-\frac{e^{-2t}}{2}}u }}^2}{2e^{-2t}(1-e^{-2t})} + c_t}\right\}
    \end{align*}
\end{proof}

\begin{proof}[Proof of \Cref{lem:stitched3}]
    Recall that for any $y\in \bb R^d$,
    \[
    h_1(y) = \frac{\norm{y-\frac{e^{-3t}u}{2}}^2}{2e^{-2t}(1-e^{-2t})} + c_t \quad \mbox{and}\quad h_2(y) = \frac{\norm{y- \tp{e^{-t}\tp{1-\frac{e^{-2t}}{2}}u }}^2}{2e^{-2t}(1-e^{-2t})} + c_t.
    \]
    To prove this lemma, we need to show that
    \[
        \norm{y-\frac{e^{-3t}u}{2}}^2 \geq \norm{y- \tp{e^{-t}\tp{1-\frac{e^{-2t}}{2}}u }}^2
    \]
    for each $y\in e^{-t}\cdot \+B_{\frac{r}{2}}(u)$. This is equivalent to say
    \begin{equation*}
        \inner{\tp{2e^{-t} - 2e^{-3t}}u}{y} \geq e^{-2t}(1-e^{-2t}) \|u\|^2,
    \end{equation*}
    and this can be further simplified to $2\inner{u}{y} \geq e^{-t}\|u\|^2$.

    Since for $y\in e^{-t}\cdot \+B_{\frac{r}{2}}(u)$, $y$ satisfies $\| y - e^{-t}u\| \leq 0.5 e^{-t} \|u\|$. Therefore, for each $y\in e^{-t}\cdot\+B_{\frac{r}{2}}(u)$, we have
    \begin{equation*}
        \|y\| \geq 0.5 e^{-t}\|u\| \quad \mbox{and} \quad 2e^{-t}\inner{u}{y} \geq \|y\|^2 + 0.75 e^{-2t} \|u\|^2.
    \end{equation*}
    This indicates that $2\inner{u}{y} \geq e^{-t}\|u\|^2$.
\end{proof}

\begin{proof}[Proof of \Cref{lem:normalizing}]
    On the one hand,
    \begin{align*}
        Z_t \leq \int_{\bb R^d} e^{-h_1(y)} \dd y + \int_{\bb R^d} e^{-h_2(y)} \dd y \leq 2.
    \end{align*}
    On the other hand,
    \begin{align*}
        Z_t &\geq \int_{e^{-t}\cdot \+B_{\frac{2r}{5}}(u)} e^{-h_2(y)} \dd y + \int_{e^{-t}\cdot \ol{\+B_{\frac{r}{2}}(u)}} e^{-h_1(y)} \dd y \\
        & = \Pr[Y\sim \+N_2]{Y\in e^{-t}\cdot \+B_{\frac{2r}{5}}(u)} + \Pr[Y\sim \+N_1]{Y\in e^{-t}\cdot \ol{\+B_{\frac{r}{2}}(u)}}\\
        &\geq 2- \frac{20 d}{\|u\|^2} >1.8.
    \end{align*}
    where in the second inequality we use \Cref{prop:stitched0} and the last inequality is due to $\|u\|^2\geq 100 d$.
\end{proof}

Before proving \Cref{lem:stitched4}, we first give the following concentration bounds for $\+N_1$ and $\+N_2$. Recall that $\+N_1$ and $\+N_2$ represent the two Gaussian distributions $\+N\tp{\frac{e^{-3t}u}{2}, \sigma_t^2 \cdot \!{Id}_d}$ and $\+N\tp{e^{-t}\tp{1-\frac{e^{-2t}}{2}}u, \sigma_t^2 \cdot \!{Id}_d}$ respectively. 
\begin{proposition}\label{prop:stitched0}
    When $e^{-2t}<0.1$, we have $\Pr[Y\sim \+N_1]{Y\in e^{-t}\cdot \ol{\+B_{\frac{r}{2}}(u)}} \geq 1 - \frac{10d}{\|u\|^2}$ and $\Pr[Y\sim \+N_2]{Y\in e^{-t}\cdot \+B_{\frac{2r}{5}}(u)} \geq 1- \frac{10d}{\|u\|^2}$.
\end{proposition}
\begin{proof}
    We first prove $\Pr[Y\sim \+N_1]{Y\in e^{-t}\cdot \ol{\+B_{\frac{r}{2}}(u)}} \geq 1 - \frac{10d}{\|u\|^2}$. By definition, 
    \[
        e^{-t}\cdot \ol{\+B_{\frac{r}{2}}(u)} = \set{y\in \bb R^d: \|y - e^{-t}u\| > 0.5e^{-t}\|u\|}.
    \]
    Since 
    \[
        \|y - e^{-t}u\| \geq \tp{e^{-t} - \frac{e^{-3t}}{2}} \|u\| - \norm{y - \frac{e^{-3t}}{2} u },
    \]
    we have $\set{y\in \bb R^d:\ \norm{y - \frac{e^{-3t}}{2} u } < \tp{0.5 e^{-t} - \frac{e^{-3t}}{2}}\|u\|} \subseteq e^{-t}\cdot \ol{\+B_{\frac{r}{2}}(u)}$. Therefore, 
    \begin{align*}
        \Pr[Y\sim \+N_1]{Y\in e^{-t}\cdot \ol{\+B_{\frac{r}{2}}(u)}} &\geq \Pr[Y\sim \+N_1]{\norm{Y - \frac{e^{-3t}}{2} u }^2 < \tp{0.5 e^{-t} - \frac{e^{-3t}}{2}}^2\|u\|^2}  \\
        &\geq 1 - \frac{d\cdot \sigma_t^2}{\tp{0.5 e^{-t} - \frac{e^{-3t}}{2}}^2\|u\|^2} \\
        &\geq 1 - \frac{10d}{\|u\|^2}
    \end{align*}
    where the second inequality follows from the Markov's inequality and the last inequality is due to $e^{-2t}< 0.1$.

    We then prove that $\Pr[Y\sim \+N_2]{Y\in e^{-t}\cdot \+B_{\frac{2r}{5}}(u)} \geq 1- \frac{10d}{\|u\|^2}$. By definition, 
    \[
        e^{-t}\cdot \+B_{\frac{2r}{5}}(u) = \set{y\in \bb R^d: \|y - e^{-t}u\| < 0.4e^{-t}\|u\|}.
    \]
    Since
    \[
        \|y - e^{-t}u\| \leq \norm{y -  e^{-t}\tp{1-\frac{e^{-2t}}{2}}u} + \frac{e^{-3t}}{2}\|u\|,
    \]
    we have $\set{y\in \bb R^d :\ \norm{ y -  e^{-t}\tp{1-\frac{e^{-2t}}{2}}u}\leq\tp{0.4 e^{-t} - \frac{e^{-3t}}{2}}\|u\|} \subseteq e^{-t}\cdot \+B_{\frac{2r}{5}}(u)$. Similarly,
    \begin{align*}
        \Pr[Y\sim \+N_2]{Y\in e^{-t}\cdot \+B_{\frac{2r}{5}}(u)} &\geq \Pr[Y\sim \+N_2]{\norm{Y -  e^{-t}\tp{1-\frac{e^{-2t}}{2}}}^2 \leq\tp{0.4 e^{-t} - \frac{e^{-3t}}{2}}^2\|u\|^2}  \\
        &\geq 1- \frac{d\cdot \sigma_t^2}{\tp{0.4 e^{-t} - \frac{e^{-3t}}{2}}^2\|u\|^2} \\
        &\geq 1- \frac{10d}{\|u\|^2}.
    \end{align*}
\end{proof}

Now we give a proof of \Cref{lem:stitched4}.
\begin{proof}[Proof of \Cref{lem:stitched4}]
    Recall that the distribution $\gamma_t$ is only supported on $e^{-t}\cdot \+B_{\frac{r}{2}}(u)$. The choice of $\delta_t$ satisfies
    \begin{equation}
        \Pr[Y\sim \nu_t]{Y\in e^{-t}\cdot \+B_{\frac{r}{2}}(u)} = (1-\delta_t)\cdot \Pr[Y\sim N_1]{Y\in e^{-t}\cdot \+B_{\frac{r}{2}}(u)} + \delta_t. \label{eq:decomp}
    \end{equation}
    From \Cref{prop:stitched0} and \Cref{lem:normalizing}, we have
    \[
        \Pr[Y\sim \nu_t]{Y\in e^{-t}\cdot \+B_{\frac{r}{2}}(u)} \geq \frac{1}{Z_t}\cdot \Pr[Y\sim \+N_2]{Y\in e^{-t}\cdot \+B_{\frac{2r}{5}}(u)} \geq \frac{1-\frac{10d}{\|u\|^2}}{2} >0.4,
    \]
    and 
    \[
        \Pr[Y\sim \nu_t]{Y\in e^{-t}\cdot \+B_{\frac{r}{2}}(u)} \leq \frac{1}{Z_t}\cdot\tp{\Pr[Y\sim \+N_2]{Y\in e^{-t}\cdot \+B_{\frac{r}{2}}(u)} + \Pr[Y\sim \+N_1]{Y\in e^{-t}\cdot \+B_{\frac{r}{2}}(u)}} \leq \frac{1}{Z_t}\cdot \tp{1 + \frac{10d}{\|u\|^2}} < 0.7.
    \]
    
    From \Cref{prop:stitched0}, $\Pr[Y\sim N_1]{Y\in e^{-t}\cdot \+B_{\frac{r}{2}}(u)} \leq \frac{10d}{\|u\|^2} \leq 0.1$. So the RHS of \Cref{eq:decomp} satisfies
    \[
        \delta_t \leq (1-\delta_t)\cdot \Pr[Y\sim N_1]{Y\in e^{-t}\cdot \+B_{\frac{r}{2}}(u)} + \delta_t \leq 0.1 + 0.9 \delta_t.
    \]
    Combining above inequalities, we have $0.3< \delta_t < 0.7$.
\end{proof}

\subsection{Smoothness for the mixture of Gaussian distributions}\label{subsec:mix}

The analysis in the previous section indicates that
% that the smoothness parameter can become larger during the OU process compared to the stitched Gaussian distribution at the start. 
an $\+O(1)$-log-smooth initial distribution does not guarantee $\+O(1)$-smoothness after the process evolves. 
% Therefore, it is worth  investigating what kind of initial distributions can preserve such desirable smoothness properties.
In this subsection, we consider a family of classic multi-modal distributions, the mixture of Gaussian distributions. Mixture of Gaussians appear to be similar to stitched Gaussians, but the analysis below will reveal fundamental differences in their smoothness behaviors. 

We study the cases where each Gaussian component has a covariance $\Sigma_i \succeq \Omega(1)\!{Id}_d$. 
% Let $p$ and $p_t$ denote the density functions of the initial distribution and the distribution at time $t$ during the OU process respectively. 
We first considered the simple case with only two components, and then extend the analysis to mixtures with multiple components.

% To be specific, we show that for a mixture of two Gaussians with mean $u_1$ and $u_2$, if their covariance matrices are the same, the distributions at the beginning and during the process will always be $\+O\tp{\max\ab\{e^{-2t}\|u_1-u_2\|^2,1\}}$-smooth. In contrast, when the covariance matrices differ, even with $\|u_1-u_2\|$ being a constant, $\log p$ is not $L$-smooth for any $L=o(d)$. 

% However, for those cases with more components, the analysis becomes more complex. Even when all covariance matrices are the same, it is challenging to derive a concise rule to characterize the relationship between smoothness and the distances between the means of the Gaussians. We give an example where the centers of components are far apart, yet the mixture distribution remains $\+O(1)$-smooth, and this smoothness is preserved during the OU process.

\subsubsection{Mixture of Gaussian distributions and its evolution during the OU process}

Consider $m$ Gaussian distributions over $\bb R^d$, each with mean $u_i\in \bb R^d$ and covariance $\Sigma_i\in \bb R^{d\times d}$. Define function $f_i:\bb R^d\to \bb R$ as $f_i(x) = \frac{1}{2}(x-u_i)^{\top}\Sigma_i^{-1}(x-u_i) + \frac{1}{2}\log\tp{\tp{2\pi}^{d} \abs{\Sigma_i}}$ for each $x\in \bb R^d$. Then $e^{-f_i}$ is the density function of the Gaussian distribution $\+N\tp{u_i,\Sigma_i}$. Suppose $\mu$ is the mixture of Gaussian with density $p_{\mu}(x) = \sum_{i=1}^m w_i e^{-f_i(x)}$, where $w_i\in(0,1)$ is the weight of the $i$-th component and $\sum_{i=1}^m w_i=1$. Then $-\grad \log p_{\mu}(x) = \frac{\sum_{i=1}^m w_i\grad f_i(x)\cdot e^{-f_i(x)}}{p_{\mu}(x)}$ and 
\begin{align}
    -\grad^2 \log p_{\mu}(x) &= \frac{\sum_{i=1}^m w_i\grad^2 f_i(x)\cdot e^{-f_i(x)} - \sum_{i=1}^m w_i\grad f_i(x)\grad f_i(x)^{\top}\cdot e^{-f_i(x)}}{\sum_{i=1}^m w_i e^{-f_i(x)}} \notag  \\
    &\quad + \frac{\tp{\sum_{i=1}^m w_i\grad f_i(x)\cdot e^{-f_i(x)}}\tp{\sum_{i=1}^m w_i\grad f_i(x)\cdot e^{-f_i(x)}}^{\top}}{\tp{\sum_{i=1}^m w_i e^{-f_i(x)}}^2} \notag \\
    &= \frac{\sum_{i=1}^m\sum_{j=1}^m w_iw_j \tp{\grad f_i(x) \grad f_j(x)^{\top} - \frac{1}{2}\grad f_i(x) \grad f_i(x)^{\top} - \frac{1}{2}\grad f_j(x) \grad f_j(x)^{\top}} e^{-f_i(x)-f_j(x)}}{\tp{\sum_{i=1}^m w_i e^{-f_i(x)}}^2} \notag \\
    &\quad + \frac{\sum_{i=1}^m w_i\grad^2 f_i(x)\cdot e^{-f_i(x)} }{\sum_{i=1}^m w_i e^{-f_i(x)}} \notag \\
    &= - \underbrace{\frac{\sum_{1\leq i<j\leq m} w_iw_j \tp{\grad f_i(x) - \grad f_j(x)}\tp{\grad f_i(x) - \grad f_j(x)}^{\top} e^{-f_i(x)-f_j(x)}}{\tp{\sum_{i=1}^m w_i e^{-f_i(x)}}^2}}_{(A)} \notag \\
    &\quad + \underbrace{\frac{\sum_{i=1}^m w_i\grad^2 f_i(x)\cdot e^{-f_i(x)} }{\sum_{i=1}^m w_i e^{-f_i(x)}}}_{(B)}. \label{eq:smooth1}
\end{align}
From \Cref{eq:smooth1}, $\grad^2 \log p_{\mu}$ is determined by two parts: the weighted mixture of the Hessian of each component (term $(B)$), and the interaction between different components (term $(A)$).

Then we see how this distribution evolves during the OU process. Recall that the trajectory of the OU process is given by $X_t = e^{-t}X_0 + \sqrt{2}\cdot e^{-t}\int_{0}^t e^s \d B_s$ and $
\sqrt{2}\cdot e^{-t}\int_{0}^t e^s \d B_s \sim \+N\tp{0, (1-e^{-2t})I_d}$. If $X_0$ is drawn from some Gaussian distribution $\+N(u,\Sigma)$, at time $t$, the distribution of $X_t$ will be $\+N\tp{e^{-t}u, e^{-2t}\Sigma + (1-e^{-2t})\!{Id}_d}$. Hence, if the initial distribution is $\mu$, i.e., the weighted mixture of $\ab\{\+N(u_i, \Sigma_i)\}_{i\in[m]}$, then the distribution of $X_t$ will be the mixture of $\ab\{\+N\tp{e^{-t}u_i, e^{-2t}\Sigma_i + (1-e^{-2t})\!{Id}_d} \}_{i\in[m]}$ with the same weights. Let $\Sigma_i^{(t)} =  e^{-2t}\Sigma_i + (1-e^{-2t})\!{Id}_d$. That is to say, this distribution is still a mixture of Gaussians and has density $p_t(x) = \sum_{i=1}^n w_i e^{-f^{(t)}_i(x)}$, where $e^{-f^{(t)}_i(x)}$ is the density function of $\+N\tp{e^{-t}u_i, \Sigma_i^{(t)}}$.

\subsubsection{Mixture of two Gaussians}\label{subsubsec:2gaussian}
We first see the case when $m=2$. 

\paragraph{Gaussians with the same covariance: distance of means determines}~

When the covariances are the same and are bounded, the rules are simple and straightforward. The smoothness of the mixture distribution is totally determined by the distance of centers.

\begin{lemma}\label{lem:2-same}
    When $m=2$ and $\Sigma_1 = \Sigma_2 = \Sigma$ for some matrix $\Sigma\succeq \Omega(1)\!{Id}_d$, we have
    \[
        -\+O\tp{\|u_1-u_2\|^2}\cdot \!{Id}_d \preceq -\grad^2 \log p_{\mu}(x) \preceq \Sigma^{-1}, 
    \]
    and 
    \[
         -\+O\tp{e^{-2t}\|u_1-u_2\|^2}\cdot \!{Id}_d \preceq -\grad^2 \log p_t(x) \preceq \+O(1)\!{Id}_d
    \]
    for any $t>0$.
    
    On the other hand, for any $L=o(\|u_1-u_2\|^2)$, $-\grad^2 \log p_{\mu}(x) \not\succeq -L\cdot \!{Id}_d$ and $-\grad^2 \log p_t(x) \not\succeq -e^{-2t}L\cdot \!{Id}_d$.
\end{lemma}
\begin{proof}
    When $m=2$ and $\Sigma_1 = \Sigma_2 = \Sigma$ for some matrix $\Sigma\in \bb R^{d\times d}$, we have $\grad f_1(x) = \Sigma^{-1}(x-u_1)$, $\grad f_2(x) = \Sigma^{-1}(x-u_2)$ and $\grad^2 f_1(x) = \grad^2 f_2(x) = \Sigma^{-1}$.
According to \Cref{eq:smooth1}, 
\begin{align*}
    -\grad^2 \log p_{\mu}(x) = - \frac{w_1w_2\cdot \Sigma^{-1}(u_1-u_2)(u_1-u_2)^{\top}\Sigma^{-1}\cdot e^{-f_1(x) - f_2(x)}}{\tp{w_1 e^{-f_1(x)} + w_2 e^{-f_2(x)}}^2} + \Sigma^{-1}.
\end{align*}
From the Cauchy-Schwartz inequality,
\[
    0 \leq \frac{w_1w_2\cdot e^{-f_1(x) - f_2(x)}}{\tp{w_1 e^{-f_1(x)} + w_2 e^{-f_2(x)}}^2} \leq \frac{1}{4}.
\]
Therefore
\[
    -\frac{1}{4}\Sigma^{-1}(u_1-u_2)(u_1-u_2)^{\top}\Sigma^{-1} + \Sigma^{-1} \preceq -\grad^2 \log p_{\mu}(x) \preceq \Sigma^{-1}.
\]
It can be easily prove that $\Sigma^{-2}$ is also upper bounded by $\+O(1)\!{Id}_d$. The result then follows from the fact that 
\[
    \Sigma^{-1}(u_1-u_2)(u_1-u_2)^{\top}\Sigma^{-1} \preceq \|\Sigma^{-1}(u_1-u_2)\|^2 \cdot \!{Id}_d = (u_1-u_2)^{\top}\Sigma^{-2}(u_1-u_2)\cdot \!{Id}_d \preceq \+O(\|u_1-u_2\|^2)\cdot \!{Id}_d.
\]

% This indicates that the potential function of the initial distribution is at least $\+O\tp{\|u_1-u_2\|^2}$-smooth (?).

Let $\Sigma^{(t)} = e^{-2t}\Sigma + (1-e^{-2t})\!{Id}_d$. For the distributions during the OU process, repeating the above calculations, we can get that
\[
     -\frac{e^{-2t}}{4}\tp{\Sigma^{(t)}}^{-1}(u_1-u_2)(u_1-u_2)^{\top}\tp{\Sigma^{(t)}}^{-1} + \tp{\Sigma^{(t)}}^{-1} \preceq -\grad^2 \log p_t(x) \preceq \tp{\Sigma^{(t)}}^{-1}.
\]
By the definition of $\Sigma^{(t)}$, $ \tp{\Sigma^{(t)}}^{-1} \preceq \frac{1}{1-ce^{-2t}}\cdot \!{Id}_d \preceq \+O(1) \!{Id}_d$ for some universal constant $c<1$. 
Therefore, at time $t$, $-\grad^2 \log p_t \succeq -\+O( \|u_1-u_2\|^2) \!{Id}_d$.
\end{proof}

\paragraph{Gaussians with different covariances}~

When the two components have different covariance matrices, the smoothness parameter can be $\+O(d)$ even when $\|u_1-u_2\|$ is small. Here is an example.

\begin{lemma}\label{lem:2-diff}
    Let $p_{\mu}(x) = \frac{1}{2}e^{-f_1(x)} + \frac{1}{2}e^{-f_2(x)}$ with $f_1(x) = \|x-u_1\|^2 + \frac{d}{2}\log \pi$ and $f_2(x) = \frac{\|x-u_2\|^2}{2} + \frac{d}{2}\log 2\pi$ for arbitrary vectors $u_1,u_2\in \bb R^d$. Then there exists $x\in \bb R^d$ such that $\|\grad^2 \log p_{\mu}(x)\|_{\!{op}} = \Omega{d\log 2 + 2\|u_1-u_2\|^2}$.
    % $\grad^2 \log p_{\mu}(x) \not\preceq L \cdot \!{Id}_d$ for any $L< d\log 2 + 2\|u_1-u_2\|^2$.
\end{lemma}
\begin{proof}
    From \Cref{eq:smooth1}, 
    \begin{align*}
        -\grad^2 \log p_{\mu}(x) = \underbrace{\frac{2e^{-f_1(x)}\cdot \!{Id}_d + e^{-f_2(x)}\cdot \!{Id}_d}{e^{-f_1(x)} + e^{-f_2(x)}}}_{(a)} - \underbrace{\frac{e^{-f_1(x)-f_2(x)}\cdot \tp{x-(2u_1-u_2)}\tp{x-(2u_1-u_2)}^{\top}}{\tp{e^{-f_1(x)} + e^{-f_2(x)}}^2}}_{(b)}.
    \end{align*}
    For $(a)$, it is easy to know that $\!{Id}_d\preceq (a) \preceq 2\!{Id}_d$. For $(b)$, we will find a specific $x$ such that $(b)$ is large.
    When $x$ satisfies $\|x - (2u_1-u_2)\|^2 = d\log 2 + 2\|u_1-u_2\|^2$, by direct calculation, we have 
    \begin{align*}
        f_1(x) - f_2(x) &= \frac{1}{2}\tp{\|x\|^2 - (2u_1-u_2)^{\top}x - x^{\top}(2u_1-u_2) + 2\|u_1\|^2 - \|u_2\|^2 - d\log 2} \\
        &= \frac{1}{2}\tp{\|x - (2u_1-u_2)\|^2 - 4\|u_1\|^2 - \|u_2\|^2 + 2u_1^{\top}u_2 + 2u_2^{\top}u_1 + 2\|u_1\|^2 - \|u_2\|^2 - d\log 2}  \\
        &= \frac{1}{2}\tp{\|x - (2u_1-u_2)\|^2 - d\log 2 - 2\|u_1-u_2\|^2} \\
        &= 0.
    \end{align*}
    Therefore, $0\preceq (b) = \frac{1}{4}\cdot \tp{x-(2u_1-u_2)}\tp{x-(2u_1-u_2)}^{\top}$ and the optimal upper bound for $(b)$ is $\frac{\|x-(2u_1-u_2)\|^2}{4}\cdot \!{Id}_d = \tp{d\log 2 + 2\|u_1-u_2\|^2}\cdot \!{Id}_d$.
\end{proof}

% These results also holds whenever the number of components is a constant.

\subsubsection{Mixture of multiple Gaussians}
Things are more complicated and subtle for the mixture of multiple Gaussian distributions, even when their covariance matrices are the same. When they have the same covariance matrix, via similar arguments in \Cref{lem:2-same}, we can prove that $ -\+O\tp{\max_{i,j\in[m]}\|u_i-u_j\|^2}\cdot \!{Id}_d \preceq-\grad^2 \log p_{\mu}(x) \preceq \+O(1)\!{Id}_d$ and $ -\+O\tp{\max_{i,j\in[m]}e^{-2t}\|u_i-u_j\|^2}\cdot \!{Id}_d \preceq-\grad^2 \log p_t(x) \preceq \+O(1)\!{Id}_d$. However, these bounds may not be tight. A large distance between the centers of the components does not necessarily imply a lack of smoothness. The potential function may still be and maintain $\+O(1)$-smooth during the OU process in this case. We give an example in this section.

Let $J\in \bb R^{d\times d}$ be a symmetric and positive definite matrix and $h$ be an arbitrary vector in $\bb R^d$. Consider the distribution $\mu$ over $\bb R^d$ with density
\begin{equation}
    p_{\mu}(x) \propto \sum_{\sigma\in \ab\{\pm 1\}^d} \exp\set{ - \frac{1}{2}x^{\top}J^{-1}x + \tp{J^{-1}h+\sigma}^{\top}x}. \label{eq:HS-mix}
\end{equation}

This distribution is induced when applying the Hubbard-Stratonovich transform to the Ising model (see Appendix E in \cite{KLR22}). 
% The Ising model with interaction matrix $J$ is a distribution over $\ab\{\pm 1\}^d$ with density $p_J(\sigma)\propto \exp\set{\inner{\sigma}{J\sigma}}$ for any $\sigma\in \ab\{\pm 1\}^d$. 
Note that each $\sigma\in  \ab\{\pm 1\}^d$ corresponds to a Gaussian component $\+N(J\sigma+h, J^{-1})$. For a vector $x\in \bb R^d$, let $x(i)$ denote its $i$-th component for any $i\in[d]$. The following lemma shows that this distribution is log-smooth and even strongly log-concave if $J$ is within a moderate range.

\begin{lemma}\label{lem:mixture1}
    If $\delta\cdot \!{Id}_d\preceq J\preceq (1-\delta)\cdot \!{Id}_d$ for some $\delta\in (0,1/2)$, the distribution defined in \Cref{eq:HS-mix} satisfies $\frac{\delta}{1-\delta}\cdot \!{Id}_d \preceq -\grad^2 \log p_{\mu}(x) \preceq \frac{1}{\delta}\cdot \!{Id}_d$ for any $x\in \bb R^d$.
\end{lemma}
\begin{proof}
    By the definition in \Cref{eq:HS-mix}, 
    \begin{align*}
        p_{\mu}(x) &\propto \sum_{\sigma\in \ab\{\pm 1\}^d} \exp\set{ - \frac{1}{2}x^{\top}J^{-1}x + \tp{J^{-1}h+\sigma}^{\top}x} \\
        &= \exp\set{ - \frac{1}{2}x^{\top}J^{-1}x + h^{\top}J^{-1}x} \cdot \sum_{\sigma\in \ab\{\pm 1\}^d} \exp\set{\sigma^{\top}x} \\
        &= \exp\set{ - \frac{1}{2}x^{\top}J^{-1}x + h^{\top}J^{-1}x} \cdot \prod_{i=1}^d \tp{e^{x(i)}+e^{-x(i)}}.
    \end{align*}
    Therefore, $-\grad \log p_{\mu}(x) = J^{-1}x - J^{-1}h - z_x$ where $z_x\in \bb R^d$ and $z_x(i) = \frac{e^{x(i)} - e^{-x(i)}}{e^{x(i)} + e^{-x(i)}}$ for each $i\in[d]$. Consequently, $-\grad^2 \log p_{\mu}(x) = J^{-1} - A_x$ where $A_x$ is a diagonal matrix in $\bb R^{d\times d}$ and $A_x(i,i) = 1 - \tp{\frac{e^{x(i)} - e^{-x(i)}}{e^{x(i)} + e^{-x(i)}}}^2$ for each $i\in[d]$.

    Since $\delta\cdot \!{Id}_d\preceq J\preceq (1-\delta)\cdot \!{Id}_d$, for any $v\in \bb R^d$,
    \[
        v^{\top}J^{-1}v = \tp{J^{-\frac{1}{2}}v}^{\top}J^{-\frac{1}{2}}v \succeq \frac{1}{1-\delta}v^{\top}\cdot J^{-\frac{1}{2}} J J^{-\frac{1}{2}}\cdot v=\frac{1}{1-\delta}v^{\top}v,
    \]
    and 
    \[
        v^{\top}J^{-1}v = \tp{J^{-\frac{1}{2}}v}^{\top}J^{-\frac{1}{2}}v \preceq \frac{1}{\delta}v^{\top}\cdot J^{-\frac{1}{2}} J J^{-\frac{1}{2}}\cdot v = \frac{1}{\delta}v^{\top}v.
    \]
    % \[
    %     \frac{3}{2}v^{\top}v = \frac{3}{2}v^{\top}\cdot J^{-\frac{1}{2}} J J^{-\frac{1}{2}}\cdot v \preceq \tp{J^{-\frac{1}{2}}v}^{\top}J^{-\frac{1}{2}}v \preceq 3v^{\top}\cdot J^{-\frac{1}{2}} J J^{-\frac{1}{2}}\cdot v = 3v^{\top}v.
    % \]
    Thus $\frac{1}{1-\delta}\cdot \!{Id}_d\preceq J^{-1}\preceq \frac{1}{\delta}\cdot \!{Id}_d$. For the matrix $A_x$, we know $0\preceq A_x\preceq \!{Id}_d$. Combining these together, we can get the desired result.

\end{proof}

% \htodo{Define the notation $x(i)$ and $p*q$.}

For two distributions $\pi$ and $\nu$ with density $p_\pi$ and $p_{\nu}$ respectively, define $\pi*\nu$ as the distribution with density $p_{\pi*\nu}(x)\propto \int_{\bb R^d} p_{\pi}(y)\cdot p_{\nu}(x-y)\dd y$.
When the initial distribution is both strongly log-concave and log-smooth, we can show that the $-\grad^2 \log p_t$ is also bounded via the following lemma and its corollary.
\begin{lemma}[Lemma 28 in \cite{LPSR21}]\label{lem:m-gaussian1}
      Suppose $\pi$ is a probability density function on $\bb R^d$ such that $M_{1}^{-1} \preceq -\grad^2 \log p_{\pi}(x) \preceq M_2^{-1}$ for some $M_1,M_2\in \bb R^{d\times d}$. Let $\nu$ be the density function of $\+N(0,M)$. Then
      \[
        (M_1+M)^{-1} \preceq  -\grad^2 \log p_{\pi*\nu}(x) \preceq (M_2+M)^{-1}.
      \]
\end{lemma}

\begin{corollary}\label{cor:mixture2}
    During the OU process with starting distribution defined in \Cref{eq:HS-mix}, 
    \[
        \frac{1}{1+\frac{1-2\delta}{\delta}e^{-2t}}\cdot \!{Id}_d \preceq -\grad^2 \log p_t(x) \preceq \frac{1}{1-(1-\delta)e^{-2t}} \cdot \!{Id}_d
    \]
    for any $t>0$ and any $x\in \bb R^d$.
\end{corollary}
\begin{proof}
    Recall that $X_t = e^{-t}X_0 + \sqrt{2}\cdot e^{-t}\int_{0}^t e^s \d B_s$. Therefore $\mu_t = \mu'*\nu$ where $\mu'$ is the distribution with density $p_{\mu'}(x) \propto p_{\mu}\tp{\frac{x}{e^{-t}}}$ and $\nu$ is $\+N(0,(1-e^{-2t})\!{Id}_d)$. From \Cref{lem:m-gaussian1}, with $M_1 = \frac{1-\delta}{\delta}\cdot e^{-2t}\!{Id}_d$, $M_2 = \delta\cdot e^{-2t}\!{Id}_d$ and $M=(1-e^{-2t})\!{Id}_d$, we have 
    \[
        \frac{1}{1+\frac{1-2\delta}{\delta}e^{-2t}}\cdot \!{Id}_d \preceq -\grad^2 \log p_t(x) \preceq \frac{1}{1-(1-\delta)e^{-2t}} \cdot \!{Id}_d.
    \]
\end{proof}

Thus, despite being a mixture of Gaussian distributions where the component centers might be far apart, it is still $\+O(1)$-log-smooth and remains $\+O(1)$-log-smooth throughout the OU process.

\begin{remark}
    The motivation for exploring the distributions in \Cref{eq:HS-mix} is to study the Ising model, which is a distribution over $\ab\{\pm 1\}^d$ with density $p_{J,h}(\sigma)\propto \exp\set{\frac{1}{2}\inner{\sigma}{J\sigma} + \inner{h}{\sigma}}$ for any $\sigma\in \ab\{\pm 1\}^d$. 

    The Hubbard-Stratonovich transform states that the Ising model can be reduced to sampling from the distribution in \Cref{eq:HS-mix}: Consider the joint distribution over $\ab\{\pm 1\}^d\times \bb R^d$ with density $p_{J,h}(\sigma,x)\propto \exp\tp{-\frac{1}{2}x^\top J^{-1} x+(J^{-1}h+\sigma)^\top x}$. We can prove that
    \begin{itemize}
        \item its marginal density on $\bb R^d$ is exactly the $p_{\mu}$ in \Cref{eq:HS-mix};
        \item $p_{\mu}(x)\propto \exp\set{ - \frac{1}{2}x^{\top}J^{-1}x + h^{\top}J^{-1}x} \cdot \prod_{i=1}^d \tp{e^{x(i)}+e^{-x(i)}}$ and thus the unnormalized density and the gradients of the potential function can be calculated in polynomial time for any $x$;
        \item the conditional distribution with density $p_{J,h}(\sigma|x) \propto \exp\set{\inner{\sigma}{x}}$ is a product distribution and can be sampled efficiently.
    \end{itemize}
    The proofs are similar to Lemma E.1 in \cite{KLR22}.

    % Its marginal density on $\bb R^d$ is exactly \Cref{eq:HS-mix}. Furthermore, we can prove that the conditional distribution $p_{J,h}(\sigma|x) \propto \exp\set{\inner{\sigma}{x}}$, which is a product distribution on $\ab\{\pm 1\}^d$ and can be efficiently sampled from (the proof is similar to Lemma E.1 in \cite{KLR22}).

    Therefore, sampling from the Ising model can be executed in two steps: 1) sample $X\sim \mu$; 2) sample $\sigma$ from the distribution with density $p_{J,h}(\sigma|X)$. Hence, the hardness of this problem is closely related to the hardness of sampling from mixture of Gaussians. Given \Cref{lem:mixture1}, when $0\prec J\prec \!{Id}_d$, the distribution $\mu$ can be simulated in polynomial time using Langevin-based algorithms (e.g., the algorithm in \cite{CCBJ18}) and thus also gives a polynomial complexity upper bound for the Ising model. On the other hand, \cite{GKK24} proved that for any real $c>1$, the existence of polynomial samplers for Ising model with arbitrary $0\prec J\prec (1+c)\!{Id}_d$ implies $\*{NP}=\*{RP}$. This in turn indicates that, assuming  $\*{NP}\neq\*{RP}$, sampling from the mixture of Gaussians in such a special structure with $0\prec J\prec (1+c)\!{Id}_d$ is generally hard.
\end{remark}
% \htodo{Is the main idea of this remark is clear?}
% \ctodo{Very good}

%% file: optimization.tex
\section{Comparison with optimization}\label{sec:sampling-vs-opt}
In this section, we prove \Cref{thm:main-opt-lb} and compare the hardness of optimization problems and sampling problems.

It is shown in \cite{MCJ+19} that for an $L$-log-smooth distribution $\mu$ which are $m$-strongly log-concave outside a ball of radius $R$, sampling can be done within $\tilde{\+O}\tp{\frac{dL^2}{\eps^2 m^2}\cdot e^{32LR^2}}$ queries, while optimizing the potential function of $\mu$ needs $\Omega\tp{\tp{\frac{LR^2}{\eps}}^{\frac{d}{2}}}$ queries in the worst case. This indicates that when $LR^2 = o(d)$, optimization can be harder than sampling.

% The relationship between optimization and sampling has always been a topic of interest. The work of \cite{MCJ+19} proves that when $f$ is $L$-smooth and  is $m$-strongly convex outside a region of radius $R$, the sampling upper bound is linear in $d$, while the optimization lower bound is exponential in $d$. This suggests that, for this specific class of functions, sampling is harder than optimization.

In this work, we consider a more general case, where the function $f$ and distribution $\mu$ with density $\propto e^{-f}$ are only required to satisfy \Cref{assump:moment} and \ref{assump:smooth}. In the optimization problem, we require the algorithm to output a point $x\in \bb R^d$ such that $\abs{f(x)-f(x^*)}\leq \eps$, where $x^*$ is the minimizer of $f$. 
% Let $\+U$ be the set of $L$-smooth functions such that for any $f\in \+U$, the second moment of $\mu\propto e^{-f}$ is $\Theta(M)$ for some $M=\Omega\tp{\frac{d}{L}}$ and $\grad f(0)=0$. 
%  We will prove the following theorem in \Cref{subsec:opt-lb}.

% show in \Cref{subsec:opt-lb} that, for any algorithm solving the optimization problem and guaranteeing $\+O(1)$ error with constant probability, there exists some functions in $\+U$, which requires at least $K_{o} = e^{\frac{d}{2}\log \Omega\tp{LM})}$ queries. 

% \subsection{Proof of \Cref{thm:main-opt-lb}}\label{subsec:opt-lb}

To prove the complexity lower bound, we first see the hard instances constructed in \cite{MCJ+19}.
\begin{lemma}[Lemma 17 in \cite{MCJ+19}]\label{lem:packing}
    For $R>r>0$, there exists $\+X_{R}\subset \bb R^d$ with $\abs{\+X_R}= \Big\lfloor \tp{\frac{R-r}{2r}}^d \Big\rfloor$ such that $\bigcup_{x\in \+X_R} \+B_r(x)\subset \+B_R$ and $\+B_r(x)\cap \+B_r(y) = \emptyset$ for any $x\neq y\in \+X_R$.
\end{lemma}

The work of \cite{MCJ+19} focused the cases where $f$ is $L$-smooth and is $m$-strongly convex outside a region of radius $R$. Consider the set $\+X_{\frac{R}{2}}$. We index the vertices in $\+X_{\frac{R}{2}}$ by $\left[\abs{\+X_{\frac{R}{2}}}\right]$. For each point $x_i\in \+X_{\frac{R}{2}}$, construct a function $f_i$ as follows:
\begin{equation}
    f_i(x) = \begin{cases}
        \frac{\eps}{2} \cdot \cos\tp{\frac{\pi}{r^2}\tp{\|x-x_i\|^2 - r^2}} - \frac{\eps}{2}, & \|x-x_i\|\leq r\\
        0, & \|x-x_i\|> r, \|x\| \leq \frac{R}{2} \\
        m\tp{\|x\| - \frac{R}{2}}^{2}, & \|x\|>\frac{R}{2}
    \end{cases}, \label{eq:opt1}
\end{equation}
where $r\defeq \sqrt{\frac{(2\pi^2 + \pi) \eps}{L}}$. 
From Lemma 18 in \cite{MCJ+19}, the functions defined in \Cref{eq:opt1} are all $L$-smooth. 
\begin{lemma}[Lemma 18 in \cite{MCJ+19}]\label{lem:opt-lb3}
    Let $L\geq 2m$. The functions in \Cref{eq:opt1} are $L$-smooth.
\end{lemma}

In Appendix C of \cite{MCJ+19}, they prove the following result.
\begin{theorem}\label{thm:lb-opt}
    For any $R>0$, $L\geq 2m>0$ and $\eps<\+O(LR^2)$, for any algorithm in $\+A$, there exists $i\in \left[\abs{\+X_{\frac{R}{2}}}\right]$, such that the algorithm requires at least $K=\Omega\tp{\tp{\frac{LR^2}{\eps}}^{\frac{d}{2}}}$ iterations on $f_i$ to guarantee that $\min_{k\leq K} \abs{f_i(x^K) - f_i(x^*)}<\eps$ with constant probability.
\end{theorem}

For each $x_i\in \+X_{\frac{R}{2}}$, consider a distribution $\mu_i$ whose density is in proportion to $e^{-f_i}$ (it is easy to see that this function is integrable). The following lemma shows that the second moment of this distribution is bounded.

\begin{lemma}\label{lem:lb-moment}
    If $LR^2\geq 6d$, $d\geq 8$ and $\eps<1$, choosing $m=\frac{L}{2}$, we have $\frac{R^2}{18(e+1)} \leq \E[X\sim \mu_i]{\|X\|^2} \leq 5R^2$.
\end{lemma}
\begin{proof}
    We first prove the upper bound for $\E[X\sim \mu_i]{\|X\|^2}$. Let $Z= \int_{x\in \bb R^d} e^{-f_i(x)} \dd x$. Let $\+B_{\frac{R}{2}}$ denote the ball with radius $R/2$ centered at $0$. Then 
    \begin{equation}
        Z\geq \int_{\|x\|\leq \frac{R}{2}} e^{-f_i(x)}\dd x \geq \int_{\|x\|\leq \frac{R}{2}} 1 \dd x = \!{vol}\tp{\+B_{\frac{R}{2}}} = \frac{\pi^{\frac{d}{2}}\cdot \tp{\frac{R}{2}}^{d}}{\Gamma\tp{\frac{d}{2}+1}} \geq \tp{\frac{e\pi\cdot R^2}{4d}}^{\frac{d}{2}}. \label{eq:lb-Z}
    \end{equation}
    By definition,
    \begin{align*}
        \E[X\sim \mu_i]{\|X\|^2} &= \int_{\+B_{\frac{R}{2}}} \|x\|^2\cdot \frac{e^{-f_i(x)}}{Z} \dd x + \int_{\bb R^d \setminus \+B_{\frac{R}{2}}} \|x\|^2\cdot \frac{e^{-f_i(x)}}{Z} \dd x \\
        &\leq 4R^2 \int_{\+B_{\frac{R}{2}}} \frac{e^{-f_i(x)}}{Z} \dd x + \frac{1}{Z}\cdot \int_{\bb R^d \setminus \+B_{\frac{R}{2}}} \|x\|^2\cdot e^{-f_i(x)} \dd x \\
        &\leq 4R^2 + \frac{1}{Z}\cdot \int_{\bb R^d \setminus \+B_{\frac{R}{2}}} \|x\|^2\cdot e^{-f_i(x)} \dd x.
    \end{align*}
    For the second term, we have
    \begin{align*}
        &\phantom{{}={}}\int_{\bb R^d \setminus \+B_{\frac{R}{2}}} \|x\|^2\cdot e^{-f_i(x)} \dd x \\
        &= \int_{\bb R^d \setminus \+B_{\frac{R}{2}}} \|x\|^2\cdot e^{-\frac{L}{2}\tp{\|x\|^2 + \frac{R^2}{4} - R\|x\|}} \dd x \\
        &\leq \int_{\bb R^d \setminus \+B_{\frac{R}{2}}} \|x\|^2\cdot e^{-\frac{L}{2}\tp{\|x\|^2 + \frac{R^2}{4} - \frac{\|x\|^2}{2}}} \dd x \\
        &= \frac{1}{2}e^{-\frac{LR^2}{8} + \frac{d}{2}\log \frac{4\pi}{L}}\cdot  \int_{\bb R^d \setminus \+B_{\frac{R}{2}}} \|x\|^2\cdot 2e^{-\frac{L\|x\|^2}{4} - \frac{d}{2}\log \frac{4\pi}{L}} \dd x\\
        &\leq \frac{1}{2}e^{-\frac{LR^2}{8} + \frac{d}{2}\log \frac{4\pi}{L}}\cdot \E[X\sim \+N(0, \frac{2\!{Id}_d}{L})]{\|X\|^2}\\
        &= \frac{d}{L}\cdot e^{-\frac{LR^2}{8} + \frac{d}{2}\log \frac{4\pi}{L}}.
    \end{align*}
    Combining \Cref{eq:lb-Z} and the fact that $LR^2\geq 6d$, we have
    \begin{align*}
        \E[X\sim \mu_i]{\|X\|^2} &\leq 4R^2 + \frac{1}{Z}\cdot \int_{\bb R^d \setminus \+B_{\frac{R}{2}}} \|x\|^2\cdot e^{-f(x)} \dd x \\
        &\leq 4R^2 + \frac{d}{L}\cdot e^{-\frac{LR^2}{8} + \frac{d}{2}\log \frac{4\pi}{L} - \frac{d}{2}\log \frac{e\pi\cdot R^2}{4d}} \\
        &< 4R^2 + \frac{d}{L} < 5 R^2.
    \end{align*}
    \bigskip

    Then we prove the lower bound for $\E[X\sim \mu_i]{\|X\|^2}$. Consider a distribution $\mu'$ supported only on $\+B_{\frac{R}{2}}$, whose density is proportional to $e^{-f_i(x)}$ on $\+B_{\frac{R}{2}}$.
    Since for each $x\in \bb R^d \setminus \+B_{\frac{R}{2}}$, $\|x\|^2 \geq \frac{R^2}{4} > \frac{R^2}{18(e+1)}$, it suffices to prove $\E[X\sim \mu']{\|X\|^2}\geq \frac{R^2}{18(e+1)}$. Let $Z' = \int_{\+B_{\frac{R}{2}}} e^{-f_i(x)}\dd x$. Then
    \[
        Z' \leq \!{vol}(\+B_{\frac{R}{2}}) + e^{\eps}\cdot \!{vol}(\+B(x_i,r)) \leq (e+1)\cdot \!{vol}(\+B_{\frac{R}{2}}).
    \]
    Therefore,
    \begin{align*}
        \E[X\sim \mu']{\|X\|^2} & = \frac{1}{Z'} \int_{\+B_{\frac{R}{2}}} \|x\|^2 e^{-f_i(x)}\dd x \\
        &\geq \frac{1}{Z'} \int_{\+B_{\frac{R}{2}}} \|x\|^2 \dd x \\
        &\geq \frac{\frac{R^2}{9}\cdot \tp{\!{vol}(\+B_{\frac{R}{2}}) - \!{vol}(\+B(0,R/3))}}{(e+1)\cdot \!{vol}(\+B_{\frac{R}{2}})} \\
        &\geq \frac{R^2}{18(e+1)}.
    \end{align*}
\end{proof}

Then \Cref{thm:main-opt-lb} is a direct corollary of \Cref{thm:lb-opt}, \Cref{lem:opt-lb3,lem:lb-moment} by choosing $R=\Theta(\sqrt{M})$. 

Note that from \Cref{thm:main-ub}, our sampling algorithm needs $\tp{\frac{LM}{d}}^{\+O(d)}$ queries to simulate the distribution $\mu \propto e^{-f}$ within $0.01$ error in total variation  distance. For $LR^2=LM=\Theta(d)$ and for sufficiently large $d$, the lower bound in \Cref{thm:main-opt-lb} is larger by a factor of $d^{\Theta(d)}$ than this upper bound for sampling. 
Compare to the results in \cite{MCJ+19}, our results further demonstrate that for a broader class of distributions or functions and a wider range of parameters, sampling is simpler than optimization.

% \htodo{Here we require $\grad f(0)=0$ in the sampling upper bound but not in optimization lower bound. But I think this is not a big deal because \cite{MCJ+19} also has this problem.}

\begin{remark}
    If we delve into the proofs of the lower bounds for sampling and optimization, we can gain some intuition about why optimization is harder. Notably, in both proofs, the hard instances are constructed by locally modifying the function values or density values on a base instance. The challenge for the algorithm is to recognize the modified local area, such as the ball $\+B_r(x_i)$ in \Cref{eq:opt1} and $\+B_{r_2}(v)$ in \Cref{sec:hardinstance}. The more balls we can pack, the harder it becomes for the algorithm to identify the modified ball.

    On the other hand, to ensure $L$-smoothness, a smooth transition is needed at the border of the modified local area, meaning the radius of the ball cannot be too small. Taking the error $\eps$ as a constant for example, we observe the following:
    \begin{itemize}
        \item In the optimization lower bound, the difference between the function value inside and outside the ball should be $\Omega(1)$. As a consequence, the radius should be $\Omega\tp{\sqrt{\frac{1}{L}}}$.
        \item In the sampling lower bound, the mass inside the ball must be $\Omega(1)$. This requires a larger difference between function values inside and outside the ball, which in turn require the radius to be $\Omega\tp{\sqrt{\frac{d}{L}}}$.
    \end{itemize}
    Thus, when the error requirements are the same, the hard instances in the sampling lower bound is easier for the algorithm to recognize. This provides an intuition for why optimization is harder than sampling.
\end{remark}

%% file: appendix.tex
\section{Auxiliary lemmas}

\subsection{The volume of $d$-balls}

% \ctodo{Can we find a reference for these lemmas?}

\begin{proposition}\label{prop:Gamma}
    For positive integer $d\geq 8$, $\tp{\frac{d}{2e}}^{\frac{d}{2}} \leq \Gamma\tp{\frac{d}{2}+1} \leq \tp{\frac{d}{e}}^{\frac{d}{2}}$.
\end{proposition}
\begin{proof}
    When $d$ is even, $\Gamma\tp{\frac{d}{2}+1} = \tp{\frac{d}{2}}!$. From Stirling's approximation,
    \[
        \tp{\frac{d}{2e}}^{\frac{d}{2}} \leq \sqrt{\pi d} \cdot \tp{\frac{d}{2e}}^{\frac{d}{2}} \leq \tp{\frac{d}{2}}! \leq 2\sqrt{\pi d} \cdot \tp{\frac{d}{2e}}^{\frac{d}{2}} \leq \tp{\frac{d}{e}}^{\frac{d}{2}}.
    \]
    % \mn{The Stirling's approximation states that for any positive integer $m$,
    % \[
    %     m!\leq \sqrt{2\pi m}\cdot \tp{\frac{m}{e}}^m \cdot e^{\frac{1}{12m}},
    % \] and
    % \[
    %     m!\geq \sqrt{2\pi m}\cdot \tp{\frac{m}{e}}^m \cdot e^{\frac{1}{12m+1}}.
    % \]}
    When $d$ is odd, $\Gamma\tp{\frac{d}{2}+1} = \frac{\sqrt{\pi}}{2^d}\cdot \frac{d!}{\tp{\frac{d-1}{2}}!}$. Similarly we have
    \[
        \frac{\sqrt{\pi}}{2^d}\cdot \frac{d!}{\tp{\frac{d-1}{2}}!} \leq \frac{\sqrt{\pi}}{2^d}\cdot \frac{\sqrt{2\pi d} \tp{\frac{d}{e}}^d }{ \sqrt{\pi(d-1)}\tp{\frac{d-1}{2e}}^{\frac{d-1}{2}}} \cdot e^{\frac{1}{12d} - \frac{1}{6(d-1)+1}} \leq \sqrt{ \frac{2\pi d}{d-1}}\cdot \tp{\frac{d-1}{d}}^{\frac{d-1}{2}} \tp{\frac{d}{2e}}^{\frac{d+1}{2}} \leq \tp{\frac{d}{e}}^{\frac{d}{2}},
    \]
    and 
    \[
        \frac{\sqrt{\pi}}{2^d}\cdot \frac{d!}{\tp{\frac{d-1}{2}}!} \geq \frac{\sqrt{\pi}}{2^d}\cdot \frac{\sqrt{2\pi d} \tp{\frac{d}{e}}^d }{ \sqrt{\pi(d-1)}\tp{\frac{d-1}{2e}}^{\frac{d-1}{2}}} \cdot e^{\frac{1}{12d+1} - \frac{1}{6(d-1)}} \geq \sqrt{2\pi}\cdot \tp{\frac{d}{2e}}^{\frac{d+1}{2}} \cdot \frac{1}{2} \geq \tp{\frac{d}{2e}}^{\frac{d}{2}}.
    \]
\end{proof}

\begin{proposition}
    The volume of $\+B_R$ is $\frac{(\pi R^2)^\frac{d}{2}}{\Gamma\tp{\frac{d}{2}+1}}$.
\end{proposition}

\begin{corollary}\label{cor:dballvolbound}
    \[
        \tp{\frac{e\pi R^2}{d}}^{\frac{d}{2}}\le\!{vol}\tp{\+B_R} \le \tp{\frac{2e\pi R^2}{d}}^{\frac{d}{2}}.
    \]
\end{corollary}

% \ctodo{Clean here}

\subsection{The trigonometric functions in \Cref{lem:disjointcap}}

For two vectors $u,v\in \bb R^d$, let $\theta\tp{u,v}$ denote the angle between $u$ and $v$. The cosine of this angle is defined as $\cos\tp{\theta\tp{u,v}} = \frac{\inner{u}{v}}{\|u\|\|v\|}$. In this section, we prove the results about trigonometric functions used in \Cref{lem:disjointcap}.

\begin{lemma}\label{lem:cos}
    For any unit vectors $v,x,w\in \bb R^d$, we have 
    \[
        \cos\tp{\theta(v,w)}= \cos\tp{\theta(v,x)}\cdot \cos\tp{\theta(x,w)} +\sin\tp{\theta(v,x)}\cdot \sin\tp{\theta(x,w)}.
    \]
\end{lemma}
\begin{proof}
    We divide $v$ into two parts, $v_{\|}$ and $v_{\perp}$, which are parallel and perpendicular to $x$ respectively. Similarly, we divide $w$ into $w_{\|}$ and $w_{\perp}$. Therefore,
    \begin{align*}
        \cos\tp{\theta(v,w)} &= \inner{v}{w}= \inner{v_\| + v_\perp}{w_\| + w_\perp}\\
        &= \inner{v_\|}{w_\|} + \inner{v_\perp}{w_\perp}\\
        &=\cos\tp{\theta(v,x)}\cdot \cos\tp{\theta(x,w)} + \sin\tp{\theta(v,x)}\cdot \sin\tp{\theta(x,w)}
    \end{align*}
\end{proof}

\begin{lemma}\label{lem:cosinBall}
    Let $u,v\in \bb R^d$ be a vector with norm $\|u\|=\|v\|=\frac{3R}{4}$. Let $\ell=\frac{\sqrt{\tp{\frac{3R}{4}}^2-2r_2^2}}{\frac{3R}{4}}$. If the caps $C_u=\set{x\in \bb R: \|x\|=\frac{3R}{4}, \cos\tp{\theta(x,u)}\geq \ell}$ and $C_v=\set{x\in \bb R: \|x\|=\frac{3R}{4}, \cos\tp{\theta(x,v)}\geq \ell}$ are disjoint. Then the balls $\+B(v,r_2)$ and $\+B(u,r_2)$ are also disjoint.
\end{lemma}
\begin{proof}
    Assume in contrast that there exists $x\in \+B(u,r_2)\cap \+B(v,r_2)$. Let $y=\frac{3R}{4}\cdot \frac{x}{\|x\|}$. We show that $y\in C_u\cap C_v$.

    Note that since $x\in \+B(u,r_2)$, 
    \[
        r_2^2\geq \|x-u\|^2=\|x\|^2 + \|u\|^2 - 2\inner{x}{u} = \|x\|^2 + \|u\|^2 - 2\|x\|\|u\|\cdot \cos\tp{\theta(x,u)}.
    \]
    So we have
    \begin{align*}
        \cos\tp{\theta(x,u)}&\geq \frac{\|x\|^2+\|u\|^2}{2\|x\|\|u\|} - \frac{r_2^2}{2\|x\|\|u\|} \\
        &\geq 1 - \frac{r_2^2}{2\|x\|\|u\|} \\
        \mr{$\|x\|\geq \|u\|-\|x-u\|\geq \frac{3R}{4}-r_2$} &\geq 1 - \frac{r_2^2}{2\tp{\frac{3R}{4}-r_2}\cdot \frac{3R}{4}} \\
        \mr{$R\geq 4r_2$}&\geq 1-\frac{r_2^2}{\tp{\frac{3R}{4}}^2}\geq \ell.
    \end{align*}
    Similarly, $\cos\tp{\theta(x,v)}\geq \ell$. This indicates $y\in C_u\cap C_v$, which leads to a conflict.
\end{proof}

\subsection{The properties of the smoothness parameter and the second moment}
The following lemma shows that scaling the domain does not change $L\cdot M$.
\begin{lemma}\label{lem:LM}
    For an $L$-smooth function $f\colon \bb R^d\to \bb R$, assume the second moment of the distribution $\mu$ with density $p_\mu(x)\propto e^{-f(x)}$ is $M$. Consider a scaled function $f'(x)=f\tp{\frac{x}{\sqrt{L}}}$. Then $f'$ is $1$-smooth and the second moment of $\nu$ with density $\propto e^{-f'(x)}$ is $L\cdot M$.
\end{lemma}
\begin{proof}
    From the chain rule, $\grad^2 f'(x) = \frac{1}{L} \grad^2 f(y)\Big|_{y=\frac{x}{\sqrt{L}}}$. This indicates that the function $f'$ is $1$-smooth. 

    For the second moment, we have
    \begin{align*}
        \E[X\sim \nu]{\|X\|^2} &= \frac{\int_{\bb R^d} \|x\|^2 e^{-f'(x)}\d x}{\int_{\bb R^d} e^{-f'(x)}\d x}\\
        \mr{substituting $x$ by $\sqrt{L}y$}&= \frac{\int_{\bb R^d} \|\sqrt{L}y\|^2 e^{-f'(\sqrt{L}y)}\cdot \tp{\sqrt{L}}^d \d y}{\int_{\bb R^d} e^{-f'(\sqrt{L}y)}\cdot \tp{\sqrt{L}}^d\d y}\\
        &=L\cdot \frac{\int_{\bb R^d} \|y\|^2 e^{-f(y)}\d y}{\int_{\bb R^d} e^{-f(y)}\d y}\\
        &=L\cdot \E[Y\sim \mu]{\|Y\|^2} \\
        &=L\cdot M.
    \end{align*}
\end{proof}

The following lemma gives a lower bound for $L\cdot M$ for any distribution $\mu$ with density $\propto e^{-f}$ such that $\grad f(0)=0$.
% The following lemma shows that for any distribution $\nu\propto e^{-f}$ satisfying \Cref{assump:moment} and \ref{assump:smooth} with $\grad f(0)=0$, we have $LM\geq d$.
\begin{lemma}\label{lem:lb-LM}
    Let $f:\bb R^d\to \bb R$ be an $L$-smooth function such that $e^{-f}$ is integrable and $\grad f(0)=0$. Let $\mu$ be the distribution with density $p_{\mu}\propto e^{-f}$. Then $\E[\mu]{\|X\|^2}\geq \frac{d}{L}$.
\end{lemma}
\begin{proof}
    By the definition of $L$-smoothness, for each $x\in \bb R^d$, 
    \[
        L\|x-0\| \geq \|\grad f(x) - \grad f(0)\| = \|\grad \log \mu(x)\| = \frac{\|\grad \mu(x)\|}{\mu(x)}.
    \]
    Therefore,
    \begin{align*}
        \E[\mu]{\|X\|^2} &= \int_{\bb R^d} \|x\|^2 p_\mu(x) \dd x\\
        &\geq \int_{\bb R^d} \|x\|\cdot \frac{1}{L}\cdot \frac{\|\grad p_\mu(x)\|}{p_\mu(x)}  \cdot p_\mu(x) \dd x \\
        &= \frac{1}{L} \int_{\bb R^d} \|x\|\cdot \|\grad p_\mu(x)\|\dd x  \\
        \mr{Cauchy-Schwartz inequality}&\geq -\frac{1}{L} \int_{\bb R^d} \sum_{i=1}^d x_i\cdot \grad p_\mu(x)_i\dd x \\
        &= -\frac{1}{L}\sum_{i=1}^d \int_{\bb R} \int_{\bb R}\cdots \int_{\bb R} x_i\cdot \grad p_\mu(x)_i\dd x_i \dd x_1\cdots \dd x_d \\
        \mr{integration by parts}&= \frac{1}{L}\sum_{i=1}^d \int_{\bb R} \int_{\bb R}\cdots \int_{\bb R}  p_\mu(x)\dd x_i \dd x_1\cdots \dd x_d \\
        &= \frac{d}{L}.
    \end{align*}
\end{proof}